%% file: main.tex
\documentclass[acmsmall,screen]{acmartFAKE}
\AtBeginDocument{%
  \providecommand\BibTeX{{%
    \normalfont B\kern-0.5em{\scshape i\kern-0.25em b}\kern-0.8em\TeX}}}

\setcopyright{acmcopyright}
\copyrightyear{2021}
\acmYear{2021}
\acmDOI{10.1145/1122445.1122456}

\input{defs}
\begin{document}

\title{Simple Uncoupled No-Regret Learning Dynamics for Extensive-Form Correlated Equilibrium}\thanks{A short version of this article appeared in \emph{Advances in Neural Information Processing Systems 33: Annual Conference on Neural Information Processing Systems 2020, NeurIPS 2020}~\citep{Celli20:NoRegret}.}

\author{Gabriele Farina}
\email{gfarina@cs.cmu.edu}
\affiliation{%
  \institution{Carnegie Mellon University}
  \streetaddress{5000 Forbes Avenue}
  \city{Pittsburgh}
  \state{Pennsylvania}
  \country{USA}
}

\author{Andrea Celli}
\email{andrea.celli@polimi.it}
\affiliation{%
 \institution{Politecnico di Milano}
 \streetaddress{Piazza Leonardo da Vinci, 32}
 \city{Milan}
 \country{Italy}}

\author{Alberto Marchesi}
\email{alberto.marchesi@polimi.it}
\affiliation{%
 \institution{Politecnico di Milano}
 \streetaddress{Piazza Leonardo da Vinci, 32}
 \city{Milan}
 \country{Italy}}

\author{Nicola Gatti}
\email{nicola.gatti@polimi.it}
\affiliation{%
 \institution{Politecnico di Milano}
 \streetaddress{Piazza Leonardo da Vinci, 32}
 \city{Milan}
 \country{Italy}}


\begin{abstract}
    The existence of simple uncoupled no-regret learning dynamics that converge to correlated equilibria in normal-form games is a celebrated result in the theory of multi-agent systems. Specifically, it has been known for more than 20 years that when all players seek to minimize their \emph{internal} regret in a repeated normal-form game, the empirical frequency of play converges to a normal-form correlated equilibrium. Extensive-form (that is, tree-form) games generalize normal-form games by modeling both sequential and simultaneous moves, as well as imperfect information.
    Because of the sequential nature and presence of private information in the game, correlation in extensive-form games possesses significantly different properties than its counterpart in normal-form games, many of which are still open research directions.
    Extensive-form correlated equilibrium (EFCE) has been proposed as the natural extensive-form counterpart to the classical notion of correlated equilibrium in normal-form games. Compared to the latter, the constraints that define the set of EFCEs are significantly more complex, as the correlation device ({\em a.k.a.} mediator) must keep into account the evolution of beliefs of each player as they make observations throughout the game. Due to that significant added complexity, the existence of uncoupled learning dynamics leading to an EFCE has remained a challenging open research question for a long time.
    In this article, we settle that question by giving the first uncoupled no-regret dynamics that converge to the set of EFCEs in $n$-player general-sum extensive-form games with perfect recall. We show that each iterate can be computed in time polynomial in the size of the game tree, and that, when all players play repeatedly according to our learning dynamics, the empirical frequency of play is proven to be a $O(1/\sqrt{T})$-approximate EFCE with high probability after $T$ game repetitions, and an EFCE almost surely in the limit.
\end{abstract}


\maketitle

\input{text/introduction}
\input{text/efgs}
\input{text/efce}

\input{text/phi_regret}
\input{text/trigger_regret}
\input{text/algorithm}

\begin{acks}
    We thank Dustin Morrill, Marc Lanctot, and Mike Bowling for a useful discussion about behavioral deviation functions, and for pointing out an incorrect statement related to their recent framework~\citep{morrill2021efficient} in a preliminary version of this article. We are also grateful to the anonymous reviewers at NeurIPS 2020, where a preliminary version of this article appeared, for their useful comments.

  This work is based on work supported by the National Science Foundation under grants IIS-1718457,
IIS-1617590, IIS-1901403, and CCF-1733556, the ARO under awards W911NF-17-1-0082 and W911NF2010081, and the  Italian MIUR PRIN 2017 Project ALGADIMAR
``Algorithms, Games, and Digital Market''.  Gabriele Farina is supported by a Facebook fellowship.
\end{acks}

\bibliographystyle{ACM-Reference-Format}
\bibliography{ref}

\end{document}

%% file: defs.tex
\usepackage{bm}
\usepackage{amsmath}
\usepackage{xparse}
\usepackage{mleftright}
\usepackage{dsfont}
\usepackage[capitalise,noabbrev]{cleveref}
\usepackage[mathscr]{eucal}
\usepackage{graphicx}
\usepackage{nicefrac}

\usepackage[normalem]{ulem}

\newtheorem{problem}{Problem}
\crefname{problem}{Problem}{Problems}
\newtheorem{assumption}{Assumption}
\crefname{assumption}{Assumption}{Assumptions}

\usepackage{tikz}
\usetikzlibrary{calc,arrows.meta,arrows}
\tikzset{cross/.style={path picture={
  \draw[black]
(path picture bounding box.south east) -- (path picture bounding box.north west) (path picture bounding box.south west) -- (path picture bounding box.north east);
}}}
\tikzstyle{chanode}=[fill=white,draw=black,circle,cross,inner sep=.8mm]
\tikzstyle{pl1node}=[fill=black,draw=black,circle,inner sep=.5mm]
\tikzstyle{pl2node}=[fill=white,draw=black,circle,inner sep=.55mm]
\tikzstyle{termina}=[fill=white,draw=black,inner sep=.6mm]
\tikzstyle{decpt}  =[fill=black,draw=black,inner sep=.8mm]
\tikzstyle{obspt}  =[fill=white,draw=black,cross,inner sep=0.95mm]

\newcommand{\bull}{\hspace*{2mm}$\bullet$\hspace*{1mm}}

\usetikzlibrary{matrix,backgrounds,positioning}
\pgfdeclarelayer{back}
\pgfsetlayers{back,background,main}
\NewDocumentCommand{\fhighlight}{O{black!15}mm}{%
\fill[#1] (#2.north west) rectangle (#3.south east);
}

\renewcommand{\vec}[1]{\bm{#1}}
\newcommand{\defeq}{\coloneqq}
\newcommand{\emptyseq}{\varnothing}
\newcommand{\Hist}{\mathscr{H}}
\newcommand{\Z}{\mathscr{Z}}
\newcommand{\A}{\mathscr{A}}
\newcommand{\cR}{\mathcal{R}}
\newcommand{\cX}{\mathcal{X}}
\newcommand{\cQ}{\mathcal{Q}}

\DeclareMathOperator{\co}{co}
\newcommand{\bbP}{\mathbb{P}}
\newcommand{\bbE}{\mathbb{E}}
\newcommand{\bbR}{\mathbb{R}}
\newcommand{\bbS}{\mathbb{S}}
\newcommand{\bbN}{\mathbb{N}_{>0}}

\NewDocumentCommand{\pure}{O{i}}{%
\ifthenelse{\equal{#1}{}}%
    {\vec{\pi}}%
    {\vec{\pi}^{(#1)}}%
}
\NewDocumentCommand{\hatpure}{O{i}}{%
\ifthenelse{\equal{#1}{}}%
    {\hat{\vec{\pi}}}%
    {\hat{\vec{\pi}}^{(#1)}}%
}
\NewDocumentCommand{\puret}{O{i}O{t}}{\vec{\pi}^{(#1),\,#2}}
\NewDocumentCommand{\Pure}{O{i}}{\Pi^{(#1)}}
\NewDocumentCommand{\tdev}{O{i}O{\hat{\sigma}}O{\hat{\vec{\pi}}}}{%
    \phi^{(#1)}_{#2\to#3}
}
\NewDocumentCommand{\Mdev}{O{i}O{\hat{\sigma}}O{\hat{\vec{\pi}}}}{%
    \vec{M}^{(#1)}_{#2\to#3}
}
\NewDocumentCommand{\cJ}{O{i}}{\mathscr{J}^{(#1)}}
\NewDocumentCommand{\Seqs}{O{i}}{\Sigma^{(#1)}}
\NewDocumentCommand{\pc}{O{z}}{p^{(c)}(#1)}
\NewDocumentCommand{\ut}{O{i}}{u^{(#1)}}
\NewDocumentCommand{\bbone}{O{XXX}}{\mathds{1}_{\{#1\}}}
\NewDocumentCommand{\Ph}{O{i}}{\Psi^{(#1)}}
\NewDocumentCommand{\Tph}{O{i}}{(\co\Ph)}
\NewDocumentCommand{\Phj}{O{i}O{\hat\sigma}}{\bar{\Psi}^{(#1)}_{#2}}

\newcommand{\alp}{\alpha_z^{(i),\,t}}

\NewDocumentCommand{\Q}{O{i}}{\mathcal{Q}^{(#1)}}
\NewDocumentCommand{\q}{}{\vec{q}}
\NewDocumentCommand{\qt}{O{i}O{t}}{\vec{q}^{(#1),\,#2}}
\NewDocumentCommand{\seq}{m}{\texttt{#1}}

\newcommand{\numberthis}[1]{\refstepcounter{equation}\tag{\theequation}\label{#1}}

\usepackage[linesnumbered,ruled,lined,noend]{algorithm2e}
\makeatletter
    \patchcmd\algocf@Vline{\vrule}{\vrule \kern-0.4pt}{}{}
    \patchcmd\algocf@Vsline{\vrule}{\vrule \kern-0.4pt}{}{}
\makeatother
\SetKwComment{Hline}{}{\vspace{-3mm}\textcolor{gray}{\hrule}\vspace{1mm}}
\definecolor{darkgrey}{gray}{0.3}
\definecolor{commentcolor}{gray}{0.5}
\SetKwComment{Comment}{\color{commentcolor}[$\triangleright$\ }{}
\SetCommentSty{}
\SetNlSty{}{\color{darkgrey}}{}
\SetKwProg{Fn}{function}{}{}
\SetKwProg{Subr}{subroutine}{}{}
\crefalias{AlgoLine}{line}%
\crefname{algocf}{Algorithm}{Algorithms}

\makeatletter
\let\cref@old@stepcounter\stepcounter
\def\stepcounter#1{%
  \cref@old@stepcounter{#1}%
  \cref@constructprefix{#1}{\cref@result}%
  \@ifundefined{cref@#1@alias}%
    {\def\@tempa{#1}}%
    {\def\@tempa{\csname cref@#1@alias\endcsname}}%
  \protected@edef\cref@currentlabel{%
    [\@tempa][\arabic{#1}][\cref@result]%
    \csname p@#1\endcsname\csname the#1\endcsname}}
\makeatother

\newcommand{\linkfootnote}[1]{%
    \protect\hyperref[#1]{%
        \hbox{\textsuperscript{\normalfont\ref{#1}}}%
    }%
}

%% file: text/introduction.tex
\section{Introduction}

The {\em Nash equilibrium} (NE)~\citep{nash1950equilibrium} is the most common notion of rationality in game theory, and its computation in two-player zero-sum games has been the flagship computational challenge in the area at the interplay between computer science and game theory (see, {\em e.g.}, the landmark results in heads-up no-limit poker by~\citet{brown2017superhuman} and~\citet{moravvcik2017stack}).
The assumption underpinning NE is that the interaction among players is fully {\em decentralized}.
Therefore, an NE is an element of the {\em uncorrelated} strategy space of the game, that is, a product of independent probability distributions over actions, one per player.
A competing notion of rationality is the {\em correlated equilibrium} (CE) proposed by~\citet{aumann1974subjectivity}.
A {\em correlated strategy} is an arbitrary probability distribution over joint action profiles---defining an action for each player---and it is customarily modeled via a trusted external {\em mediator} that draws an action profile from this distribution, and privately recommends to each player their component.
A correlated strategy is a CE if no player has an incentive to choose an action different from the mediator's recommendation, because, assuming that all other players follow their recommended action, the suggested action is the best in expectation.

Many real-world strategic interactions involve more than two players with arbitrary ({\em i.e.}, general-sum) utilities.
In those settings, the CE is an appealing solution concept, as it overcomes several weaknesses of the NE.
First, the NE is prone to equilibrium selection issues, raising the question on how players can select an equilibrium while they are assumed not to be able to communicate with each other.
Second, computing an NE is computationally intractable, being \textsf{PPAD}-complete even in two-player games~\citep{chen2006settling,daskalakis2009complexity}, whereas a CE can be computed in polynomial time.\footnote{In normal-form games, a CE can be computed in polynomial time via linear programming. In extensive-form games, the computational complexity of computing a CE depends on the specific notion of correlation that is adopted. As discussed in more detail in the following, the problem can be solved in polynomial time for the notion studied in this article.}
Third, the social welfare that can be attained by an NE may be arbitrarily lower than what can be achieved via a CE~\citep{koutsoupias1999worst,roughgarden2002bad,celli2018}.
Lastly, in normal-form (that is, simultaneous-move) games, the notion of CE arises from simple \emph{uncoupled} learning dynamics even in general-sum settings with an arbitrary number of players. In words, these learning dynamics are such that each player adjusts their strategy on the basis of their own payoff function, and on other players' strategies, but not on the payoff functions of other players. The existence of uncoupled dynamics enables to overcome the---often unreasonable---assumption that players have perfect knowledge of other players' payoff functions, while at the same time offering a parallel, scalable avenue for finding equilibria. In contrast, in the case of the NE, uncoupled learning dynamics are only known in the two-player zero-sum setting~\citep{hart2000simple,Hart03:Uncoupled,Cesa-Bianchi06:Prediction}.
All of the above considerations contribute to the idea that CE is often a better prescriptive solution concept than NE in general-sum and multiplayer settings.

%
%

{\em Extensive-form correlated equilibrium} (EFCE), introduced by~\citet{von2008extensive}, is a natural extension of the correlated equilibrium to the case of extensive-form (that is, tree-form, sequential) games.
Extensive-form games generalize normal-form games by modeling both sequential and simultaneous moves, as well as imperfect information. 
In an EFCE, the mediator draws, before the beginning of the sequential interaction, a recommended action for each of the possible decision points (that is, {\em information sets}) that players may encounter in the game, but these recommendations are not immediately revealed to each player.
Instead, the mediator incrementally reveals relevant individual moves as players reach new information sets.
At any decision point, the acting player is free to deviate from the recommended action, but doing so comes at the cost of future recommendations, which are no longer issued to that player if they deviate.
It is up to the mediator to make sure that the recommended behavior is indeed an equilibrium---that is, that no player would be better off ever deviating from following the mediator's recommendations at each information set.
Compared to the constraints that characterize the set of CEs in normal-form games, those that define the set of EFCEs in extensive-form games are significantly more complex. Indeed, the main challenge of the EFCE case is that the mediator must keep into account the evolution of beliefs of each player as they make observations throughout the game tree.

%
In general-sum extensive-form games with an arbitrary number of players (including potentially the {\em chance player} modeling exogenous stochastic events), the problem of computing a feasible EFCE can be solved in polynomial time in the size of the game tree~\citep{huang2008computing} via a variation of the {\em Ellipsoid Against Hope} algorithm~\citep{papadimitriou2008,jiang2015polynomial}.
\citet{dudik2009sampling} provide an alternative sampling-based algorithm to compute EFCEs. However, their algorithm is centralized and based on MCMC sampling, which limits its applicability on large-scale problems. 
In practice, these approaches cannot scale beyond toy problems. On the other hand, methods based on uncoupled learning dynamics usually work quite well in large real-world problems, while retaining the nice properties of uncoupled dynamics that we discussed above.
%
%
The following fundamental research question remains open: {\em is it possible to devise simple uncoupled learning dynamics that converge to an EFCE?} 

We show that the answer is positive. 
To do so, we introduce the idea of \emph{trigger regret}, which builds on an equivalent formulation of EFCE based on the notion of {\em trigger agent}, introduced by~\citet{Gordon08:No}.
This is a notion of {internal regret} suitable for extensive-form games that naturally expresses the regret incurred by trigger agents in the definition of EFCE.
Specifically, trigger regret is a particular instantiation of the framework known as {\em phi-regret minimization} introduced by~\citet{Stoltz07:Learning}, and building on previous work by~\citet{Greenwald03:General}.
In general, phi-regret minimization works with a notion of regret defined with respect to a given set of linear transformations on the decision set.
In order to define trigger regret, we identify suitable linear transformations that allow us to encode the behavior of trigger agents in the definition of EFCE. We call them \emph{canonical trigger deviation functions}.
Intuitively, trigger deviation functions encode the possible ways in which a trigger agent may deviate from recommended behavior in the EFCE formulation; this happens whenever a corresponding action sequence is recommended, which results in the agent playing from that point on some continuation strategy different from the recommended one.
Our core result on trigger regret is the following: if each player plays according to any uncoupled no-regret learning dynamics that minimizes trigger regret, then the resulting empirical frequency of play approaches the set of EFCEs.

In the rest of the article, we provide an efficient (that is, requiring time polynomial in the size of the game tree at each iteration) algorithm that minimizes trigger regret.
The algorithm is based on the general framework by~\citet{Gordon08:No}, which builds a phi-regret minimizer by employing a standard regret minimizer having the set of linear transformations defining phi-regret as decision space.
For our purposes, this boils down to designing a regret minimizer for the set of all valid canonical trigger deviation functions, which we do by exploiting non-trivial combinatorial structures of the set.
Moreover, in order to efficiently compute the next strategy to play at each iteration, the framework by~\citet{Gordon08:No} also needs that each linear transformation in the set admits a fixed point, and that such fixed point can be computed efficiently. 
Thus, our algorithm requires a procedure to compute a fixed point of any linear mapping defining a canonical trigger deviation function in time polynomial in the size of the game tree.
We show that such procedure can be implemented by visiting the game tree in a top down fashion, so as to incrementally build a fixed point.
In conclusion, our main result is that the proposed algorithm minimizes trigger regret, exhibiting $O(\sqrt{T})$ trigger regret with high probability after $T$ iterations.
Thus, when all players play according to the uncoupled learning dynamics defined by our algorithm, the empirical frequency of play after $T$ game repetitions is proven to be a $O(1/\sqrt{T})$-approximate EFCE with high probability, and an EFCE almost surely in the limit.
%
%
These results generalize the seminal work by~\citet{hart2000simple} to the extensive-form game case via a simple and natural framework.

\paragraph{ Related Work}

The study of adaptive procedures leading to a CE dates back to at least the seminal works by~\citet{foster1997calibrated},~\citet{fudenberg1995consistency,fudenberg1999conditional}, and~\citet{hart2000simple,hart2001general}; see also the monograph by~\citet{fudenberg1998theory}.
In particular, the work by~\citet{hart2000simple} proves that simple dynamics based on the notion of {\em internal regret} converge to a CE in normal-form games. The strategy that the authors introduce---the so-called \emph{regret matching}---is conceptually simple, and guarantees that if all players follow this strategy, then the empirical frequency of play converges to the set of CEs (see also~\citet{cahn2004general}).
Other works describe extensions to the models studied in the aforementioned papers. For example,~\citet{Stoltz07:Learning} describe an adaptive procedure converging to a CE in games with an infinite, but compact, set of actions, while~\citet{kakade2003correlated} consider efficient algorithms for computing correlated equilibria in graphical games.

In more recent years, a growing effort has been devoted to understanding the relationships between no-regret learning dynamics and equilibria in extensive-form games.
These games pose additional challenges when compared to normal-form games, due to their sequential nature and the presence of imperfect information.
While in two-player zero-sum extensive-form games it is widely known that no-regret learning dynamics converge to an NE---with the \emph{counterfactual regret minimization} (CFR) algorithm and its variations being the state of the art for equilibrium finding in such games~\citep{zinkevich2008regret,tammelin2014solving,tammelin2015solving,lanctot2009monte,brown2019solving}---the general case is less understood.
\citet{celli2019learning} provide some variations of the classical CFR algorithm for $n$-player general-sum extensive-form games, showing that they provably converge to a {\em normal-form coarse correlated equilibrium}, which is based on a form of correlation that is less appealing than that of EFCE in sequential games.

In a recent paper, \citet{morrill2020hindsight} conduct a study on various forms of correlation in extensive-form games, defining a taxonomy of solution concepts.
Each of their solution concepts is attained by a particular set of no-regret learning dynamics, which is obtained by instantiating the phi-regret minimization framework \citep{Greenwald03:General,Stoltz07:Learning,Gordon08:No} with a suitably-defined deviation function.
As part of their analysis, \citet{morrill2020hindsight} investigate some properties of the well-established CFR regret minimization algorithm~\citep{zinkevich2008regret} applied to $n$-player general-sum extensive-form games, establishing that it is hindsight-rational with respect to a specific set of deviation functions, which the authors coin {\em blind counterfactual deviations}.%
\footnote{In a very recent working paper, \citet{morrill2021efficient} extend their prior work~\citep{morrill2020hindsight} by identifying a general class of deviations---called {\em behavioral deviations}---that induce equilibria that can be found through uncoupled no-regret learning dynamics. Behavioral deviations are defined as those specifying an action transformation independently at each information set of the game. As the authors note, the deviation functions involved in the definition of EFCE do not fall under that category. However, the authors suggest that a particular class of behavioral deviation functions---called \emph{causal partial sequence deviations}---induce solution concepts that are (subsets of) EFCEs. In private communications with the authors, they indicated that they are working towards completing the proof of that claim. Once completed, their result would beget an alternative set of no-regret learning dynamics that converge to EFCE, based on a different set of deviation functions than those we use in this article.}

%% file: text/efgs.tex
\section{Preliminaries}

In this section we provide some standard definitions related to extensive-form games and regret minimization which will be employed in the remainder of the article. A more comprehensive treatment of basic concepts in game theory can be found in the book by~\citet{Shoham08:Multiagent}, and an introduction to the theory of learning in games can be found in the book by~\citet{Cesa-Bianchi06:Prediction}.

\subsection{Mathematical Notation and Algorithmic Conventions}\label{sec:notation}

In this article we adopt the following notational and algorithmic conventions.
\begin{itemize}
    \item We denote the set of real numbers as $\bbR$, the set of nonnegative real numbers as $\bbR_{\ge0}$, and the set $\{1,2,\dots\}$ of positive integers as $\bbN$.
    \item The set $\{1,\ldots,n\}$, where $n\in\bbN$, is compactly denoted as $[n]$. The empty set is denoted as $\emptyset$.
    \item Given a set $S$, we denote its convex hull with the symbol $\co {S}$.
    \item Vectors and matrices are marked in bold.
    \item Given a discrete set $S = \{s_1,\dots,s_n\}$, we denote as $\bbR^{|S|}$ (resp., $\bbR_{\ge0}^{|S|}$) the set of real (resp., nonnegative real) $|S|$-dimensional vectors whose entries are denoted as $\vec{x}[s_1], \dots,\vec{x}[s_n]$.
    \item Similarly, given a discrete set $S$, we denote as $\bbR^{|S|\times|S|}$ (resp., $\bbR_{\ge0}^{|S|\times|S|}$) the set of real (resp., nonnegative real) $|S|\times|S|$ square matrices $\vec{M}$ whose entries are denoted as $\vec{M}[s_r,s_c]$ ($s_r,s_c\in S$), where $s_r$ corresponds to the row index and $s_c$ to the column index.
    \item Given a discrete set $S$, we denote by $\Delta^{|S|}$ the $|S|$-simplex, that is, the set $\Delta^{|S|} \defeq \{\vec{x}\in\bbR_{\ge0}^{|S|}: \sum_{s\in S} \vec{x}[s] = 1\}$. The symbol $\Delta^n$, with $n\in\bbN$, is used to mean $\Delta^{|[n]|}$. 
    \item Given a discrete set $S$, we use the symbol $\bbS^{|S|} \subseteq \bbR_{\ge0}^{|S|\times|S|}$ to denote the set of stochastic matrices, that is, nonnegative square matrices whose columns all sum up to~$1$. The symbol $\bbS^n$, where $n\in\bbN$, is used to mean $\bbS^{|[n]|}$. 
    \item Given two functions $f:X\to Y$ and $g:Y\to Z$, we denote by $g\circ f: X\to Z$ their composition $\vec{x} \mapsto g(f(x))$.
    \item Given a proposition $\textsf{p}$, we denote with $\bbone[\textsf{p}]$ the indicator function of that proposition: $\bbone[\textsf{p}] = 1$ if $\textsf{p}$ is true, and $\bbone[\textsf{p}] =0$ if not.
    \item Given a partially ordered set $(S, \prec)$ and two elements $s,s'\in S$, we use the standard derived symbols $s \preceq s'$ to mean that $(s = s') \lor (s \prec s')$, $s \succ s'$ to mean that $s' \prec s$, and $s \succeq s'$ to mean that $s' \preceq s$. Furthermore, we use the crossed symbols $\not\prec, \not\preceq, \not\succ$, and  $\not\succeq$ to mean that the relations $\prec,\preceq,\succ$, and $\succeq$ (respectively) do not hold.
    \item Several of the algorithms presented in this article take as input, give as output, or otherwise manipulate, linear functions. Therefore, in order to study the complexity of our routines, it is necessary to settle on a representation for such linear functions. Unless otherwise specified, we will always assume that a linear function $f$ is stored in memory using coordinates relative to the canonical basis of their domain and codomain, and call that representation the $\emph{canonical representation}$ of $f$, denoted $\langle f\rangle$. Specifically:
    \begin{itemize}
        \item If $f$ is a linear function from $\bbR^{|S|}$ (for some discrete set $S$) to $\bbR$, then its canonical representation $\langle f\rangle$ is the (unique) vector $\vec{v} \in\bbR^{|S|}$ such that
        \[
            f(\vec{x}) = \vec{v}^\top \vec{x}\qquad\forall\ \vec{x}\in\bbR^{|S|},
        \]
        where $\top$ denotes transposition.
        \item If $f$ is a linear function from $\bbR^{|S|}$ to $\bbR^{|S|}$ (for some discrete set $S$), then its canonical representation $\langle f\rangle$ is the (unique) matrix $\vec{M}\in \bbR^{|S|\times |S|}$ such that
        \[
            f(\vec{x}) = \vec{M} \vec{x} \qquad\forall\ \vec{x} \in \bbR^{|S|}.
        \]
        \item If $F$ is a linear functional, mapping linear functions $\phi:\bbR^{|S|}\to\bbR^{|S|}$ to reals, then its canonical representation $\langle F\rangle$ is the (unique) matrix $\vec{\Lambda} \in \bbR^{|S|\times|S|}$ such that
        \[
            F(\phi) = \sum_{s_r, s_c\in S} \vec{\Lambda}[s_r, s_c]\cdot \langle\phi\rangle[s_r, s_c]\qquad \forall\ \phi:\bbR^{|S|}\to\bbR^{|S|},\numberthis{eq:functional canonical}
        \]
        where $\langle\phi\rangle$ is the canonical representation of $\phi$.
    \end{itemize}
\end{itemize}

\subsection{Extensive-Form Games}\label{sec:extensive form}

In this subsection we introduce some standard concepts, terminology, and notation that we will use to deal with extensive-form games. A summary of the notation we introduce can be found in \cref{tab:notation}. \cref{ex:preliminary,ex:preliminary part 2} demonstrate some of the notation in a simple extensive-form game.
\begin{table}[htp]
    \begin{tabular}{p{1.5cm}p{11.5cm}}
    \toprule
        \centering\textbf{Symbol} & \textbf{Description}\\
    \midrule
        \centering$\Hist$ & Set of nodes of the game tree.\\
        \centering$\Hist^{(i)}$& Set of nodes at which Player~$i$ acts.\\
        \centering $\A(h)$ & Actions available to the player acting at $h \in \Hist$ (empty set if $h$ is a terminal node).\\
        \centering$\cJ$ & Information partition of Player~$i$.\\
        \centering$\A(j)$ & Set of actions available at any node in the information set $j$.\\
    \midrule
        \centering$\Z$ & Set of terminal nodes (leaves of the game tree).\\
        \centering$\ut(z)$ & Payoff of Player~$i$ at terminal node $z \in \Z$.\\
        \centering$\pc$ & Product of probabilities of all the stochastic events on the path from the root to terminal node $z\in\Z$.\\
    \midrule
        \centering$\Seqs$ & Set of sequences of Player~$i$, defined as $\Seqs \defeq \{(j,a) : j \in \cJ, a \in A_j\} \cup \{\emptyseq\}$,\\
        \centering$\emptyseq$ & where the special element $\emptyseq$ is called the \emph{empty sequence}.\\
                \centering$\Seqs_*$ & Set of sequences of Player~$i$, excluding the empty sequence $\emptyseq$.\\
        \centering$\sigma^{(i)}(z)$ & Last sequence of Player~$i$ encountered on the path from the root to node $z\in\Z$.\\
        \centering$\sigma^{(i)}(j)$ & Last sequence of Player~$i$ on the path from the root to any node in $j \in \cJ$.\\
    \midrule
        \centering$j' \prec j$ & Information set $j\in\cJ$ is an ancestor of $j'\in\cJ$, that is, there exists a path in the game tree connecting a node $h\in j$ to some node $h'\in j'$. \\
        \centering$\sigma\prec\sigma'$ & Sequence $\sigma$ precedes sequence $\sigma'$, where $\sigma,\sigma'$ belong to the same player. \\
        \centering $\sigma\succeq j$ & Sequence $\sigma=(j',a')$ is such that $j' \succeq j$.\\
        \centering $\Seqs_j$ & Sequences at $j \in \cJ$ and all of its descendants, $\Seqs_j\defeq \{\sigma\in\Seqs: \sigma \succeq j\}$. \\
    \midrule
        \centering $\Q$ & Sequence-form strategies of Player~$i$ (\cref{def:sequence_form_polytope}).\\
        \centering $\Q_j$ & Sequence-form strategies for the subtree\linkfootnote{fn:subtree} rooted at $j \in \cJ$ (\cref{def:Qj}).\\
        \centering $\Pure$ & Deterministic sequence-form strategies of Player~$i$.\\
        \centering $\Pure_j$ & Deterministic sequence-form strategies for the subtree\linkfootnote{fn:subtree} rooted at $j \in \cJ$.\\
        \centering $\Pi$ & Set of joint deterministic sequence-form strategies, $\Pi\defeq\bigtimes_{i\in[n]}\Pure$.\\
    \bottomrule
    \end{tabular}
    \caption{Summary of game-theoretic notation used in this article.}
    \label{tab:notation}
\end{table}

An extensive-form game is played on an oriented rooted \emph{game tree}. We denote by $\Hist$ the set of nodes of the game tree. Each node $h\in\Hist$ that is not a leaf of the game tree is called a \emph{decision node}, and has an associated player that acts at that node. In an $n$-player extensive-form game, the set of valid players is the set $[n]\cup\{c\}$, where $c$ denotes the \emph{chance player}---a fictitious player that selects actions according to a known fixed probability distribution and models exogenous stochasticity of the environment (for example, a roll of a dice or a drawing a card from a deck). The player that acts at $h$ is free to pick any one of the actions $\A(h)$ that are available at $h$. For each possible action $a\in\A(h)$, an edge connects $h$ to the node to which the game transitions whenever action $a$ is picked at $h$. Given a player~$i \in [n]\cup\{c\}$, we denote with $\Hist^{(i)} \subseteq \Hist$ the set of all decision nodes that belong to Player~$i$.

Leaves of the game tree are called \emph{terminal nodes}, and represent the outcomes of the game. As such, they are not associated with any acting player, and the set of actions is conventionally set to the empty set. The set of all terminal nodes in the game is denoted with the letter $\Z$. So, the set of all nodes in the game tree is the disjoint union $\Hist = \Hist^{(1)} \cup \dots \cup \Hist^{(n)}\cup \Z$. When the game transitions to a terminal node~$z\in\Z$, payoffs are assigned to each of the non-chance players by the set of functions $\{\ut:\Z\to\mathbb{R}\}_{i\in[n]}$.
Furthermore, we let $p^{(c)} : \Z\to (0,1)$ denote the function assigning each terminal node $z\in\Z$ to the product of probabilities of chance moves encountered on the path from the root of the game tree to $z$.

\subsubsection{Imperfect information}
To model imperfect information, the nodes $\Hist^{(i)}$ of each player $i\in[n]$ are partitioned into a set $\cJ$ of groups, called \emph{information sets}. Each information set $j\in\cJ$ groups together nodes that Player~$i$ cannot distinguish between. Since a player always knows what actions are available at a decision node, any two nodes $h,h'\in j$ must have the same action set, that is, $\A(h) = \A(h')$. For that reason, we slightly overload notation and write $\A(j)$ to mean the set of actions available at any node that belongs to information set $j$.

As is standard in the literature, we assume that the extensive-form game has \emph{perfect recall}, that is, information sets are such that no player forgets information once acquired. 
An immediate consequence of perfect recall is that, for any player $i\in[n]$ and any two nodes $h,h'$ in the same information set $j \in \cJ$, the sequence of Player~$i$'s actions encountered along the path from the root to $h$ and from the root to $h'$ must coincide (or otherwise Player~$i$ would be able to distinguish among the nodes, since the player remembers all of the actions they played in the past). This suggests the following partial ordering $\prec$ on the set $\cJ$: we write $j\prec j'$---and say that $j\in\cJ$ is an \emph{ancestor} of $j'\in\cJ$ or equivalently that $j'$ is a \emph{descendant} of $j$---if there exist nodes $h'\in j'$ and $h\in j$ such that the path from the root of the game tree to $h'$ passes through $h$. 

It is a well-known consequence of perfect recall that the partially ordered set $(\cJ,\prec)$ is a forest for any player $i\in[n]$, in the precise sense that, given any information set $j\in\cJ$, the set of all of its predecessors forms a chain (that is, it is well-ordered by $\prec$).

\subsubsection{Sequences}\label{sec:sequences}
For any player $i\in[n]$, and given an information set $j\in\cJ$ and an action $a\in\A(j)$, we denote as $\sigma=(j,a)$ the sequence of Player~$i$'s actions encountered on the path from the root of the game tree down to action $a$ (included) at any node in information set $j$.
In perfect-recall extensive-form games, such sequence is guaranteed to be uniquely determined because paths that reach decision nodes belonging to the same information set identify the same sequence of Player~$i$'s actions. 
A special element $\emptyseq$ denotes the \emph{empty sequence} of Player~$i$.
Then, the set of Player~$i$'s sequences is defined as 
\[
    \Seqs\defeq\mleft\{ (j,a): j\in\cJ, a\in\A(j)\mright\}\cup \mleft\{\emptyseq\mright\}.
\]
Moreover, we let  $\Seqs_* \defeq \Seqs \setminus \{\emptyseq\}$ be the set of all sequences of Player~$i$ other than the empty one.

Given a node $h\in\Hist$, we denote by $\sigma^{(i)}(j)\in\Seqs$ the last sequence (information set-action pair) of Player~$i$ encountered on the path from the the root of the game tree to any node in $j$. If Player~$i$ does not act before $h$, then $\sigma^{(i)}(h)$ is set to the empty sequence $\emptyseq$. 
If $\sigma^{(i)}(j) = \emptyseq$ we say that information set $j$ is a {\em root information set} of Player~$i$, while, whenever $\sigma^{(i)}(j)=(j',a)$, we say that information set $j$ is {\em immediately reachable} from sequence $\sigma^{(i)}(j)$, because $j'\prec j$ and Player~$i$ does not need to take other actions after choosing $a$ at $j'$ in order to reach $j$. 
Analogously, for all $z\in\Z$, we define $\sigma^{(i)}(z)\in\Seqs$ as the last sequence of Player~$i$'s actions encountered the path from the root of the game tree to terminal node $z$ (notice that $\sigma^{(i)}(z) = \emptyseq$ whenever Player~$i$ never plays on the path from the root to $z$).

Just like information sets, there exists a natural partial ordering on sequences, which we also denote with the same symbol $\prec$. For every $i \in [n]$ and any pair of sequences $\sigma,\sigma' \in \Seqs$, the relation $\sigma\prec\sigma'$ holds if $\sigma=\emptyseq\neq \sigma'$, or if the sequences are of the form $\sigma=(j,a),\sigma'=(j',a')$, and the set of Player $i$'s actions encountered on the path from the root of the tree to any node in~$j'$ includes playing action~$a$ at one of information set~$j$'s nodes. As for information sets, it is a direct consequence of the perfect recall assumption that the partially ordered set $(\Seqs, \prec)$ is a forest.
%
%
Finally, we introduce the overloaded notation $\sigma\succeq j$ (or equivalently $j \preceq \sigma$), defined for any player $i \in [n]$, information set $j\in\cJ$, and sequence $\sigma\in\Seqs$, to mean that the sequence of Player~$i$ 's actions that is denoted by $\sigma$ must lead the player to pass through (some node in) $j$; formally $\sigma = (j',a') \in \Seqs_* \land j'\succeq j$.
With that, we let $\Seqs_j \defeq\{\sigma\in\Seqs: \sigma\succeq j\} \subseteq\Seqs$ be the set of Player~$i$'s sequences that terminate at $j$ or any of its descendant information sets.

\begin{example}\label{ex:preliminary}
    To illustrate some of the concepts and notation described so far, we consider the simple two-player extensive-form game in \cref{fig:preliminary_example}, in which black round nodes belong to Player~$1$, and white round nodes belong to Player~$2$. The gray clusters of nodes identify the information sets. Since we chose different action numbers for different information sets, there exists a one-to-one correspondence between actions and sequences, and we will sometimes refer to sequences using the corresponding action number. For example, we will sometimes refer to sequence ``\seq{3}'' to mean sequence $(\textsc{b},\seq{3})$, sequence ``\seq{8}'' to mean sequence $(\textsc{d},\seq{8})$, \emph{et cetera}.
    Player~$1$ has four information sets---denoted \textsc{a}, \textsc{b}, \textsc{c}, and \textsc{d}---with two actions each. Player~$2$ only has two information sets, \textsc{r} and \textsc{s}, each with two actions. Information set \textsc{d} of Player~$1$ contains two nodes, and models Player~$1$'s lack of knowledge of the action taken by Player 2 at information set \textsc{s}. The partial ordering between information sets for Player~$1$ is $\textsc{a}\prec\textsc{b},\textsc{a}\prec\textsc{c},\textsc{a}\prec\textsc{d}$.
    Moreover, we have that $\sigma^{(1)}(\textsc{a})=\emptyseq$, $\sigma^{(1)}(\textsc{b})=\sigma^{(1)}(\textsc{c})=\seq{1}$, and  $\sigma^{(1)}(\textsc{d})=\seq{2}$. For the terminal node $z$ in the picture, $\sigma^{(1)}(z)=\seq{3}$.
    Finally, we have that $\Sigma^{(1)}_{\textsc{a}}=\Sigma_*^{(1)}=\{\seq{1},\seq{2},\dots,\seq{8}\}$, $\Sigma^{(1)}_{\textsc{b}}=\{\seq{3},\seq{4}\}$, $\Sigma^{(1)}_{\textsc{c}}=\{\seq{5},\seq{6}\}$, and $\Sigma^{(1)}_{\textsc{d}}=\{\seq{7},\seq{8}\}$.
    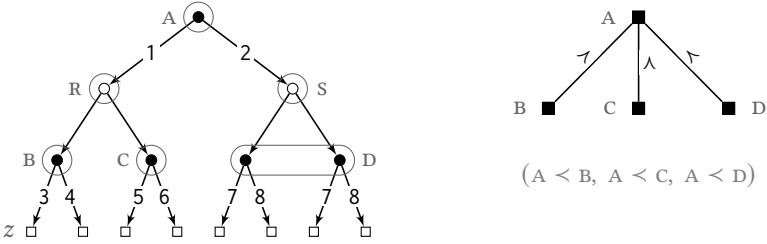
\begin{figure}[th]
        \centering
        \raisebox{1cm}{\input{imgs/small_efg}}
        \qquad\qquad
        \raisebox{1cm}{\input{imgs/infoset_forest}}
        \caption{(Left) Example of an extensive-form game with two players. Black round nodes belong to Player 1, white round nodes belong to Player 2. Small white square nodes represent terminal nodes.  The gray partitions represent the information sets of the game. The numbers on the edges identify each of Player~$1$'s actions. (Right): Forest of information sets of Player~$1$, corresponding to the partially ordered set $(\cJ[1],\prec)$.}
        \label{fig:preliminary_example}
    \end{figure}
\end{example}

\subsubsection{Sequence-form strategies}
Conceptually, a strategy for a player specifies a probability distribution over the actions at each information set for that player. So, perhaps the most intuitive representation of a strategy, called a \emph{behavioral strategy} in the literature, is as a vector that assigns to each information set-action pair $(j,a)\in\Seqs_*$ the probability of picking action $a$ at information set $j$. That representation has a major drawback: the probability of reaching any given terminal node $z \in \Z$ is expressed as the product of several entries in the vector (one per each action on the path from the root of the game tree to $z$), rendering critical quantities---including the expected utility of a player---a non-convex function of the behavioral strategies of the players. As is standard in the literature, to soundly overcome the issue of non-convexity, throughout this article we will exclusively use a different representation of strategies, known as the \emph{sequence-form} representation~\citep{Romanovskii62:Reduction,Koller96:Efficient,Stengel96:Efficient}.

Like behavioral strategies, a \emph{sequence-form strategy}\footnote{Sequence-form strategies are also known under the term \emph{realization plans} in the literature (\emph{e.g.}, \citep{Stengel96:Efficient}). We will not use that latter term in this article.} for Player~$i \in [n]$ is a vector $\q\in\mathbb{R}^{|\Seqs|}_{\ge 0}$. However, unlike behavioral strategies, each entry $\q[(j,a)]$ of a sequence-form strategy $\q$ contains the \emph{product} of the probabilities of playing all of Player~$i$'s actions on the path from the root of the game tree down to action $a$ at information set $j$ included. Furthermore, the entry $\q[\emptyseq]$ corresponding to the empty sequence is defined as the constant value $1$.

To ensure consistency, all sequence-form strategies must satisfy the probability-mass-conservation constraints
\[
\q{}[\emptyseq]=1, \hspace{1cm} \q{}[\sigma^{(i)}(j)]=\sum_{a\in\A(j)}\q{}[(j,a)], \hspace{.3cm}\forall\, j\in\cJ.
\]
The above probability-mass-conservation constraints are linear, and therefore the set of sequence-form strategies is a convex polytope, suggesting the following definition. 

\begin{definition}\label{def:sequence_form_polytope}
The \emph{sequence-form strategy polytope} for Player $i \in [n]$ is the convex polytope
\[
\Q\defeq \mleft\{ \q\in \mathbb{R}^{|\Seqs|}_{\ge 0}: \q{}[\emptyseq]=1\hspace{.3cm}\text{\normalfont and} \hspace{.3cm} \q{}[\sigma^{(i)}(j)]=\sum_{a\in\A(j)}\q{}[(j,a)], \hspace{.1cm}\forall\, j\in\cJ \mright\}.
\]
\end{definition}

As we mentioned in \cref{sec:sequences}, the partially ordered set $(\cJ,\prec)$ is a forest. Thus, it makes sense to consider partial strategies that only specify behavior at an information set $j$ and all of its descendants $j' \succ j$. We make that formal through the following definition.
\begin{definition}\label{def:Qj}
Let $i\in[n]$ be a player and $j\in\cJ$ be an information set for Player~$i$. The \emph{set of sequence-form strategies for the subtree}\footnote{\label{fn:subtree}The term ``subtree'' does not refer to a subtree of the game tree, but rather to a subtree of the partially ordered set $(\cJ,\prec)$. In other words, the term \emph{subtree} here refers to the fact that the quantities are specified only at information set $j$ and all of its descendants.} \emph{rooted at $j$}, denoted $\Q_j$, is the set of all vectors $\q \in \bbR_{\ge 0}^{|\Seqs_j|}$ such that probability-mass-conservation constraints hold at information set $j$ and all of its descendants $j' \succ j$, specifically
\[
    \Q_j\defeq \mleft\{ \q\in\mathbb{R}_{\ge0}^{|\Seqs_j|} : \sum_{a\in\A(j)} \q[(j,a)]=1, \hspace{.3cm}\text{\normalfont and}\hspace{.3cm} \q[\sigma^{(i)}(j')]=\sum_{a\in\A(j')}\q{}[(j',a)] \quad\forall\, j'\succ j \mright\}.
\]
\end{definition}

\subsubsection{Deterministic sequence-form strategies}

Deterministic strategies are those that select, at each information set at which the player acts, exactly one action with probability one. Since the probability mass on each action is either $0$ or $1$, the set of deterministic sequence-form strategies for Player~$i$---which we denote with the capital letter $\Pure$---corresponds exactly with the set of all sequence-form strategies whose components are all either 0 or 1. 
\begin{definition}\label{def:deterministic_sequence_form}
The set of \emph{deterministic sequence-form strategies} for Player $i \in [n]$ is the set
\[
    \Pure\defeq \Q\cap \{0,1\}^{|\Seqs|}.
\]
Similarly, the set of \emph{deterministic sequence-form strategies for the subtree\linkfootnote{fn:subtree} rooted at $j$} is 
\[
    \Pure_j \defeq \Q_j \cap \{0,1\}^{|\Seqs_j|}.
\]
\end{definition}

The set of deterministic sequence-form strategies corresponds one-to-one to the game-theoretic notion of \emph{reduced normal-form strategies} (\emph{e.g.},~\citet[Section 4]{Stengel96:Efficient}). %
Furthermore, Kuhn's Theorem~\cite{kuhn1953} implies that
\[
    \Q = \co\Pure, \Q_j = \co\Pure_j  \qquad\quad\forall\,i\in[n], j\in\cJ.
\]

When it is important to emphasize that an arbitrary sequence-form strategy $\q\in\Q$ (or $\q \in \Q_j$ for some $j\in\cJ$) of Player~$i \in [n]$ need \emph{not} be a deterministic sequence-form strategy, we will say that $\q$ is a \emph{mixed} sequence-form strategy.

Given a sequence-form strategy $\q\in\Q$, it is possible to build an unbiased sampling scheme resulting in a (random) deterministic strategy $\vec{\pi}\in\Pure$ such that $\bbE[\vec{\pi}]=\q$.
A natural unbiased sampling procedure is the following. Start from any root information set of Player~$i$, that is, an information set $j \in \cJ$ such that $\sigma^{(i)}(j) = \emptyseq$. Given any information set $j \in \cJ$, an action $a_j\in\A(j)$ is sampled with probability $\q{}[(j,a_j)]/\q{}[\sigma^{(i)}(j)]$; then, the same procedure is applied recursively to all information sets immediately reachable from sequence $(j,a_j)$, that is, the information sets $j' \in\cJ$ such that $\sigma^{(i)}(j') = (j,a_j)$. The process is repeated for all the root information sets of Player~$i$. 
The final deterministic sequence-form strategy $\vec{\pi}$ is obtained by setting $\vec{\pi}[(j,a_j)]=1$ for each information set $j \in \cJ$ visited during the procedure, and all other entries equal to $0$.

Finally, we denote as $\Pi\defeq\bigtimes_{i\in[n]}\Pure$ the set of \emph{joint} deterministic sequence-form strategies of all the players. Therefore, an element of $\Pi$ is a tuple $\vec{\pi} = (\vec{\pi}^{(1)},\dots,\vec{\pi}^{(n)})$ specifying a deterministic sequence-form strategy $\pure$ for each player~$i \in [n]$.

\begin{example}\label{ex:preliminary part 2}
Continuing \cref{ex:preliminary}, in \cref{fig:seq_strategies} we provide one (mixed) sequence-form strategy $\vec{q}\in\Q[1]$, and five deterministic sequence-form strategies $\{\vec{q}_{\seq{135}},\vec{q}_{\seq{136}},\vec{q}_{\seq{145}},\vec{q}_{\seq{27}},\vec{q}_{\seq{28}}\}\subseteq\Pure[1]$ for the small game in \cref{fig:preliminary_example}~(Left).
One can check that these vectors are valid sequence-form strategies by verifying that the probability-mass-conservation constraints of \cref{def:sequence_form_polytope} hold.
Let us consider the mixed sequence-form strategy $\q$. There, $\q[\seq{1}]=\q[\seq{2}]=0.5$, and therefore Player 1 will select a sequence between $\seq{1}$ and $\seq{2}$ uniformly at random. Suppose Player 1 selects sequence $\seq{1}$. Then, if Player 1 reached information set \textsc{b}, she would select sequences $\seq{3}$ and $\seq{4}$ with probability $0.25 / 0.5=0.5$ each. On the other hand, if Player 1 reached information set \textsc{c}, she would choose sequence $\seq{5}$ with probability $0.1/0.5=0.2$, and sequence $\seq{6}$ with probability $0.4/0.5=0.8$. Analogously, if Player 1 played sequence $\seq{2}$ at information set \textsc{a}, upon reaching information set \textsc{d} she would play sequence $\seq{8}$ with probability $0.5/0.5=1$. As expected, the probability of playing action $a$ at a generic information set $j$ can be obtained by dividing $\q[(j,a)]$ by  $\q[\sigma^{(i)}(j)]$.
As a second example, consider the deterministic sequence-form strategy $\q_{\seq{136}}$. When Player~$1$ plays according to that strategy, she will always choose sequence $\seq{1}$ at information set \textsc{a}, sequence $\seq{3}$ at information set \textsc{b}, and sequence $\seq{6}$ at information set \textsc{c}. It is impossible for the player to reach information set \textsc{d} given her strategy at \textsc{a} and correspondigly $\q_{\seq{136}}[\seq{7}] = \q_{\seq{136}}[\seq{8}] = 0$.

\begin{figure}
\centering
    \setlength{\tabcolsep}{10.5pt}
    \begin{tabular}{p{1.5cm}|p{1.5cm}|p{1.5cm}|p{1.5cm}|p{1.5cm}|p{1.5cm}}
    \toprule
        \hfil $\vec{q}$ \hfil
    & 
        \hfil $\vec{q}_{\seq{135}}$ \hfil
    &
        \hfil $\vec{q}_{\seq{136}}$ \hfil
    &
        \hfil $\vec{q}_{\seq{145}}$ \hfil
    &
        \hfil $\vec{q}_{\seq{27}}$ \hfil
    &
        \hfil $\vec{q}_{\seq{28}}$ \hfil\\
    \midrule
    \!\!\begin{tikzpicture}[baseline=-\the\dimexpr\fontdimen22\textfont2\relax ]
            \tikzset{every left delimiter/.style={xshift=1.5ex},
                     every right delimiter/.style={xshift=-1ex}};
            \matrix [matrix of math nodes,left delimiter=(,right delimiter=),row sep=.007cm,column sep=.007cm,color=black!20](m)
            {
            |[black]| 1.00 \\
            |[black]| 0.50 \\
            |[black]|0.50 \\
            |[black]| 0.25 \\
            |[black]| 0.25 \\
            |[black]| 0.10 \\
            |[black]| 0.40 \\
            0.00 \\
            |[black]| 0.50 \\
            };
            
            \node[left=3pt of m-1-1] (left-0) {\small$\emptyseq$};
            \node[left=3pt of m-2-1] (left-0) {\small$\seq{1}$};
            \node[left=3pt of m-3-1] (left-0) {\small$\seq{2}$};
            \node[left=3pt of m-4-1] (left-0) {\small$\seq{3}$};
            \node[left=3pt of m-5-1] (left-0) {\small$\seq{4}$};
            \node[left=3pt of m-6-1] (left-0) {\small$\seq{5}$};
            \node[left=3pt of m-7-1] (left-0) {\small$\seq{6}$};
            \node[left=3pt of m-8-1] (left-0) {\small$\seq{7}$};
            \node[left=3pt of m-9-1] (left-0) {\small$\seq{8}$};
    \end{tikzpicture}\
    &
    \!\!\begin{tikzpicture}[baseline=-\the\dimexpr\fontdimen22\textfont2\relax ]
            \tikzset{every left delimiter/.style={xshift=1.5ex},
                     every right delimiter/.style={xshift=-1ex}};
            \matrix [matrix of math nodes,left delimiter=(,right delimiter=),row sep=.007cm,column sep=.007cm,color=black!25](m)
            {
            |[black]| 1.0 \\
            |[black]| 1.0 \\
            0.0 \\
            |[black]| 1.0 \\
            0.0 \\
            |[black]| 1.0 \\
            0.0 \\
            0.0 \\
            0.0 \\
            };
            
            \node[left=3pt of m-1-1] (left-0) {\small$\emptyseq$};
            \node[left=3pt of m-2-1] (left-0) {\small$\seq{1}$};
            \node[left=3pt of m-3-1] (left-0) {\small$\seq{2}$};
            \node[left=3pt of m-4-1] (left-0) {\small$\seq{3}$};
            \node[left=3pt of m-5-1] (left-0) {\small$\seq{4}$};
            \node[left=3pt of m-6-1] (left-0) {\small$\seq{5}$};
            \node[left=3pt of m-7-1] (left-0) {\small$\seq{6}$};
            \node[left=3pt of m-8-1] (left-0) {\small$\seq{7}$};
            \node[left=3pt of m-9-1] (left-0) {\small$\seq{8}$};
    \end{tikzpicture}\
    &
    \!\!\begin{tikzpicture}[baseline=-\the\dimexpr\fontdimen22\textfont2\relax ]
            \tikzset{every left delimiter/.style={xshift=1.5ex},
                     every right delimiter/.style={xshift=-1ex}};
            \matrix [matrix of math nodes,left delimiter=(,right delimiter=),row sep=.007cm,column sep=.007cm,color=black!20](m)
            {
            |[black]| 1.0 \\
            |[black]| 1.0 \\
            0.0 \\
            |[black]| 1.0 \\
            0.0 \\
            0.0 \\
            |[black]| 1.0 \\
            0.0 \\
            0.0 \\
            };
            
            \node[left=3pt of m-1-1] (left-0) {\small$\emptyseq$};
            \node[left=3pt of m-2-1] (left-0) {\small$\seq{1}$};
            \node[left=3pt of m-3-1] (left-0) {\small$\seq{2}$};
            \node[left=3pt of m-4-1] (left-0) {\small$\seq{3}$};
            \node[left=3pt of m-5-1] (left-0) {\small$\seq{4}$};
            \node[left=3pt of m-6-1] (left-0) {\small$\seq{5}$};
            \node[left=3pt of m-7-1] (left-0) {\small$\seq{6}$};
            \node[left=3pt of m-8-1] (left-0) {\small$\seq{7}$};
            \node[left=3pt of m-9-1] (left-0) {\small$\seq{8}$};
    \end{tikzpicture}\
    &
    \!\!\begin{tikzpicture}[baseline=-\the\dimexpr\fontdimen22\textfont2\relax ]
            \tikzset{every left delimiter/.style={xshift=1.5ex},
                     every right delimiter/.style={xshift=-1ex}};
            \matrix [matrix of math nodes,left delimiter=(,right delimiter=),row sep=.007cm,column sep=.007cm,color=black!20](m)
            {
            |[black]| 1.0 \\
            |[black]| 1.0 \\
            0.0 \\
            0.0 \\
            |[black]| 1.0 \\
            |[black]| 1.0 \\
            0.0 \\
            0.0 \\
            0.0 \\
            };
            
            \node[left=3pt of m-1-1] (left-0) {\small$\emptyseq$};
            \node[left=3pt of m-2-1] (left-0) {\small$\seq{1}$};
            \node[left=3pt of m-3-1] (left-0) {\small$\seq{2}$};
            \node[left=3pt of m-4-1] (left-0) {\small$\seq{3}$};
            \node[left=3pt of m-5-1] (left-0) {\small$\seq{4}$};
            \node[left=3pt of m-6-1] (left-0) {\small$\seq{5}$};
            \node[left=3pt of m-7-1] (left-0) {\small$\seq{6}$};
            \node[left=3pt of m-8-1] (left-0) {\small$\seq{7}$};
            \node[left=3pt of m-9-1] (left-0) {\small$\seq{8}$};
    \end{tikzpicture}\
    &
    \!\!\begin{tikzpicture}[baseline=-\the\dimexpr\fontdimen22\textfont2\relax ]
            \tikzset{every left delimiter/.style={xshift=1.5ex},
                     every right delimiter/.style={xshift=-1ex}};
            \matrix [matrix of math nodes,left delimiter=(,right delimiter=),row sep=.007cm,column sep=.007cm,color=black!20](m)
            {
            |[black]| 1.0 \\
            0.0 \\
            |[black]| 1.0 \\
            0.0 \\
            0.0 \\
            0.0 \\
            0.0 \\
            |[black]| 1.0 \\
            0.0 \\
            };
            
            \node[left=3pt of m-1-1] (left-0) {\small$\emptyseq$};
            \node[left=3pt of m-2-1] (left-0) {\small$\seq{1}$};
            \node[left=3pt of m-3-1] (left-0) {\small$\seq{2}$};
            \node[left=3pt of m-4-1] (left-0) {\small$\seq{3}$};
            \node[left=3pt of m-5-1] (left-0) {\small$\seq{4}$};
            \node[left=3pt of m-6-1] (left-0) {\small$\seq{5}$};
            \node[left=3pt of m-7-1] (left-0) {\small$\seq{6}$};
            \node[left=3pt of m-8-1] (left-0) {\small$\seq{7}$};
            \node[left=3pt of m-9-1] (left-0) {\small$\seq{8}$};
    \end{tikzpicture}\
    &
    \!\!\begin{tikzpicture}[baseline=-\the\dimexpr\fontdimen22\textfont2\relax ]
            \tikzset{every left delimiter/.style={xshift=1.5ex},
                     every right delimiter/.style={xshift=-1ex}};
            \matrix [matrix of math nodes,left delimiter=(,right delimiter=),row sep=.007cm,column sep=.007cm,color=black!20](m)
            {
            |[black]| 1.0 \\
            0.0 \\
            |[black]| 1.0 \\
            0.0 \\
            0.0 \\
            0.0 \\
            0.0 \\
            0.0 \\
            |[black]| 1.0 \\
            };
            
            \node[left=3pt of m-1-1] (left-0) {\small$\emptyseq$};
            \node[left=3pt of m-2-1] (left-0) {\small$\seq{1}$};
            \node[left=3pt of m-3-1] (left-0) {\small$\seq{2}$};
            \node[left=3pt of m-4-1] (left-0) {\small$\seq{3}$};
            \node[left=3pt of m-5-1] (left-0) {\small$\seq{4}$};
            \node[left=3pt of m-6-1] (left-0) {\small$\seq{5}$};
            \node[left=3pt of m-7-1] (left-0) {\small$\seq{6}$};
            \node[left=3pt of m-8-1] (left-0) {\small$\seq{7}$};
            \node[left=3pt of m-9-1] (left-0) {\small$\seq{8}$};
    \end{tikzpicture}\
    \\
    \bottomrule
    \end{tabular}
    \caption{Examples of sequence-form strategies for Player~$i$ in the game of \cref{fig:preliminary_example} (Left).
    }
    \label{fig:seq_strategies}
\end{figure}
\end{example}

%% file: imgs/small_efg.tex
\begin{tikzpicture}[>=latex',baseline=0pt,scale=.98]
    \def\done{.8*1.6}
    \def\dtwo{.40*1.6}
    \def\dleaf{.22*1.6}
    \def\dvert{-.8*1.2}

    \node[fill=black,draw=black,circle,inner sep=.5mm] (A) at (0, 0) {};
    \node[fill=white,draw=black,circle,inner sep=.5mm] (X) at ($(-\done,\dvert)$) {};
    \node[fill=white,draw=black,circle,inner sep=.5mm] (Y) at ($(\done,\dvert)$) {};
    \node[fill=black,draw=black,circle,inner sep=.5mm] (B) at ($(X) + (-\dtwo, \dvert)$) {};
    \node[fill=black,draw=black,circle,inner sep=.5mm] (C) at ($(X) + (\dtwo, \dvert)$) {};
    \node[fill=white,draw=black,inner sep=.6mm] (l1) at ($(B) + (-\dleaf, \dvert)$) {};
    \node[fill=white,draw=black,inner sep=.6mm] (l2) at ($(B) + (\dleaf, \dvert)$) {};
    \node[fill=white,draw=black,inner sep=.6mm] (l3) at ($(C) + (-\dleaf, \dvert)$) {};
    \node[fill=white,draw=black,inner sep=.6mm] (l4) at ($(C) + (\dleaf, \dvert)$) {};
    
    \node[fill=black,draw=black,circle,inner sep=.5mm] (D1) at ($(Y) + (-\dtwo, \dvert)$) {};
    \node[fill=black,draw=black,circle,inner sep=.5mm] (D2) at ($(Y) + (\dtwo, \dvert)$) {};
    \node[fill=white,draw=black,inner sep=.6mm] (l5) at ($(D1) + (-\dleaf, \dvert)$) {};
    \node[fill=white,draw=black,inner sep=.6mm] (l7) at ($(D1) + (\dleaf, \dvert)$) {};
    \node[fill=white,draw=black,inner sep=.6mm] (l8) at ($(D2) + (-\dleaf, \dvert)$) {};
    \node[fill=white,draw=black,inner sep=.6mm] (l10) at ($(D2) + (\dleaf, \dvert)$) {};

    \draw[semithick] (A) edge[->] node[fill=white,inner sep=.9] {\small\seq{1}} (X);
    \draw[->,semithick] (A) --node[fill=white,inner sep=.9] {\small\seq{2}} (Y);
    \draw[->,semithick] (B) --node[fill=white,inner sep=.9] {\small\seq{3}} (l1);
    \draw[->,semithick] (B) --node[fill=white,inner sep=.9] {\small\seq{4}} (l2);
    \draw[->,semithick] (C) --node[fill=white,inner sep=.9] {\small\seq{5}} (l3);
    \draw[->,semithick] (C) --node[fill=white,inner sep=.9] {\small\seq{6}} (l4);
    \draw[->,semithick] (D1) --node[fill=white,inner xsep=0,inner ysep=.9] {\small\seq{7}} (l5);
    \draw[->,semithick] (D1) --node[fill=white,inner xsep=0,inner ysep=.9,xshift=.6] {\small\seq{8}} (l7);
    \draw[->,semithick] (D2) --node[fill=white,inner xsep=0,inner ysep=.9] {\small\seq{7}} (l8);
    \draw[->,semithick] (D2) --node[fill=white,inner xsep=0,inner ysep=.9,xshift=.6] {\small\seq{8}} (l10);
    \draw[->,semithick] (X) -- (B);
    \draw[->,semithick] (X) -- (C);
    \draw[->,semithick] (Y) -- (D1);
    \draw[->,semithick] (Y) -- (D2);
    
    \draw[black!60!white] (X) circle (.2);
    \node[black!60!white]  at ($(X) + (-.4, 0)$) {\textsc{r}};
    \draw[black!60!white] (Y) circle (.2);
    \node[black!60!white]  at ($(Y) + (.4, 0)$) {\textsc{s}};
    
    \draw[black!60!white] (A) circle (.2);
    \node[black!60!white]  at ($(A) + (-.4, 0)$) {\textsc{a}};
    
    \draw[black!60!white] (B) circle (.2);
    \node[black!60!white]  at ($(B) + (-.38, 0)$) {\textsc{b}};
    
    \draw[black!60!white] (C) circle (.2);
    \node[black!60!white]  at ($(C) + (-.38, 0)$) {\textsc{c}};
    
    \draw[black!60!white] ($(D1) + (0, .2)$) arc (90:270:.2);
    \draw[black!60!white] ($(D1) + (0, .2)$) -- ($(D2) + (0, .2)$);
    \draw[black!60!white] ($(D1) + (0, -.2)$) -- ($(D2) + (0, -.2)$);
    \draw[black!60!white] ($(D2) + (0, -.2)$) arc (-90:90:.2);
    \node[black!60!white]  at ($(D2) + (.4, 0)$) {\textsc{d}};
    
    \node[black!60!white]  at ($(l1) + (-.3, 0)$) {$z$};
\end{tikzpicture}

%% file: imgs/infoset_forest.tex
\begin{tikzpicture}[baseline=0pt,scale=1]
    \def\done{1.2}
    \def\dtwo{.6}
    \def\dleaf{.33}
    \def\dvert{-1.2}

    \node[decpt] (A) at (0, 0) {};
    \node[decpt] (D) at ($(\done,\dvert)$) {};
    \node[decpt] (B) at ($(-\done, \dvert)$) {};
    \node[decpt] (C) at ($(0, \dvert)$) {};


    \draw[semithick] (B) edge node[above,sloped,fill=white,inner sep=1.5pt]{\small$\succ$} (A);
    \draw[semithick] (A) edge node[above,sloped,fill=white,inner sep=1.5pt]{\small$\prec$} (C);
    \draw[semithick] (A) edge node[above,sloped,fill=white,inner sep=1.5pt]{\small$\prec$} (D);


    \node[black!60!white]  at ($(A) + (-.4, 0)$) {\textsc{a}};
    \node[black!60!white]  at ($(B) + (-.38, 0)$) {\textsc{b}};
    \node[black!60!white]  at ($(C) + (-.38, 0)$) {\textsc{c}};
    \node[black!60!white]  at ($(D) + (.4, 0)$) {\textsc{d}};
    
    \node[gray,anchor=north] at (0,-1.8) {$(\textsc{a} \prec \textsc{b},~~~\textsc{a} \prec \textsc{c},~~~\textsc{a} \prec \textsc{d})$};
\end{tikzpicture}

%% file: text/efce.tex
\subsection{Extensive-Form Correlated Equilibrium (EFCE)}

Extensive-form correlated equilibrium has been proposed by~\citet{von2008extensive} as the natural counterpart to (normal-form) correlated equilibrium in extensive-form games.
In an EFCE, before the beginning of the game the mediator draws a recommended action for each of the possible information sets that players may encounter in the game, according to some probability distribution defined over joint reduced normal-form strategies.
These recommendations are not immediately revealed to each player. Instead, the mediator incrementally reveals relevant action recommendations as players reach new information sets.
At any information set, the acting player is free to deviate from the recommended action, but doing so comes at the cost of future recommendations, which are no longer issued if the player deviates.
In an EFCE, the recommended behavior is incentive-compatible for each player, that is, no player is strictly better off ever deviating from any of the mediator’s recommended actions.


Multiple equivalent definitions of EFCE can be given. In this article, we follow the equivalent formulation given by~\citet{farina2019correlation} based on the concept of \emph{trigger agents} introduced by~\citet{Gordon08:No}~and~\citet{dudik2009sampling}. In what follows, we will assume that an extensive-form game has been fixed.
%

\begin{definition}[Trigger agent]
    Let $i\in[n]$ be a player, let $\hat\sigma = (j,a) \in \Seqs_*$, and let $\hat{\vec{\pi}} \in \Pure_j$.
    The $(\hat\sigma,\hat{\vec{\pi}})$-{\em trigger agent} is the agent that plays the game as Player~$i$ according to the following rules.
    \begin{itemize}
        \item If the trigger agent has never been recommended to play action $a$ at information set $j$, the trigger agent will follow whatever recommendation is issued by the mediator. 
        \item When the trigger agent reaches information set $j$ and is recommended to play action $a$, we say that the trigger agent ``gets triggered'' by the {\em trigger sequence} $\hat\sigma = (j,a)$. This means that, from that point on, the trigger agent will disregard the recommendations and play according to the \emph{continuation strategy} $\hat{\vec{\pi}}$ from information set $j$ onward (that is, at $j$ and all of its descendant information sets). 
    \end{itemize}
\end{definition}
An EFCE is a probability distribution $\vec{\mu} \in \Delta^{|\Pi|}$ over $\Pi$ such that for any player $i \in [n]$, trigger sequence $\hat{\sigma} = (j,a) \in \Sigma_*^{(i)}$, and continuation strategy $\hat{\vec{\pi}} \in \Pure_j$, the expected utility of the $(\hat\sigma,\hat{\vec{\pi}})$-trigger agent is \emph{not} strictly greater than the expected utility that Player~$i$ would obtain by always following all of the mediator's recommendations.

In order to turn the above condition into an analytic expression, it is useful to introduce the following additional quantities.
Given a distribution $\vec{\mu} \in \Delta^{|\Pi|}$, we let $r_{\vec{\mu}}(z)$ be the probability that the game ends in terminal node $z \in \Z$ when all players follow recommendations issued by the mediator according to $\vec{\mu}$; in particular, for every $z \in \Z$, it holds:
\[
r_{\vec{\mu}}(z) \defeq \sum_{\substack{ (\pure[1],\dots,\pure[n]) \in \Pi \\ \pure[i][\sigma^{(i)}(z)] = 1 \,\,\, \forall i \in [n] }} \vec{\mu}[(\pure[1],\dots,\pure[n])],
\]
where the summation is over all joint strategies $(\pure[1],\dots,\pure[n]) \in \Pi$ such that terminal node $z$ is reachable when each player $i \in [n]$ plays according to $\pure[i]$.
%
%
%
Additionally, given a trigger sequence $\hat\sigma = (j,a) \in \Seqs_*$ for a player $i \in [n]$ and a continuation strategy $\hat{\vec{\pi}} \in \Pure_j$, we let $r^{(i)}_{\vec{\mu}, \hat{\sigma} \to \hat{\vec{\pi}}}(z)$ be the probability with which the $(\hat\sigma,\hat{\vec{\pi}})$-trigger agent reaches terminal node $z$. In particular, for every terminal node $z \in \Z$ such that $\sigma^{(i)}(z) \succeq j$ it holds that:
\[
r^{(i)}_{\vec{\mu},\,\hat{\sigma} \to \hat{\vec{\pi}}}(z) \defeq \left( \sum_{\substack{ (\pure[1],\dots,\pure[n]) \in \Pi \\ \pure[i'][\sigma^{(i')}(z)] = 1 \,\,\, \forall i'\neq i \\ \pure[i][\hat\sigma] = 1  }} \vec{\mu}[(\pure[1],\dots,\pure[n])] \right) \hat{\vec{\pi}} [ \sigma^{(i)}(z) ] .
\]

We can now state the formal definition of EFCE and approximate EFCE.

\begin{definition}[$\epsilon$-EFCE; EFCE]\label{def:efce}
Given $\epsilon \geq 0$, a probability distribution $\vec{\mu} \in \Delta^{|\Pi|}$ is an \emph{$\epsilon$-approximate EFCE} (or $\epsilon$-EFCE for short) if, for every player $i \in [n]$, trigger sequence $\hat\sigma = (j,a) \in \Seqs_*$, and continuation strategy $\hat{\vec{\pi}} \in \Pure_j$, the expected utility of the $(\hat\sigma,\hat{\vec{\pi}})$-trigger agent is larger than the expected utility that Player~$i$ would obtain by always following all of the mediator's recommendations by at most an amount $\epsilon$. In symbols, 
\begin{equation*}
    \sum_{\substack{z \in Z \\ \sigma^{(i)}(z) \succeq  \hat\sigma}} u^{(i)}(z)\, p_c(z)\, r_{\vec{\mu}}(z) \ge \sum_{\substack{z \in Z\\\sigma^{(i)}(z) \succeq j}} u^{(i)}(z)\, p_c(z)\, r^{(i)}_{\vec{\mu},\, \hat\sigma \to \hat{\vec{\pi}}}(z) - \epsilon.
\end{equation*}
A probability distribution $\vec{\mu} \in \Delta^{|\Pi|}$ is an EFCE if it is a $0$-EFCE.
\end{definition}


%% file: text/phi_regret.tex
\subsection{Regret Minimization and Phi-Regret Minimization}\label{sec:phirm}

In this article we will make heavy use of a mathematical object---one of the core abstractions in the field of online optimization---called a \emph{regret minimizer}.

\begin{definition}\label{def:regret minimizer}
Let $\cX$ be a set. A \emph{regret minimizer for $\cX$} is an abstract model for a decision maker that repeatedly interacts with a black-box environment. At each time $t$, the regret minimizer supports two operations:
\begin{itemize}
    \item \textsc{NextElement} has the effect that the regret minimizer will output an element $\vec{x}^{t} \in \cX$;
    \item $\textsc{ObserveUtility}(\ell^t)$ provides the environment's feedback to the regret minimizer, in the form of a linear utility function $\ell^t : \cX \to \bbR$ that evaluates how good the last-output point $\vec{x}^t$ was.
    The utility function can depend adversarially on the outputs $\vec{x}^1, \dots, \vec{x}^{t-1}$ of the regret minimizer, but not on $\vec{x}^t$.%
\end{itemize}
\end{definition}

Calls to $\textsc{NextElement}$ and $\textsc{ObserveUtility}$ keep alternating to each other: first, the regret minimizer will output a point $\vec{x}^1$, then it will received feedback $\ell^1$ from the environment, then it will output a new point $\vec{x}^2$, and so on.
The decision making encoded by the regret minimizer is \emph{online}, in the sense that at each time $t$, the output of the regret minimizer can depend on the prior outputs $\vec{x}^1, \dots,\vec{x}^{t-1}$ and corresponding observed utility functions $\ell^1,\dots,\ell^{t-1}$, but no information about future losses is available. The objective for the regret minimizer is to output points so that the \emph{cumulative regret} (or simply \emph{regret})
\begin{equation}\label{eq:cum regret}
    R^T \defeq \max_{\vec{x}^*\in\cX} \sum_{t=1}^T \Big( \ell^t(\vec{x}^*) - \ell^t(\vec{x}^t) \Big)
\end{equation}
grows asymptotically sublinearly in the time $T$.

Many regret minimizers that guarantee a cumulative regret $R^T = O(\sqrt{T})$ at all times $T$ for any convex and compact set $\cX$ are known in the literature (see, \emph{e.g.}, \citet{Cesa-Bianchi06:Prediction}).

A \emph{phi-regret minimizer} is an extension of the concept of a regret minimizer introduced by \citet{Stoltz07:Learning}, building on previous work by \citet{Greenwald03:General}.

\begin{definition}\label{def:phi regret minimizer}
Given a set $\cX$ of points and a set $\Phi$ of linear transformations $\phi:\cX\to\cX$, a \emph{phi-regret minimizer relative to $\Phi$ for the set $\cX$}---abbreviated in the term \emph{``$\Phi$-regret minimizer''}---is an object with the same semantics and operations of a regret minimizer, but whose quality metric is its \emph{cumulative phi-regret relative to $\Phi$} (or simply \emph{phi-regret relative to $\Phi$}, or \emph{$\Phi$-regret} for short)
\begin{equation}\label{eq:cum phi regret}
    R^T \defeq \max_{\phi^* \in \Phi} \sum_{t=1}^T \Big( \ell^t(\phi^*(\vec{x}^t)) - \ell^t(\vec{x}^t) \Big),
\end{equation}
instead of its cumulative regret. Once again, the goal for a phi-regret minimizer is to guarantee that its phi-regret grows asymptotically sublinearly as time $T$ increases.
\end{definition}

In the special case of the set of constant transformations $\Phi^\text{const} \defeq \{\cX \ni \vec{x} \mapsto \hat{\vec{x}}: \hat{\vec{x}}\in\cX\}$, the definition of cumulative phi-regret~\eqref{eq:cum phi regret} reduces to that of cumulative regret given in~\eqref{eq:cum regret}. So, a regret minimizer is a special case of a phi-regret minimizer.

A general construction by~\citet{Gordon08:No} gives a way to construct a $\Phi$-regret minimizer for $\cX$ starting from any regret minimizer (in the sense of \cref{def:regret minimizer}) for the set of functions $\Phi$. Specifically, let $\cR_\Phi$ be a 
regret minimizer for the set of transformations $\Phi$ whose cumulative regret grows sublinearly, and assume that every $\phi\in\Phi$ admits a fixed point $\vec{x} = \phi(\vec{x})$. Then, a $\Phi$-regret minimizer $\cR$ can be constructed starting from $\cR_\Phi$ as follows:
\begin{itemize}
    \item Each call to $\cR.\textsc{NextElement}$ first calls \textsc{NextElement} on $\cR_\Phi$ to obtain the next transformation $\phi^t$. Then, a fixed point $\vec{x}^t = \phi^t(\vec{x}^t)$ is returned.
    \item Each call to $\cR.\textsc{ObserveUtility}(\ell^t)$ with linear utility function $\ell^t$ constructs the linear utility function $L^t: \phi \mapsto \ell^t(\phi(\vec{x}^t))$, where $\vec{x}^t$ is the last-output strategy, and passes it to $\cR_\Phi$ by calling $\cR_\Phi.\textsc{ObserveUtility}(L^t)$.
\end{itemize}
The proof of correctness of the above construction is deceptively simple, and we recall it next.
Since $\cR_\Phi$ outputs transformations $\phi^1,\phi^2,\dots$ and receives utilities $\phi \mapsto \ell^1(\phi(\vec{x}^1)), \phi\mapsto\ell^2(\phi(\vec{x}^2)), \dots$, its cumulative regret $R_\Phi^T$ is
\[
    R_\Phi^T = \max_{\phi^*\in\Phi} \sum_{t=1}^T \Big( \ell^t(\phi^*(\vec{x}^t)) - \ell^t(\phi^t(\vec{x}^t)) \Big).
\]
Now, since $\vec{x}^t$ is a fixed point of $\phi^t$, $\phi^t(\vec{x}^t) = \vec{x}^t$, and therefore we can write
\begin{align}
    R_\Phi^T = \max_{\phi^*\in\Phi} \sum_{t=1}^T \Big( \ell^t(\phi^*(\vec{x}^t)) - \ell^t(\vec{x}^t) \Big),
\end{align}
where the right-hand side is exactly the $\Phi$-regret $R^T$ cumulated by $\cR$, as defined in~\eqref{eq:cum phi regret}. So, because the regret cumulated by $\cR_\Phi$ grows sublinearly by hypothesis, then so does the $\Phi$-regret cumulated by $\cR$.

%% file: text/trigger_regret.tex
\section{Trigger Regret and Relationship with EFCE}\label{sec:trigger regret}

In this section, we introduce the notion of \emph{trigger deviation function}, building on an idea by \citet[Section 3]{Gordon08:No}. We also introduce a connected notion of \emph{trigger regret minimization}, which is an instance of phi-regret minimization as recalled in \cref{sec:phirm}. The central result of this section, \cref{thm:empirical efce}, establishes a formal connection between EFCE and agents that minimize their trigger regret, thereby extending and generalizing the classic connection between correlated equilibrium and no-internal-regret in normal-form games~\citep{hart2000simple} to the extensive-form game counterpart. 

\begin{definition}[Trigger deviation function]\label{def:tdev}
Let $\hat\sigma = (j,a) \in \Seqs_*$, and $\hatpure[]\in\Pure_j$. We call ``\emph{trigger deviation function corresponding to trigger $\hat{\sigma}$ and continuation strategy $\hatpure[]$}'', any linear function $f : \bbR^{|\Seqs|}\to\bbR^{|\Seqs|}$ whose effect on deterministic sequence-form strategies is as follows:
\begin{itemize}
    \item all strategies $\pure[]\in\Pure$ that do not prescribe the sequence $\hat\sigma$ are left unmodified. In symbols,
    \begin{equation}\label{eq:trig rule A}
        f(\pure[]) = \pure[] \quad\qquad \forall\ \pure[] \in \Pure: \pure[]{}[\hat\sigma] = 0;
    \end{equation}
    \item all strategies $\pure[] \in\Pure$ that prescribe sequence $\hat\sigma = (j,a)$ are modified so that the behavior at $j$ and all of its descendants is replaced with the behavior prescribed by the continuation strategy $\hatpure[]$. In symbols,
    \begin{equation}\label{eq:trig rule B}
        f(\pure[])[\sigma] = \begin{cases}
            \pure[]{}[\sigma] & \text{if }\sigma \not\succeq j\\
            \hatpure[]{}[\sigma] & \text{if }\sigma\succeq j,\\
        \end{cases} \quad\qquad \forall\ \sigma \in \Seqs, \pure[] \in \Pure: \pure[]{}[\hat\sigma] = 1.
    \end{equation}
\end{itemize}
\end{definition}

At this stage, it is technically unclear whether a linear function that satisfies \cref{def:tdev} exists for all valid choices of $\hat\sigma$ and $\hatpure[]$. We show that this is indeed the case, by explicitly exhibiting a linear function, which we call the \emph{canonical trigger deviation function}. We start with a definition.

\begin{definition}\label{def:M}
    Let $\hat\sigma = (j,a)\in\Seqs_*$ and $\vec{y}\in\bbR_{\ge0}^{|\Seqs_j|}$. We denote with $\Mdev[i][\hat\sigma][\vec{y}] \in \bbR_{\ge0}^{|\Seqs|\times|\Seqs|}$ the matrix whose entries are defined as
    \begin{equation}\label{eq:Mdev}
        \Mdev[i][\hat\sigma][\vec{y}]{}[\sigma_r, \sigma_c] = \begin{cases}
            1 & \text{if } \sigma_c \not\succeq \hat\sigma \text{ and } \sigma_r = \sigma_c \\
            \vec{y}[\sigma_r] & \text{if } \sigma_c = \hat\sigma \text{ and } \sigma_r \succeq j\\
            0 & \text{otherwise},
        \end{cases}
        \qquad\qquad \forall\ \sigma_r,\sigma_c \in \Seqs.
    \end{equation}
    Furthermore, we denote with the symbol $\tdev[i][\hat\sigma][\vec{y}]$ the linear map
    $
        \bbR^{|\Seqs|}\ni\vec{x} \mapsto \Mdev[i][\hat\sigma][\vec{y}]\,\vec{x}.
    $
\end{definition}

In the following, we will focus on trigger deviation functions defined through the linear mapping of~\cref{eq:Mdev}. We call such deviation functions \emph{canonical trigger deviation functions}. For every player $i\in[n]$, we define $\Ph$ to be the set of all canonical trigger deviation functions. Formally:

\begin{definition}\label{def:canonical tdev}
    Let $\hat\sigma = (j,a) \in \Seqs_*$ and $\hatpure[]\in\Pure_j$. The function $\tdev$ is called
    the ``\emph{canonical trigger deviation function corresponding to trigger $\hat\sigma$ and continuation strategy $\hatpure[]$}''. Furthermore, the set of all canonical trigger deviation functions is denoted with the symbol
    \[
        \Ph \defeq \mleft\{\tdev : \hat\sigma = (j,a)\in\Seqs_*, \hatpure[]\in\Pure_j \mright\}.
    \]
\end{definition}
\begin{lemma}\label{lem:canonical}
    For any $\hat\sigma = (j,a) \in \Seqs_*$ and $\hatpure[]\in\Pure_j$, the linear function $\tdev$ as defined in \cref{def:canonical tdev} is a trigger deviation function in the sense of \cref{def:tdev}.
\end{lemma}
\begin{proof}
    The proof just amounts to a simple application of several definitions.
    Let $\pure[] \in \Pure$ be an arbitrary deterministic sequence-form strategy. By expanding the matrix-vector multiplication $\Mdev\pure[]$ using the definition~\eqref{eq:Mdev}, we obtain that for all $\sigma\in\Seqs$
    \begin{align*}
        (\Mdev\,\pure[])[\sigma] &= \pure[][\sigma]\bbone[\sigma \not\succeq\hat\sigma] + \hatpure[][\sigma] \pure[][\hat\sigma] \bbone[\sigma\succeq j].\numberthis{eq:matrix vector}
    \end{align*}
    There are only two possibilities:
    \begin{itemize}
        \item If $\pure[][\hat\sigma] = 0$, then~\eqref{eq:matrix vector} simplifies to
        \[
            (\Mdev\,\pure[])[\sigma] = \begin{cases}
                \pure[][\sigma] & \text{if } \sigma\not\succeq\hat\sigma\\
                0 & \text{otherwise}.
            \end{cases}
        \]
        Since by case hypothesis the probability of the sequence of actions from the root of the game tree down to $\hat\sigma$ is zero, then necessarily the probability of any longer sequence of actions $\sigma \succeq \hat\sigma$ must be zero as well, that is $\pure[][\sigma] = 0$ for all $\sigma\succeq\hat\sigma$. So, $\Mdev\,\pure[] = \pure[]$ and~\eqref{eq:trig rule A} holds.
        
        \item Conversely, assume $\pure[][\hat\sigma]=1$. This means that at information set $j \in \cJ$ action $a$ is selected (with probability $1$), and therefore $\pure[][\sigma] = 0$ for all $\sigma = (j,a'): a' \in A_j, a' \neq a$. This means that $\pure[][\sigma]\bbone[\sigma\not\succeq\hat\sigma] = \pure[][\sigma]\bbone[\sigma\not\succeq j]$ for all $\sigma\in\Seqs$. Substituting that equality into~\eqref{eq:matrix vector} gives Equation~\eqref{eq:trig rule B}, as we wanted to show.\qedhere
    \end{itemize} 
\end{proof}

In the following example we show how linear mappings of canonical trigger deviation functions operate. In particular, we show how they modify some deterministic sequence-form strategies on a simple extensive-form game.

\begin{example}\label{ex:deviations}
We build on the small extensive-form game of~\cref{fig:preliminary_example}, and the sequence-form strategies defined in~\cref{ex:preliminary}, to provide some concrete intuition behind canonical trigger deviation functions as defined in~\cref{def:canonical tdev}.

\begin{figure}[t]
    \begin{tabular}{p{4.25cm}|p{4.25cm}|p{4.25cm}}
    \toprule
        \hfil $\vec{M}_a$ \hfil
    & 
        \hfil $\vec{M}_b$ \hfil
    &
        \hfil $\vec{M}_c$ \hfil
    \\[1mm]
        \hfil Trigger sequence: \seq{1} \hfil 
    &
        \hfil Trigger sequence: \seq{2}\hfil 
    &
        \hfil Trigger sequence: \seq{3}\hfil 
    \\
        \hfil Continuation: \seq{2}, \seq{7} \hfil 
    &
        \hfil Continuation: \seq{1}, \seq{3}, \seq{5}\hfil 
    &
        \hfil Continuation: \seq{4}\hfil 
    \\
    \midrule
        \!\!\begin{tikzpicture}[baseline=-\the\dimexpr\fontdimen22\textfont2\relax ]
            \tikzset{every left delimiter/.style={xshift=1.5ex},
                     every right delimiter/.style={xshift=-1ex}};
            \matrix [matrix of math nodes,left delimiter=(,right delimiter=),row sep=.007cm,column sep=.007cm,color=black!25](m)
            {
            |[black]| 1 &           0 & 0 & 0 & 0 & 0 & 0 & 0 & 0 \\
                      0 & |[black]| 0 & 0 & 0 & 0 & 0 & 0 & 0 & 0 \\
                      0 & |[black]| 1 & |[black]| 1 & 0 & 0 & 0 & 0 & 0 & 0 \\
                      0 &           0 & 0 & |[black]| 0 & |[black]| 0 & |[black]| 0 & |[black]| 0 & 0 & 0 \\
                      0 &           0 & 0 & |[black]| 0 & |[black]| 0 & |[black]| 0 & |[black]| 0 & 0 & 0 \\
                      0 &           0 & 0 & |[black]| 0 & |[black]| 0 & |[black]| 0 & |[black]| 0 & 0 & 0 \\
                      0 &           0 & 0 & |[black]| 0 & |[black]| 0 & |[black]| 0 & |[black]| 0 & 0 & 0 \\
                      0 & |[black]| 1 & 0 & 0 & 0 & 0 & 0 & |[black]| 1 & 0 \\
                      0 & |[black]| 0 & 0 & 0 & 0 & 0 & 0 & 0 & |[black]| 1 \\
            };
            
            \node[left=3pt of m-1-1] (left-0) {\small$\emptyseq$};
            \node[left=3pt of m-2-1] (left-1) {\small\seq{1}};
            \node[left=3pt of m-3-1] (left-2) {\small\seq{2}};
            \node[left=3pt of m-4-1] (left-3) {\small\seq{3}};
            \node[left=3pt of m-5-1] (left-4) {\small\seq{4}};
            \node[left=3pt of m-6-1] (left-5) {\small\seq{5}};
            \node[left=3pt of m-7-1] (left-6) {\small\seq{6}};
            \node[left=3pt of m-8-1] (left-7) {\small\seq{7}};
            \node[left=3pt of m-9-1] (left-8) {\small\seq{8}};
            
            \node[above=3pt of m-1-1] (top-0) {\small$\emptyseq$};
            \node[above=3pt of m-1-2] (top-1) {\small\seq{1}};
            \node[above=3pt of m-1-3] (top-2) {\small\seq{2}};
            \node[above=3pt of m-1-4] (top-3) {\small\seq{3}};
            \node[above=3pt of m-1-5] (top-4) {\small\seq{4}};
            \node[above=3pt of m-1-6] (top-5) {\small\seq{5}};
            \node[above=3pt of m-1-7] (top-6) {\small\seq{6}};
            \node[above=3pt of m-1-8] (top-7) {\small\seq{7}};
            \node[above=3pt of m-1-9] (top-8) {\small\seq{8}};
            
            \begin{pgfonlayer}{back}
            \fhighlight[black!10!white]{m-4-4}{m-7-7}
            \fhighlight[black!30!white]{m-2-2}{m-3-2}
            \fhighlight[black!25!white]{m-8-2}{m-9-2}
            \end{pgfonlayer}
        \end{tikzpicture}
    &%
        \!\!\begin{tikzpicture}[baseline=-\the\dimexpr\fontdimen22\textfont2\relax ]
            \tikzset{every left delimiter/.style={xshift=1.5ex},
                     every right delimiter/.style={xshift=-1ex}};
        \matrix [matrix of math nodes,left delimiter=(,right delimiter=),row sep=.007cm,column sep=.007cm,color=black!25](m)
        {
        |[black]| 1 & 0 & 0 & 0 & 0 & 0 & 0 & 0 & 0 \\
                  0 & |[black]| 1 & |[black]| 1 & 0 & 0 & 0 & 0 & 0 & 0 \\
                  0 & 0 & |[black]| 0 & 0 & 0 & 0 & 0 & 0 & 0 \\
                  0 & 0 & |[black]| 1 & |[black]| 1 & 0 & 0 & 0 & 0 & 0 \\
                  0 & 0 & |[black]| 0 & 0 & |[black]| 1 & 0 & 0 & 0 & 0 \\
                  0 & 0 & |[black]| 1 & 0 & 0 & |[black]| 1 & 0 & 0 & 0 \\
                  0 & 0 & |[black]| 0 & 0 & 0 & 0 & |[black]| 1 & 0 & 0 \\
                  0 & 0 & 0 & 0 & 0 & 0 & 0 &  |[black]| 0 & |[black]| 0 \\
                  0 & 0 & 0 & 0 & 0 & 0 & 0 & |[black]| 0 & |[black]| 0 \\
        };
        
        \node[left=3pt of m-1-1] (left-0) {\small$\emptyseq$};
        \node[left=3pt of m-2-1] (left-1) {\small\seq{1}};
        \node[left=3pt of m-3-1] (left-2) {\small\seq{2}};
        \node[left=3pt of m-4-1] (left-3) {\small\seq{3}};
        \node[left=3pt of m-5-1] (left-4) {\small\seq{4}};
        \node[left=3pt of m-6-1] (left-5) {\small\seq{5}};
        \node[left=3pt of m-7-1] (left-6) {\small\seq{6}};
        \node[left=3pt of m-8-1] (left-7) {\small\seq{7}};
        \node[left=3pt of m-9-1] (left-8) {\small\seq{8}};
        
        \node[above=3pt of m-1-1] (top-0) {\small$\emptyseq$};
        \node[above=3pt of m-1-2] (top-1) {\small\seq{1}};
        \node[above=3pt of m-1-3] (top-2) {\small\seq{2}};
        \node[above=3pt of m-1-4] (top-3) {\small\seq{3}};
        \node[above=3pt of m-1-5] (top-4) {\small\seq{4}};
        \node[above=3pt of m-1-6] (top-5) {\small\seq{5}};
        \node[above=3pt of m-1-7] (top-6) {\small\seq{6}};
        \node[above=3pt of m-1-8] (top-7) {\small\seq{7}};
        \node[above=3pt of m-1-9] (top-8) {\small\seq{8}};
        
        \begin{pgfonlayer}{back}
        \fhighlight[black!10!white]{m-8-8}{m-9-9}
        \fhighlight[black!25!white]{m-2-3}{m-3-3}
        \fhighlight[black!25!white]{m-4-3}{m-7-3}
        \end{pgfonlayer}
        \end{tikzpicture}\
    &%
        \!\!\begin{tikzpicture}[baseline=-\the\dimexpr\fontdimen22\textfont2\relax ]
            \tikzset{every left delimiter/.style={xshift=1.5ex},
                     every right delimiter/.style={xshift=-1ex}};
        \matrix [matrix of math nodes,left delimiter=(,right delimiter=),row sep=.007cm,column sep=.007cm,color=black!25](m)
        {
        |[black]| 1 & 0 & 0 & 0 & 0 & 0 & 0 & 0 & 0 \\
        0 & |[black]| 1 & 0 & 0 & 0 & 0 & 0 & 0 & 0 \\
        0 & 0 & |[black]| 1 & 0 & 0 & 0 & 0 & 0 & 0 \\
        0 & 0 & 0 & |[black]| 0 & 0 & 0 & 0 & 0 & 0 \\
        0 & 0 & 0 & |[black]| 1 & |[black]| 1 & 0 & 0 & 0 & 0 \\
        0 & 0 & 0 & 0 & 0 & |[black]| 1 & 0 & 0 & 0 \\
        0 & 0 & 0 & 0 & 0 & 0 & |[black]| 1 & 0 & 0 \\
        0 & 0 & 0 & 0 & 0 & 0 & 0 &  |[black]| 1 & 0 \\
        0 & 0 & 0 & 0 & 0 & 0 & 0 & 0 & |[black]| 1 \\
        };
        
        \node[left=3pt of m-1-1] (left-0) {\small$\emptyseq$};
        \node[left=3pt of m-2-1] (left-1) {\small\seq{1}};
        \node[left=3pt of m-3-1] (left-2) {\small\seq{2}};
        \node[left=3pt of m-4-1] (left-3) {\small\seq{3}};
        \node[left=3pt of m-5-1] (left-4) {\small\seq{4}};
        \node[left=3pt of m-6-1] (left-5) {\small\seq{5}};
        \node[left=3pt of m-7-1] (left-6) {\small\seq{6}};
        \node[left=3pt of m-8-1] (left-7) {\small\seq{7}};
        \node[left=3pt of m-9-1] (left-8) {\small\seq{8}};
        
        \node[above=3pt of m-1-1] (top-0) {\small$\emptyseq$};
        \node[above=3pt of m-1-2] (top-1) {\small\seq{1}};
        \node[above=3pt of m-1-3] (top-2) {\small\seq{2}};
        \node[above=3pt of m-1-4] (top-3) {\small\seq{3}};
        \node[above=3pt of m-1-5] (top-4) {\small\seq{4}};
        \node[above=3pt of m-1-6] (top-5) {\small\seq{5}};
        \node[above=3pt of m-1-7] (top-6) {\small\seq{6}};
        \node[above=3pt of m-1-8] (top-7) {\small\seq{7}};
        \node[above=3pt of m-1-9] (top-8) {\small\seq{8}};
        
        \begin{pgfonlayer}{back}
        \fhighlight[black!25!white]{m-4-4}{m-5-4}
        \end{pgfonlayer}
        \end{tikzpicture}
    \\
    \bottomrule
    \end{tabular}
    \caption{
    Matrices defining different canonical trigger deviation functions (\cref{def:canonical tdev}) for the simple extensive-form game of~\cref{fig:preliminary_example}.
    Entries highlighted in dark gray represent the entries of the matrix defined in the second case of~\cref{eq:Mdev}. Let $\hat\sigma=(j,a)\in\Seqs[1]_*$ be the trigger sequence of the trigger deviation function. All indices $(\sigma_r,\sigma_c)$ such that $\sigma_r,\sigma_c \succeq j$ are highlighted in light gray.
    }
    \label{fig:trigger_deviation_matrices}
\end{figure}

\begin{itemize}

\item First, let us consider the trigger deviation function $\phi_a \defeq \tdev[1][\textsc{1}][\hatpure[]_a]$ is such that the trigger sequence is $\hat\sigma=(\textsc{a}, \seq{1})$, and the continuation strategy $\hatpure[]_a$ is such that Player 1 plays action $\seq{2}$ at information set \textsc{a}, and subsequently sequence $\seq{7}$ at information set \textsc{d}. 
The matrix defining the corresponding linear map according to~\cref{def:M} and~\cref{def:canonical tdev} is depicted in~\cref{fig:trigger_deviation_matrices} (Left), and we denote it by $\vec{M}_a$. 
In order to understand how this linear mapping modifies sequence-form strategy vectors we provide some examples using deterministic sequence-form strategy vectors defined in~\cref{fig:seq_strategies}. First, we observe that any deterministic sequence-form strategy choosing action $\seq{1}$ with probability $1$ triggers a deviation which follows the continuation strategy $\hatpure[]$. The deviation for those sequence-form strategies results in a final deterministic sequence-form strategy equal to $\vec{q}_{\seq{27}}$. For example, using some of the deterministic-sequence form strategies of~\cref{fig:seq_strategies}, we have the following: 
\[
\vec{M}_a \vec{q}_{\seq{135}} = \vec{M}_a \vec{q}_{\seq{136}} = \vec{M}_a \vec{q}_{\seq{145}} =  \vec{q}_{\seq{27}}.
\]
On the other hand, deterministic sequence-form strategies that do not select sequence $\seq{1}$ are left unmodified by the linear mapping. For instance,
\[
\vec{M}_a \vec{q}_{\seq{28}} = \vec{q}_{\seq{28}} \hspace{.2cm} \text{\normalfont and }\hspace{.2cm} \vec{M}_a \vec{q}_{\seq{27}} = \vec{q}_{\seq{27}}.
\]

\item Second, we examine the trigger deviation function $\phi_b \defeq \tdev[1][\seq{2}][\hatpure[]_b]$ for trigger sequence $\hat\sigma=(\textsc{a}, \seq{2})$, where the continuation strategy $\hatpure[]_b$ is defined so that Player 1 plays action $\seq{1}$ at information set \textsc{a}, sequence $\seq{3}$ at information set \textsc{b}, and action $\seq{5}$ at information set \textsc{c}.
The corresponding linear mapping is denoted by $\vec{M}_b$ and is reported in~\cref{fig:trigger_deviation_matrices} (Center).
As in the previous case, all deterministic sequence-form strategy vectors which choose action $\seq{2}$ at information set $\textsc{a}$ with probability $1$ are modified so that be strategy at \textsc{a} and its descendants \textsc{b,c,d} matches the continuation strategy. 
For example, we have that 
\[
\vec{M}_b \vec{q}_{\seq{27}} = \vec{M}_b \vec{q}_{\seq{28}} = \vec{q}_{\seq{135}}.
\]
Moreover, sequence-form strategies which do not trigger Player 1 on $\hat\sigma$ are left unchanged. This is the case for the following strategy vectors:
\[
\vec{M}_b \vec{q}_{\seq{136}} = \vec{q}_{\seq{136}} \hspace{.1cm}\text{\normalfont and }\hspace{.1cm} \vec{M}_b \vec{q}_{\seq{145}} = \vec{q}_{\seq{145}}.
\]

\item As a final example,~\cref{fig:trigger_deviation_matrices} (Right) reports the deviation matrix $\vec{M}_c$ corresponding to a trigger deviation function $\phi_c\defeq\tdev[1][\seq{3}][\hatpure[]_c]$ defined by trigger sequence $\hat\sigma=(\textsc{b}, \seq{3})$ and continuation strategy $\hatpure[]_c$ selecting action $\seq{4}$ at information set \textsc{b}. For instance, we have that $\vec{M}_c \vec{q}_{\seq{135}} = \vec{M}_c \vec{q}_{\seq{145}}$, and  $\vec{M}_c \vec{q}_{\seq{145}} = \vec{q}_{\seq{145}}$.

\end{itemize}
\end{example}

We are now ready to define the concept of trigger regret minimization, which extends and generalizes the homonymous notion in the conference version of this paper~\citep{Celli20:NoRegret}, as well as the notion of internal regret minimization in normal-form games. 

\begin{definition}\label{def:trigger regret}
    For every $i \in [n]$, we call \emph{trigger regret minimizer for player~$i$} any $\Ph$-regret minimizer for the set of deterministic sequence-form strategies $\Pure$.
\end{definition}

The following theorem shows that if each player $i\in[n]$ in the game plays according to a $\Ph$-regret minimizer, then the empirical frequency of play approaches the set of EFCEs.
%

\begin{theorem}\label{thm:empirical efce}
    For each player $i \in [n]$, let $\puret[i][1], \puret[i][2], \dots, \puret[i][T] \in \Pure$ be deterministic sequence-form strategies whose cumulative $\Ph$-regret with respect to the sequence of linear utility functions
    \begin{equation}\label{eq:def loss}
        \ell^{(i),\,t} : \Pure \ni \pure \mapsto \sum_{z \in Z}  u^{(i)}(z) \, p_c(z)\, \mleft(\prod_{i' \neq i} \puret[i']{}[\sigma^{(i')}(z)]\mright)\,\pure[i]{}[\sigma^{(i)}(z)]
    \end{equation}
    is $R^{(i),\,T}$. Then, the empirical frequency of play defined as the probability distribution $\vec{\mu} \in \Delta^{|\Pi|}$ that draws each joint profile $(\pure[1],\dots,\pure[n]) \in \Pi$ with probability
    \[
        \vec{\mu}[(\pure[1],\dots,\pure[n])] \defeq \frac{1}{T}\sum_{t=1}^T\bbone[(\puret[1],\dots,\puret[n]) = (\pure[1], \dots,\pure[n])] 
    \]
    is an $\epsilon$-EFCE, where $\epsilon \defeq \frac{1}{T}\max_{i\in[n]} R^{(i),\,T}$.
\end{theorem}
\begin{proof}
    It is immediate to check that $\vec{\mu}$ is indeed a valid element of the $|\Pi|$-simplex. Furthermore, the utility function $\ell^{(i),\,t}$ clearly satisfies the requirement of being independent on $\puret$, for all $i\in[n]$. We will show that $\vec{\mu}$ defines an $\epsilon$-EFCE by verifying that the definition holds (\cref{def:efce}). Fix any player $i \in [n]$, trigger sequence $\hat\sigma = (j,a) \in \Seqs_*$, and continuation strategy $\hat{\vec{\pi}} \in \Pure_j$. Since by hypothesis the cumulative $\Ph$-regret is upper bounded by $R^{(i),\,T}$, and $R^{(i),\,T} \le T\epsilon$ by definition of $\epsilon$, we must have
    \begin{equation*}
        T\epsilon \ge \sum_{t=1}^T \ell^{(i),\,t}\mleft(\tdev(\puret\hspace{.8pt})\mright) - \ell^{(i),\,t}\mleft(\puret\mright).
    \end{equation*}
    By expanding the definition of the utility function, which was given in~\eqref{eq:def loss}, the previous inequality is equivalent to
    \begin{equation}\label{eq:conv step1}
        T\epsilon \ge \sum_{t=1}^T \sum_{z\in \Z} \alp\cdot\mleft(\tdev(\puret\hspace{.8pt})[\sigma^{(i)}(z)] - \puret[i]{}[\sigma^{(i)}(z)]\mright),
    \end{equation}
    where we used the symbol
    \[
        \alp \defeq u^{(i)}(z) \, p_c(z)\, \mleft(\prod_{i' \neq i} \puret[i']{}[\sigma^{(i')}(z)]\mright)
    \]
    to lighten the notation.
    Since $\tdev$ is a trigger deviation function (\cref{lem:canonical}),~\eqref{eq:trig rule A} and~\eqref{eq:trig rule B} apply, and in particular it follows that
    \[
        \tdev(\puret{})[\sigma] = \puret{}[\sigma]
    \]
    for all $t = 1,\dots, T$ and $\sigma \not\succeq j$. So, the summation term in~\eqref{eq:conv step1} is zero for all terminal states $z \in \Z$ such that $\sigma^{(i)}(z) \not\succeq j$, and thus we can safely restrict the domain of the summation over terminal states $z \in \Z^{(i)}_j \defeq \{z \in \Z : \sigma^{(i)}(z) \succeq j\}$ only, obtaining
    \begin{equation}\label{eq:conv step2}
        T\epsilon \ge \sum_{t=1}^T \sum_{z\in \Z^{(i)}_j}\alp\cdot \mleft(\tdev(\puret\hspace{.8pt})[\sigma^{(i)}(z)] - \puret[i]{}[\sigma^{(i)}(z)]\mright).
    \end{equation}
    We now study the term $\tdev(\puret\hspace{.8pt})[\sigma^{(i)}(z)]$ for a generic $t\in\{1,\dots,T\}$ and $z \in \Z^{(i)}_j$, by splitting into cases contingent on the value of $\puret[i][t]{}[\hat\sigma]\in\{0,1\}$. If $\puret[i][t]{}[\hat\sigma] = 0$, then~\eqref{eq:trig rule A} applies, and therefore
    \[ 
        \tdev(\puret\hspace{.8pt})[\sigma^{(i)}(z)] - \puret[i]{}[\sigma^{(i)}(z)] = 0.
    \]
    If, on the contrary, $\puret[i][t]{}[\hat\sigma] = 1$, then~\eqref{eq:trig rule B} applies, and $\tdev(\puret\hspace{.8pt})[\sigma^{(i)}(z)] = \hat{\vec{\pi}}[\sigma^{(i)}(z)]$, where we used the fact that $\sigma^{(i)}(z) \succeq j$ by definition of $z \in \Z^{(i)}_j$. So, at all $t = 1,\dots, T$ and for all $z\in \Z^{(i)}_j$, it holds that
    \[
        \tdev(\puret\hspace{.8pt})[\sigma^{(i)}(z)] - \puret[i]{}[\sigma^{(i)}(z)] = \puret[i][t]{}[\hat\sigma]\mleft(\hat{\vec{\pi}}[\sigma^{(i)}(z)] - \puret[i]{}[\sigma^{(i)}(z)]\mright),
    \]
    and thus~\eqref{eq:conv step2} can be equivalently written as
    \begin{equation}\label{eq:conv step3}
        T\epsilon \ge \sum_{t=1}^T \sum_{z\in \Z^{(i)}_j}\!\! \puret[i][t]{}[\hat\sigma]\, \alp\cdot\mleft(\hat{\vec{\pi}}[\sigma^{(i)}(z)] - \puret[i]{}[\sigma^{(i)}(z)]\mright).   
    \end{equation}
    We now make the crucial observation that time $t$ appears in $\alpha^{(i),\,t}_z$ and~\eqref{eq:conv step3} only as a superscript in the strategies $\puret[1],\dots,\puret[n]$, and nowhere else. Therefore, by introducing the functions
    \begin{align}
      \alpha^{(i)}_{z}: \Pi \ni (\pure[1]\!,\dots,\pure[n]) &\mapsto u^{(i)}(z)\, p_c(z)\,\mleft(\prod_{i'\neq i}\pure[i']{}[\sigma^{(i')}(z)]\mright) \text{, and} \label{eq:def_alfa}\\
      v^{(i)}_{\hat\sigma\to\hat{\vec{\pi}}}: \Pi\ni \vec{\pi} = (\pure[1]\!,\dots,\pure[n]) &\mapsto\!\!\! \sum_{z\in \Z^{(i)}_j}\!\!\! \pure{}[\hat\sigma]\, \alpha^{(i)}_z(\vec{\pi})\cdot\mleft(\hat{\vec{\pi}}[\sigma^{(i)}(z)] - \pure{}[\sigma^{(i)}(z)]\mright),\label{eq:def v}
    \end{align}
    we can rewrite~\eqref{eq:conv step3} as 
    \begin{align*}
        T\epsilon &\ge \sum_{t=1}^T v^{(i)}_{\hat\sigma\to\hat{\vec{\pi}}}(\puret[1],\dots,\puret[n])
                  = \sum_{t=1}^T \sum_{\vec{\pi}\in\Pi} \bbone[(\puret[1],\dots,\puret[n]) = \vec{\pi}]\cdot v^{(i)}_{\hat\sigma\to\hat{\vec{\pi}}}(\vec{\pi})
                  \\&= \sum_{\vec{\pi}\in\Pi} \mleft(\sum_{t=1}^T \bbone[(\puret[1],\dots,\puret[n])=\vec{\pi}]\mright) v^{(i)}_{\hat\sigma\to\hat{\vec{\pi}}}(\vec{\pi})
                  = T\sum_{\vec{\pi} \in \Pi}\vec{\mu}[\vec{\pi}] v^{(i)}_{\hat\sigma\to\hat{\vec{\pi}}}(\vec{\pi}), \numberthis{eq:conv step4}
    \end{align*}
    where we used the definition of $\vec{\mu}$ in the third equality. Dividing by $T$ in~\eqref{eq:conv step4}, we can further write
    \[
        \epsilon \ge \sum_{\vec{\pi}\in\Pi} \vec{\mu}[\vec{\pi}]\, v^{(i)}_{\hat\sigma\to\hat{\vec{\pi}}}(\vec{\pi}).\numberthis{eq:conv step5}
    \]
    By expanding the definition of $v^{(i)}_{\hat\sigma\to\hat{\vec{\pi}}}$ in~\eqref{eq:conv step5},
    \begin{align*}
        \epsilon &\ge \sum_{\vec{\pi} =(\pure[1],\dots,\pure[n])\in\Pi} \mleft(\vec{\mu}[\vec{\pi}] \, \sum_{z\in \Z^{(i)}_j}\!\!\! \pure{}[\hat\sigma]\, \alpha^{(i)}_z(\vec{\pi})\cdot\mleft(\hat{\vec{\pi}}[\sigma^{(i)}(z)] - \pure{}[\sigma^{(i)}(z)]\mright)\mright)\\
            &= \sum_{z\in \Z^{(i)}_j}\sum_{\vec{\pi}=(\pure[1],\dots,\pure[n])\in\Pi} \vec{\mu}[\vec{\pi}]\,\pure{}[\hat\sigma]\, \alpha^{(i)}_z(\vec{\pi})\cdot\mleft(\hat{\vec{\pi}}[\sigma^{(i)}(z)] - \pure{}[\sigma^{(i)}(z)]\mright).\numberthis{eq:conv step6}
    \end{align*}
    The right-hand side of~\eqref{eq:conv step6} can be simplified further by noticing that, by definition of $\alpha^{(i)}_z(\vec{\pi})$, 
    \begin{align*}
        \pure[i][\hat\sigma] \alpha^{(i)}_z(\pure[1], \dots, \pure[n]) &= u^{(i)}(z)p_c(z) \pure[i][\hat\sigma]\mleft(\prod_{i'\neq i} \pure[i'][\sigma^{(i')}(z)]\mright)\\
            &= \begin{cases}
                u^{(i)}(z)p_c(z) & \text{if } \pure[i][\hat\sigma] = 1, \pure[i'][\sigma^{(i')}(z)] = 1 \,\,\forall i' \neq i\\
                0 & \text{otherwise}.
            \end{cases}
    \end{align*}
    Substituting the above expression into~\eqref{eq:conv step6}, we obtain
    \begin{align*}\allowdisplaybreaks
        \epsilon &\ge \sum_{z\in \Z^{(i)}_j}u^{(i)}(z)p_c(z) \mleft( \sum_{\substack{ (\pure[1],\dots,\pure[n]) \in \Pi \\ \pure[i'][\sigma^{(i')}(z)] = 1 \,\,\, \forall i'\neq i \\ \pure[i][\hat\sigma] = 1  }} \vec{\mu}[(\pure[1],\dots,\pure[n])] \cdot\mleft(\hat{\vec{\pi}}[\sigma^{(i)}(z)] - \pure{}[\sigma^{(i)}(z)]\mright)\mright) \\
        & = \sum_{z\in \Z^{(i)}_j}u^{(i)}(z)p_c(z) r^{(i)}_{\vec{\mu},\,\hat{\sigma} \to \hat{\vec{\pi}} }(z) - \sum_{z\in \Z^{(i)}_j}u^{(i)}(z)p_c(z) \mleft( \sum_{\substack{ (\pure[1],\dots,\pure[n]) \in \Pi \\ \pure[i'][\sigma^{(i')}(z)] = 1 \,\,\, \forall i' \in [n] \\ \pure[i][\hat\sigma] = 1  }} \vec{\mu}[(\pure[1],\dots,\pure[n])] \mright)\, ,
    \end{align*}
    where, in order to get the last equality, we used the definition of $r^{(i)}_{\vec{\mu},\,\hat{\sigma} \to \hat{\vec{\pi}}}$ and we dropped the factor $\pure{}[\sigma^{(i)}(z)] \in \{ 0,1 \}$ in the second summation by adding the condition $\pure[i][\sigma^{(i)}(z)] = 1$ to its domain.
    Notice that, by definition of $\Z^{(i)}_j$, the first summation above is exactly the term appearing in the left-hand-side of the inequality in the definition of $\epsilon$-EFCE (Definition~\ref{def:efce}).
    Moreover, since for any $(\pure[1],\dots,\pure[n]) \in \Pi$ it holds that $\pure[i][\sigma^{(i)}(z)] = 1$ and $\pure[i][\hat\sigma] = 1$ only for terminal nodes $z\in \Z^{(i)}_j$ such that $ \hat\sigma \preceq \sigma^{(i)}(z)$, we can restrict the domain of the second summation above to $z \in \Z : \sigma^{(i)}(z) \succeq \hat\sigma $ and equivalently rewrite it as
    \[
       \sum_{\substack{z \in \Z\\ \sigma^{(i)}(z) \succeq \hat\sigma }}u^{(i)}(z)p_c(z) \left( \sum_{\substack{ (\pure[1],\dots,\pure[n]) \in \Pi \\ \pure[i'][\sigma^{(i')}(z)] = 1 \,\,\, \forall i' \in [n] }} \vec{\mu}[(\pure[1],\dots,\pure[n])] \right) =  \sum_{\substack{z \in \Z\\ \sigma^{(i)}(z) \succeq \hat\sigma }}u^{(i)}(z)p_c(z) r_{\vec{\mu}}(z),
    \]
    which is exactly the first term appearing in the right-hand-side of the inequality in the definition of $\epsilon$-EFCE (Definition~\ref{def:efce}).
    Thus, we obtain that
    \begin{align*}
        \epsilon &\ge \sum_{\substack{ z \in \Z \\ \sigma^{(i)}(z) \succeq j }}u^{(i)}(z)p_c(z) r^{(i)}_{\vec{\mu},\,\hat{\sigma} \to \hat{\vec{\pi}} }(z) - \sum_{ \substack{z \in \Z\\ \sigma^{(i)}(z) \succeq \hat\sigma } }u^{(i)}(z)p_c(z) r_{\vec{\mu}}(z),
    \end{align*}
    for all $\hat\sigma = (j,a)\in\Seqs_*$ and $\hat{\vec{\pi}} \in \Pure_j$, which is the definition of $\vec{\mu}$ being an $\epsilon$-EFCE. 
\end{proof}

%% file: text/algorithm.tex
\section{Efficient No-Trigger-Regret Algorithm}\label{sec:no trigger regret algo}

\cref{thm:empirical efce} in \cref{sec:trigger regret} immediately implies that if all players $i\in[n]$ play according to the strategies output by a $\Ph$-regret minimizer for the set of deterministic sequence-form strategies $\Pure$, their
empirical frequency of play converges to an EFCE. Therefore, the existence of uncoupled no-regret learning dynamics that converge to EFCE can be proved constructively by showing that one such $\Ph$-regret minimizer can be constructed for each player $i \in [n]$. More precisely, in this section we seek to solve the following problem.

\begin{problem}\label{prob:pure}
    Given any player $i \in [n]$, construct a $\Ph$-regret minimizer for the set of the player's deterministic sequence-form strategies $\Pure$, such that:
    \begin{itemize}
        \item it is efficient: the \textsc{NextElement} and the \textsc{ObserveUtility} operations both run in polynomial time in the number $|\Seqs|$ of sequences of the player; and
        \item it guarantees low regret: after any $T$ observed linear utility functions and for any $\delta \in (0, 1)$, with probability at least $1-\delta$ the cumulative $\Ph$-regret is $O(\sqrt{T} + \sqrt{T\log (1/\delta)})$.
    \end{itemize}
\end{problem}

The central result of this section, \cref{thm:final}, provides a solution to \cref{prob:pure}.

\subsection{Overview}

Before delving into the details of the construction of our $\Ph$-regret minimizer for the set of deterministic sequence-form strategies $\Pure$ of a generic player $i \in [n]$, we give an overview of the main logical steps that we use to attack \cref{prob:pure}.
\begin{itemize}
    \item In \cref{sec:deterministic to mixed} we show that one can soundly move the attention from the set of \emph{deterministic} strategies $\Pure$ to the set of \emph{mixed} strategies $\Q = \co \Pure$. In particular, in the rest of the section we will seek to construct a $\Ph$-regret minimizer for the set $\Q$ (as opposed to $\Pure$) that guarantees sublinear regret in the worst case. 

    \item In \cref{sec:relaxing} we show that the convex hull $\co\Ph$ of the set of canonical trigger deviation functions possesses a combinatorial structure that can be leveraged to construct an efficient regret minimizer for it. 

    \item Finally, in \cref{sec:fixed point} we prove that given any $\phi \in \co\Ph$, there exists a fixed-point sequence-form strategy $\vec{q} \in \Q$ such that $\phi(\vec{q}) = \vec{q}$, and that such a fixed-point strategy can be found in polynomial time in the number of sequences $|\Seqs|$ of Player~$i$.
\end{itemize}

Together, the last two steps enable us to apply the construction by \citet{Gordon08:No} described in \cref{sec:phirm} to obtain an efficient $\Tph$-regret minimizer for the set of sequence-form strategies $\Q$ with worst-case sublinear regret guarantees. Since $\co\Ph \supseteq \Ph$, that $\Tph$-regret minimizer is also a $\Ph$-regret minimizer, and the construction is complete.

\input{text/deterministic_to_mixed}
\input{text/coph_rm}
\input{text/fixed_point}

%% file: text/deterministic_to_mixed.tex
\subsection{From Deterministic to Mixed Strategies}\label{sec:deterministic to mixed}

Suppose that a regret minimizer for a generic discrete set $\cX$ were sought, but only regret minimizers for the convex hull $\co\cX$ were known. It seems natural to wonder whether one could take any regret minimizer for $\co\cX$ and convert it into a regret minimizer for $\cX$, by sampling the outputs $\bar{\vec{x}}^t$ of the former using an unbiased estimator $\vec{x}^t \in \cX$, with $\bbE[\vec{x}^t] = \bar{\vec{x}}^t$. It is a folkore result, justified by a concentration argument, that this is indeed the case (see, for instance, \citep[page 192]{Cesa-Bianchi06:Prediction}). In particular, in the case of our interest where $\cX = \Pure$, the following can be shown.
\begin{lemma}\label{lem:deterministic to mixed}
    Let $i\in[n]$ be any player, and $\bar{\cR}^{(i)}$ be any $\Ph$-regret minimizer for the set $\Q$ of \emph{mixed} sequence-form strategies, guaranteeing $\bar{R}^{(i),\,T} = O(\sqrt{T})$ regret in the worst case upon observing $T$ linear utility functions. Consider the algorithm $\cR^{(i)}$ whose \textsc{NextElement} and \textsc{ObserveUtility} operations are defined as follows at all times $t$.
    \begin{itemize}
        \item $\cR^{(i)}.\textsc{NextElement}$ calls $\bar{\cR}^{(i)}.\textsc{NextElement}$, thereby obtaining a mixed sequence-form strategy $\qt \in \Q$. Then, an unbiased sampling scheme (such as the natural sampling scheme described in \cref{sec:extensive form}) is used to sample a random deterministic sequence-form strategy $\puret{} \in \Pure$ in linear time in the number of sequences $|\Seqs|$. Finally, $\puret{}$ is returned to the caller;
        \item $\cR^{(i)}.\textsc{ObserveUtility}(\ell^{t})$ calls $\bar{\cR}^{(i)}.\textsc{ObserveUtility}(\ell^t)$ with the same utility function $\ell^{t}$.
    \end{itemize}
    Furthermore, assume that the observed utility functions $\ell^1,\ell^2,\dots$ have range upper bounded by a constant $D$, that is, $\max_{\vec{q},\vec{q}' \in \Q} \{\ell^t(\vec{q}) - \ell^t(\vec{q}')\} \le D$ for all $t=1,\dots, T$.
    Then, $\cR^{(i)}$ is a $\Ph$-regret minimizer for the set of deterministic sequence-form strategies $\Pure$, whose cumulative $\Ph$-regret satisfies, at all times $T$ and for all $\delta\in(0,1)$, the inequality
    \[
        \bbP\mleft[R^{(i),\,T} \le \bar{R}^{(i),\,T} + D\sqrt{8T \log(1/\delta)}\mright] \ge 1-\delta.
    \]
\end{lemma}
\begin{proof}
    Let $\ell^1, \ell^2, \dots$ be the sequence of linear utility functions observed by $\cR^{(i)}$, and fix any $\phi \in \Ph$. We introduce the discrete-time stochastic process
    \begin{equation}\label{eq:wt}
        w^t \defeq \ell^t(\phi(\puret{})) - \ell^t(\puret) - \ell^t(\phi(\qt)) + \ell^t(\qt).
    \end{equation}
    Since i) $\ell^t$ and $\phi$ are both linear functions, ii) $\ell^t$ is independent on $\puret$, and iii) $\puret$ is an unbiased estimator of $\qt$ at all times $t$ by hypothesis, then $w^t$ is a martingale difference sequence. Furthermore, each increment $|w^t|$ can be easily upper bounded, at all times $t$, according to
    \begin{equation}\label{eq:wt bound}
        |w^t| \le |\ell^t(\phi(\puret{})) - \ell^t(\puret)| + |\ell^t(\phi(\qt)) - \ell^t(\qt)| \le 2D,
    \end{equation}
    where the second inequality follows from the fact that $\phi$ maps sequence-form strategies to sequence-form strategies, as well as the hypothesis that $\ell^t$ has range upper bounded by $D$.
    
    For any $T$, let $R^{(i),\,T}(\phi)$ and $\bar{R}^{(i),\,T}(\phi)$ denote the regret cumulated by $\cR^{(i)}$ and $\bar{\cR}^{(i)}$, respectively, compared to always picking transformation $\phi$; in symbols
    \[
        R^{(i),\,T}(\phi) \defeq \sum_{t=1}^T \ell^t(\phi(\puret)) - \ell^t(\puret), \qquad
        \bar{R}^{(i),\,T}(\phi) \defeq \sum_{t=1}^T \ell^t(\phi(\qt)) - \ell^t(\qt).
    \]
    It is immediate to see from definition~\eqref{eq:wt} of $w^t$, that
    \begin{equation}\label{eq:wt sum}
        \sum_{t=1}^T w^t = R^{(i),\,T}(\phi) - \bar{R}^{(i),\,T}(\phi)\qquad\quad\forall\ T \in \{1, 2, \dots\}.
    \end{equation}
    Hence, using the Azuma-Hoeffding concentration inequality\footnote{%
        We recall the classic Azuma-Hoeffding inequality~\citep{Hoeffding63:Probability,Azuma67:Weighted} for martingale difference sequences (\emph{e.g.,} \citep[Theorem 3.14]{McDiarmid98:Concentration}).
        \begin{lemma}[Azuma-Hoeffding inequality]
            \label{lem:azuma}
            Let $Y_1, \dots, Y_n$ be a martingale difference sequence with $a_k \le Y_k \le b_k$ for each $k$, for suitable constants $a_k, b_k$. Then for any $\tau \ge 0$,
        $
            \bbP\mleft[\sum Y_k \le \tau \mright] \ge 1 - e^{-2\tau^2 / \sum (b_k - a_k)^2}.
        $
        \end{lemma}
    }, it follows that, for all $T$,
    \begin{align*}
        \bbP\mleft[R^{(i),\,T}(\phi) - \bar{R}^{(i),\,T}(\phi) \le D\sqrt{8T \log\mleft(\frac{1}{\delta}\mright)}\mright] &= \bbP\mleft[\sum_{t=1}^T w^t \le D\sqrt{2T \log\mleft(\frac{1}{\delta}\mright)}\mright] \\
        &\ge 1 - \exp\mleft\{-\frac{2}{\sum_{t=1}^T |w^t|^2} \mleft(D\sqrt{8T \log\mleft(\frac{1}{\delta}\mright)}\mright)^2\mright\}\\
        &\ge 1 - \exp\mleft\{-\frac{2}{(4D)^2 T} \mleft(D\sqrt{8T \log\mleft(\frac{1}{\delta}\mright)}\mright)^2\mright\}
        = 1 - \delta,\numberthis{eq:one minus delta}
    \end{align*}
    where we used \eqref{eq:wt sum} in the equality and~\eqref{eq:wt bound} in the second inequality.
    Since the above analysis holds for any choice of $\phi\in\Ph$, it holds in particular for any $\phi^*$ that maximizes $R^{(i),\,T}(\phi)$, yielding
    \begin{align*}
        \bbP\mleft[R^{(i),\,T} \le \bar{R}^{(i),\,T} + D\sqrt{8T \log\mleft(\frac{1}{\delta}\mright)}\mright]
        &= 
        \bbP\mleft[R^{(i),\,T}(\phi^*) \le \bar{R}^{(i),\,T} + D\sqrt{8T \log\mleft(\frac{1}{\delta}\mright)}\mright]\\
        &\ge 
        \bbP\mleft[R^{(i),\,T}(\phi^*) \le \bar{R}^{(i),\,T}(\phi^*) + D\sqrt{8T \log\mleft(\frac{1}{\delta}\mright)}\mright] \ge 1-\delta,
    \end{align*}
    where the first inequality follows from the fact that $\bar{R}^{(i),\,T} = \max_{\phi\in\Ph}\bar{R}^{(i),\,T}(\phi) \ge \bar{R}^{(i),\,T}(\phi^*)$ and the second inequality follows from~\eqref{eq:one minus delta}.
\end{proof}

\cref{lem:deterministic to mixed} immediately implies that in order to solve \cref{prob:pure} it is enough to solve the following problem.

\begin{problem}\label{prob:mixed}
    Given any player $i\in[n]$, construct a $\Ph$-regret minimizer for the set of mixed sequence-form strategies $\Q$ such that:
    \begin{itemize}
        \item it is efficient: both the \textsc{NextElement} and the \textsc{ObserveUtility} operations can be implemented in polynomial time in $|\Seqs|$; and
        \item it guarantees low regret: after any $T$ observed utilities, the cumulative $\Ph$-regret is upper bounded as $O(\sqrt{T})$.
    \end{itemize}
\end{problem}

The remainder of this section gives an algorithm that solves \cref{prob:mixed} and, thus, indirectly also \cref{prob:pure}.

%% file: text/coph_rm.tex
\subsection{Regret Minimizer for the Convex Hull of the Set of Trigger Deviation Functions}\label{sec:relaxing}

In this subsection we begin the construction of a phi-regret minimizer relative to the convex hull $\co\Ph$ of the set of trigger deviation functions $\Ph$, for the set $\Q$. Since $\co\Ph \supseteq \Ph$, any such $(\co\Ph)$-regret minimizer is trivially also a $\Ph$-regret minimizer for $\Q$.

In order to obtain our $(\co\Ph)$-regret minimizer, we will leverage the general construction due to \citet{Gordon08:No} that we recalled at the end of \cref{sec:phirm}. In our particular case, that framework reduces to showing the following:
\begin{enumerate}
    \item existence of a regret minimizer for the set of deviations $\co\Ph$; and
    \item existence of a fixed point $\vec{q} = \phi(\vec{q})$ for any $\phi \in \co\Ph$.
\end{enumerate}

In this subsection we will focus on point (1), while in the next subsection we will focus on point (2). Specifically, the central result of this subsection, \cref{thm:analysis R i}, will constructively establish the existence of an efficient regret minimizer $\tilde{\mathcal{R}}^{(i)}$ for the set $\co\Ph$.

The starting point of our approach is the observation that, because the convex hull operation is associative, $\co\Ph = \co\{\tdev: \hat\sigma=(j,a)\in\Seqs_*, \hatpure[]\in\Pure_j\}$ can be evaluated in two stages: first, for each sequence $\hat\sigma = (j,a)\in\Seqs_*$ one can define the set
\begin{equation*}
    \Phj \defeq \co\mleft\{\tdev{}: \hat{\vec{\pi}} \in \Pure_{j} \mright\};
\end{equation*}
and, then, one can take the convex hull of all $\Phj$, that is,
\begin{equation}\label{eq:co expansion}
    \co\Ph = \co \mleft\{\Phj: \hat\sigma\in\Seqs_* \mright\}.
\end{equation}
Our construction of $\tilde{\mathcal{R}}^{(i)}$ will follow a similar structure. First, for each $\hat\sigma\in\Seqs_*$ we will construct a regret minimizer $\tilde{\mathcal{R}}^{(i)}_{\hat\sigma}$ for the set of deviations $\Phj$ (\cref{sec:rm ph sigma}). Then, we will combine all the regret minimizers $\tilde{\mathcal{R}}^{(i)}_{\hat\sigma}$ into a composite regret minimizer $\tilde{\mathcal{R}}^{(i)}$ for $\co\Ph$ (\cref{sec:rm ch}). 

\subsubsection{Regret Minimizer for $\Phj$}\label{sec:rm ph sigma}

Fix any $\hat\sigma=(j,a)\in\Seqs_*$. As we will show, a regret minimizer for the set $\Phj$ can be constructed starting from any regret minimizer for the set $\Q_j$. The crucial insight lies in the observation that the mapping
\[
    h^{(i)}_{\hat\sigma}: \bbR^{|\Seqs_j|} \ni \vec{y} \mapsto \tdev[i][\hat\sigma][{\vec{y}}]
\]
is affine, since the entries in $\Mdev[i][\hat\sigma][\vec{y}]$ are defined using only constants and linear combinations of entries in $\vec{y}$ (\cref{def:M}). Hence, using the properties of affine functions, we can write
\[
    \Phj = \co\mleft\{\tdev : \hatpure[]\in\Pure_j\mright\} = \co h^{(i)}_{\hat\sigma}(\Pure_j) = h^{(i)}_{\hat\sigma}(\co \Pure_j) =  h^{(i)}_{\hat\sigma}(\Q_j),
\]
which ultimately informs the following characterization of the set $\Phj$.
\begin{lemma}\label{lem:ph j structure}
    For all sequences $\hat\sigma = (j,a)\in\Seqs_*$, $\Phj$ is the image of $\Q_j$ under the affine mapping $h^{(i)}_{\hat\sigma}$. In symbols,
    \[
        \Phj = \mleft\{\tdev[i][\hat\sigma][\q_{\hat\sigma}] : \q_{\hat\sigma}\in\Q_j\mright\}.
    \]
\end{lemma}

As a consequence of \cref{lem:ph j structure}, given any $\hat\sigma=(j,a)\in\Seqs_*$, all transformations $\phi \in\Phj$ are of the form $\phi =\tdev[i][\hat\sigma][\q_{\hat\sigma}]$ for some $\q_{\hat\sigma}\in\Q_j$. Thus, the cumulative regret incurred by a generic sequence of transformations $\phi^1 = \tdev[i][\hat\sigma][\q_{\hat\sigma}^1], \dots, \phi^T = \tdev[i][\hat\sigma][\q_{\hat\sigma}^T]$ against linear utility functions $L^1, \dots, L^T$ can be written as
\begin{align*}
    \max_{\phi^* \in \Phj} \sum_{t=1}^T L^t(\phi^*) - L^t(\phi^t) &= \max_{\hat{\vec{q}}^* \in \Q_j} \sum_{t=1}^T L^t(\tdev[i][\hat\sigma][\hat{\vec{q}}^*]) - L^t(\tdev[i][\hat\sigma][\q_{\hat\sigma}^t])\\
    &= \max_{\hat{\vec{q}}^* \in \Q_j} \sum_{t=1}^T (L^t \circ h^{(i)}_{\hat\sigma})(\hat{\vec{q}}^*) - (L^t \circ h^{(i)}_{\hat\sigma})(\q_{\hat\sigma}^t).\numberthis{eq:regret equal}
\end{align*}
Since $L^t$ is linear and $h^{(i)}_{\hat\sigma}$ is affine, their composition $ L^t\circ h^{(i)}_{\hat\sigma}$ is affine, and therefore the shifted function
\[
    g^{(i),\,t}_{\hat\sigma} : \bbR^{|\Seqs_j|} \ni \vec{x} \mapsto 
    L^t(h^{(i)}_{\hat\sigma}(\vec{x})) - L^t(h^{(i)}_{\hat\sigma}(\vec{0}))
\]
is linear.\footnote{We shift $L^t\circ h^{(i)}_{\hat\sigma}$ purely for technical reasons. We do it so that $g^{(i),\,t}$ is a \textit{linear} utility function, and thus it can be passed in as feedback to a regret minimizer.} 
Furthermore, from~\eqref{eq:regret equal} it follows that
\begin{equation}\label{eq:regret transformation}
    \max_{\phi^* \in \Phj} \sum_{t=1}^T L^t(\phi^*) - L^t(\tdev[i][\hat\sigma][\q_{\hat\sigma}^t])
    = \max_{\hat{\vec{q}}^* \in \Q_j}\sum_{t=1}^T g^{(i),\,t}_{\hat\sigma}(\hat{\vec{q}}^*) - g^{(i),\,t}_{\hat\sigma}(\q_{\hat\sigma}^t).
\end{equation}

\cref{eq:regret transformation} suggests that if the continuation strategies $\q_{\hat\sigma}^t \in \Q_j$ are picked by a regret minimizer $\tilde{\cR}^{(i)}_{\cQ, \hat\sigma}$ that observes the linear utility functions $g^{(i),\,t}_{\hat\sigma}$ at all times $t$, the regret cumulated with respect to utility functions $L^t$ by the corresponding elements $\tdev[i][\hat\sigma][\q_{\hat\sigma}^t]$ grows sublinearly. We make that construction explicit in Algorithm~\ref{algo:R hatsigma}. Its formal guarantees are stated in \cref{prop:regret bound R hatsigma}.

    \begin{algorithm}[ht]\caption{Regret minimizer $\tilde{\cR}^{(i)}_{\hat\sigma}$ for the set $\Phj \defeq \{\tdev[i][\hat\sigma][\q_{\hat\sigma}] : \q_{\hat\sigma}\in\Q_j\}$}\label{algo:R hatsigma}
        \DontPrintSemicolon
            \KwData{\bull $i \in [n]$ player\newline
                    \bull $\hat\sigma = (j,a) \in \Seqs_*$ trigger sequence\newline
                    \bull $\tilde{\cR}^{(i)}_{\cQ,\,\hat\sigma}$ regret minimizer for set $\Q_j$ (\emph{e.g.}, the CFR algorithm~\citep{zinkevich2008regret})}
            \BlankLine
        \Fn{\normalfont\textsc{NextElement}()}{
            $\q_{\hat\sigma}^t \gets \tilde{\cR}^{(i)}_{\cQ,\,\hat\sigma}.\textsc{NextElement}()$\;
            \Return{$\tdev[i][\hat\sigma][\q_{\hat\sigma}^t]$\ , represented in memory implicitly through the vector $\q_{\hat\sigma}^t$}
        }
        \Hline{}
        \Fn{\normalfont\textsc{ObserveUtility}($L^t$)}{
            $g^{(i),\,t}_{\hat\sigma} \gets \text{linear function }\bbR^{|\Seqs_j|}\ni\vec{x} \mapsto L^t(h^{(i)}_{\hat\sigma}(\vec{x})) - L^t(h^{(i)}_{\hat\sigma}(\vec{0}))$\;\vspace{1mm}
            $\tilde{\cR}^{(i)}_{\cQ,\,\hat\sigma}.\textsc{ObserveUtility}(g^{(i),\,t}_{\hat\sigma})$\;
        }
    \end{algorithm}

Algorithm~\ref{algo:R hatsigma} can be instantiated with any regret minimizer $\tilde{\cR}^{(i)}_{\cQ,\,\hat\sigma}$ for the set of sequence-form strategies $\Q_j$. The following proposition formalizes the cumulative regret guarantee when  $\tilde{\cR}^{(i)}_{\cQ,\,\hat\sigma}$ is set to the CFR algorithm~\citep{zinkevich2008regret}, which so far has arguably been the most widely used regret minimizer for sequence-form strategy spaces.

\begin{proposition}\label{prop:regret bound R hatsigma}
Let $i\in[n]$ be any player, and $\hat\sigma=(j,a)\in\Seqs_*$ be any trigger sequence, and consider the regret minimizer $\tilde{\cR}^{(i)}_{\hat\sigma}$ (Algorithm~\ref{algo:R hatsigma}), where $\tilde{\cR}_{\cQ,\,\hat\sigma}$ is set to be the CFR regret minimizer~\citep{zinkevich2008regret}. 
Upon observing a sequence of linear utility functions $L^1,\dots,L^T:\co\Ph\to\bbR$, the regret cumulated by the elements $\phi^1 = \tdev[i][\hat\sigma][\vec{q}_{\hat\sigma}^1],\dots,\phi^T = \tdev[i][\hat\sigma][\vec{q}_{\hat\sigma}^T]$ output by $\tilde{\cR}^{(i)}_{\hat\sigma}$ satisfies
    \[
        R^T = \max_{\phi^*\in\Phj} \sum_{t=1}^T L^t(\phi^*) - L^t(\phi^t) \le D\, |\Seqs_j|\, \sqrt{T},
    \]
    where $D$ is the maximum range of $L^1, \dots, L^t$, that is, any constant such that $\max_{\phi,\phi'\in\Phj}\{L^t(\phi)-L^t(\phi')\}\le D$ for all $t=1,\dots, T$.
    Furthermore, the \textsc{NextElement} and the \textsc{ObserveUtility} operations run in $O(|\Seqs|)$ time.
\end{proposition}
\begin{proof}
    The regret cumulated by $\tilde{\cR}^{(i)}_{\hat\sigma}$ upon observing linear utility functions $L^1,\dots,L^t$ equals the regret cumulated by the CFR algorithm upon observing linear utility functions $g^{(i),\,t}_{\hat\sigma} : \bbR^{|\Seqs_j|}\ni\vec{x} \mapsto L^t(h^{(i)}_{\hat\sigma}(\vec{x})) - L^t(h^{(i)}_{\hat\sigma}(\vec{0}))$, as shown in \eqref{eq:regret transformation}. Furthermore, the range of $g^{(i),\,t}_{\hat\sigma}$ satisfies the inequality
    \begin{align*}
        \max_{\vec{q},\,\vec{q}'\in\Q_j} g^{(i),\,t}_{\hat\sigma}(\vec{q}) - g^{(i),\,t}_{\hat\sigma}(\vec{q}') &= \max_{\vec{q},\,\vec{q}'\in\Q_j} L^t(h^{(i)}_{\hat\sigma}(\vec{q})) - L^t(h^{(i)}_{\hat\sigma}(\vec{q}'))\\
            &= \max_{\phi,\,\phi'\in\Phj} L^t(\phi)-L^t(\phi') \le D.
    \end{align*}
    So, applying the regret bound of the CFR algorithm (Theorems~3 and~4 of~\citet{zinkevich2008regret}),
    \[
        R^T \le D \mleft(\sum_{j' \succeq j}\sqrt{|\A(j')|}\mright)\sqrt{T} \le D\mleft(\sum_{j'\succeq j} |\A(j')|\mright)\sqrt{T} = D|\Seqs_j|\sqrt{T},
    \]
    completing the proof of the regret bound.
    
    The complexity analysis of the \textsc{NextElement} operation follows trivially from the fact that CFR's \textsc{NextElement} operation runs in linear time in $|\Seqs_j|$. So, we focus on the complexity of the \textsc{ObserveUtility} operation. Fix any time $t$, and let $\vec{\Lambda}^t \defeq \langle L^t\rangle$ be the canonical representation of the linear utility function $L^t$ (\cref{sec:notation}). Since the canonical representation of $h^{(i)}_{\hat\sigma}(\vec{x})$ is the matrix $\Mdev[i][\hat\sigma][\vec{x}]$ for all $\vec{x}\in\bbR_{\ge0}^{|\Seqs_j|}$, using~\eqref{eq:functional canonical} we obtain
    \begin{align*}
        L^t(h^{(i)}_{\hat\sigma}(\vec{x})) - L^t(h^{(i)}_{\hat\sigma}(\vec{0})) &=
        \sum_{\sigma_r,\sigma_c\in\Seqs_j} \vec{\Lambda}^t[\sigma_r,\sigma_c]\,\mleft(\Mdev[i][\hat\sigma][\vec{x}]{}[\sigma_r,\sigma_c] - \Mdev[i][\hat\sigma][\vec{0}]{}[\sigma_r,\sigma_c]\mright)\\
        &= \sum_{\sigma_r\ge j} \vec{\Lambda}^t[\sigma_r,\hat\sigma]\,\vec{x}[\sigma_r],
    \end{align*}
    where the second equality follows from expanding the definitions of $\Mdev[i][\hat\sigma][\vec{x}]{}[\sigma_r,\sigma_c]$ and $\Mdev[i][\hat\sigma][\vec{0}]$ given in~\eqref{eq:Mdev}. So, the canonical representation $\langle g^{(i),\,t}_{\hat\sigma}\rangle$ of $g^{(i),\,t}_{\hat\sigma}$ is the vector $(\vec{\Lambda}^t[\sigma_r,\hat\sigma])_{\sigma_r\in\Seqs_j}$, which can be clearly computed and stored in memory in $O(|\Seqs_j|)$ time. Using the fact that CFR's \textsc{ObserveUtility} operation runs in linear time in $|\Seqs_j|$, the complexity bound of the statement follows.
\end{proof}

\subsubsection{Regret Minimizer for $\co\Ph$}\label{sec:rm ch}

Recently, \citet{Farina19:Regret} showed that a regret minimizer for a composite set of the form $\co\{\cX_1, \dots, \cX_m\}$ can be constructed by combining any individual regret minimizers for $\cX_1,\dots,\cX_m$ through a construction---called a \emph{regret circuit}---which we describe next.

\begin{proposition}[\citet{Farina19:Regret}, Section 4.3\footnote{Technically, \citet{Farina19:Regret} only prove the bound~\eqref{eq:circuit bound} for the case $m=2$. However, as mentioned by the authors, the extension to generic $m\in\bbN$ is direct.
}]\label{prop:ch circuit}
    Let $\cX_1, \dots, \cX_m$ be a finite collection of sets, and let $\cR_1, \dots,\cR_m$ be any regret minimizers for them. Furthermore, let $\cR_\Delta$ be any regret minimizer for the $m$-simplex $\Delta^m \defeq \{(\lambda_1, \dots,\lambda_m) \in \bbR_{\ge 0}^m, \sum_k\lambda_K = 1\}$. A regret minimizer $\cR_\text{co}$ for the set $\co\{\cX_1, \dots, \cX_m\}$ can be constructed starting from $\cR_1,\dots,\cR_m$ and $\cR_\Delta$ as follows.
    \begin{itemize}
        \item $\cR_\text{co}.\textsc{NextElement}$ calls \textsc{NextElement} on each of the regret minimizers $\cR_1,\dots,\cR_m$, obtaining elements $\vec{x}_1^t, \dots, \vec{x}_m^t$. Then, it calls the \textsc{NextElement} operation on $\cR_\Delta$, obtaining an element of the simplex $\vec{\lambda}^t = (\lambda_1^t, \dots,\lambda_m^t)$. Finally, it returns the element
        \[
            \lambda_1^t \vec{x}_1^t + \dots + \lambda_m^t \vec{x}_m^t \in \co\{\cX_1, \dots, \cX_m\}.
        \]
        \item $\cR_\text{co}.\textsc{ObserveUtility}(L^t)$ forwards the linear utility function $L^t$ to each of the regret minimizers $\cR_1, \dots, \cR_m$. Then, it calls the \textsc{ObserveUtility} operation on $\cR_\text{co}$ with the linear utility function $(\lambda_1, \dots,\lambda_m) \mapsto L^t(\vec{x}^t_1)\lambda_1 + \dots + L^t(\vec{x}_m^t)\lambda_m$.
    \end{itemize}
    In doing so, the regret $R_\text{co}^T$ cumulated by $\cR_\text{co}$ upon observing any $T$ linear utility functions relates to the regrets $R_1^T, \dots, R_m^T, R_\Delta^T$ cumulated by $\cR_1,\dots,\cR_m,\cR_\Delta$, respectively, according to the inequality
    \begin{equation}\label{eq:circuit bound}
        R_\text{co}^T \le R_\Delta^T + \max\{R_1^T, \dots, R_m^T\}.
    \end{equation}
\end{proposition}

We apply the construction described in \cref{prop:ch circuit} to obtain our regret minimizer $\cR^{(i)}$ for the set $\co\Ph = \co\{\Phj:\hat\sigma\in\Seqs\}$ starting from the regret minimizers $\tilde{\cR}^{(i)}_{\hat\sigma}$ (Algorithm~\ref{algo:R hatsigma}), one for each sequence $\hat\sigma\in\Seqs_*$, as well as any regret minimizer $\cR_\Delta^{(i)}$ for the simplex $\Delta^{|\Seqs_*|}$. Pseudocode is given in Algorithm~\ref{algo:R i}.

\begin{algorithm}[th]\caption{Regret minimizer $\tilde{\cR}^{(i)}$ for the set $\co\Ph = \co\{\Phj:\hat\sigma\in\Seqs\}$}\label{algo:R i}
        \DontPrintSemicolon
            \KwData{\bull $i \in [n]$ player\newline
                    \bull $\tilde{\cR}^{(i)}_{\hat\sigma}$ (one for each $\hat\sigma\in\Seqs_*$) regret minimizer for $\Phj$ as defined in Algorithm~\ref{algo:R hatsigma}\newline
                    \bull $\cR^{(i)}_\Delta$ regret minimizer for $\Delta^{|\Seqs_*|}$ (\emph{e.g.}, regret matching~\citep{hart2000simple})}
            \BlankLine
        \Fn{\normalfont\textsc{NextElement}()}{
            $\vec{\lambda}^t \gets \tilde{\cR}^{(i)}_{\Delta}.\textsc{NextElement}()$\;
            \For{$\hat\sigma\in\Seqs_*$}{
                $\tdev[i][\hat\sigma][\vec{q}_{\hat\sigma}^t] \gets \tilde{\cR}^{(i)}_{\hat\sigma}.\textsc{NextElement}()$\;
            }
            \Return{$\sum_{\hat\sigma\in\Seqs_*} \vec{\lambda}^t[\hat\sigma]\,\tdev[i][\hat\sigma][\vec{q}_{\hat\sigma}^t]$\ , represented in memory implicitly as list $\{(\vec{\lambda}^t[\hat\sigma], \vec{q}^t_{\hat\sigma})\}_{\hat\sigma\in\Seqs_*}$}
        }
        \Hline{}
        \Fn{\normalfont\textsc{ObserveUtility}($L^t$)}{
            \For{$\hat\sigma\in\Seqs_*$}{
                $\tilde{\cR}^{(i)}_{\hat\sigma}.\textsc{ObserveUtility}(L^t)$\;
            }
            $\ell_\lambda^t \gets$ linear utility function $\vec{\lambda} \mapsto \sum_{\hat\sigma\in\Seqs_*} \vec{\lambda}[\hat\sigma]\,L^t(\tdev[i][\hat\sigma][\vec{q}_{\hat\sigma}^t])$\;
            $\tilde{\cR}^{(i)}_{\Delta}.\textsc{ObserveUtility}(\ell_\lambda^t)$\;
        }
    \end{algorithm}

Combining \cref{prop:ch circuit} and \cref{prop:regret bound R hatsigma}, we obtain the following result.

\begin{theorem}\label{thm:analysis R i}
    Consider the regret minimizer $\tilde{\cR}^{(i)}$ (Algorithm~\ref{algo:R i}), where $\cR_\Delta^{(i)}$ is set to the regret matching algorithm, and $\tilde{\cR}^{(i)}_{\hat\sigma}$ is instantiated as described in \cref{prop:regret bound R hatsigma}. Upon observing a sequence of linear utility functions $L^1,\dots,L^T:\co\Ph\to\bbR$, the regret cumulated by $\tilde{\cR}^{(i)}$ satisfies
    \[
        R^T = \max_{\phi^* \in \co\Ph}\sum_{t=1}^T L^t(\phi^*) - L^t(\phi^t) \le 2D\,|\Seqs|\,\sqrt{T},
    \]
    where $D$ is the maximum range of $L^1, \dots, L^t$, that is, any constant such that $\max_{\phi,\phi'\in\co\Ph}\{L^t(\phi)-L^t(\phi')\}\le D$ for all $t=1,\dots, T$.
    Furthermore, the \textsc{NextElement} and the \textsc{ObserveUtility} operations run in $O(|\Seqs|^2)$ time.
\end{theorem}
\begin{proof}
    At all $t$, the range of the linear utility function $\vec{\lambda} \mapsto \sum_{\hat\sigma\in\Seqs_*} \vec{\lambda}[\hat\sigma]\,L^t(\tdev[i][\hat\sigma][\vec{q}_{\hat\sigma}^t])$ is upper bounded by $D$. Hence, from the known regret bound of the regret matching algorithm, the regret cumulated by $\cR_\Delta^{(i)}$ after $T$ iterations is upper bound as
    \[
        R^T_\Delta \le D \sqrt{|\Sigma^{(i)}|}\sqrt{T}\le D |\Sigma^{(i)}|\sqrt{T}.
    \]
    On the other hand, the regret bound in \cref{prop:regret bound R hatsigma} shows that, for all $\hat\sigma=(j,a)\in\Seqs_*$, the regret $R^T_{\hat\sigma}$ cumulated by $\tilde{\cR}_{\hat\sigma}^{(i)}$ is upper bounded as $R^T_{\hat\sigma} \le D |\Seqs_j| \sqrt{T}$. So, using~\eqref{eq:circuit bound} together with the fact that $|\Seqs_j| \le |\Seqs|$ for all $j\in\cJ$, yields the regret bound in the statement.
    
    The complexity analysis of \textsc{NextElement} is completely straightforward. So, we focus on the complexity of \textsc{ObserveLoss}. There, the only operation whose analysis is not immediately obvious is the construction of the linear utility function $\ell^t_\lambda: \vec{\lambda} \mapsto \sum_{\hat\sigma\in\Seqs_*} \vec{\lambda}[\hat\sigma]\,L^t(\tdev[i][\hat\sigma][\vec{q}_{\hat\sigma}^t])$, where it is necessary to check that its canonical representation (\cref{sec:notation}), given by the vector $\mleft(L^t(\tdev[i][\hat\sigma][\vec{q}_{\hat\sigma}^t])\mright)_{\hat\sigma\in\Seqs_*}$, can be computed in $O(|\Seqs|^2)$ time. Fix any $\hat\sigma\in\Seqs_*$. The canonical representation of $\tdev[i][\hat\sigma][\vec{q}_{\hat\sigma}^t]$ is $\Mdev[i][\hat\sigma][\vec{q}_{\hat\sigma}^t]$, which is a matrix with $O(|\Seqs|)$ nonzero entries. So, using~\eqref{eq:functional canonical}, the evaluation of $L^t(\tdev[i][\hat\sigma][\vec{q}_{\hat\sigma}^t])$ via the canonical representations of $L^t$ (given as input) and $\tdev[i][\hat\sigma][\vec{q}_{\hat\sigma}^t]$ takes $O(|\Seqs|)$ time. So, the representation of $\ell_\lambda$ can be computed in $O(|\Seqs|^2)$ time, confirming the analysis in the statement.
\end{proof}

%% file: text/fixed_point.tex
\subsection{Computation of the Next Strategy}\label{sec:fixed point}

In this subsection we complete the construction of our $\Tph$-regret minimizer for $\Q$ (\cref{prob:mixed}) started in \cref{sec:relaxing}, by showing that every transformation $\phi\in\co\Ph$ admits a fixed point $\Q \ni \vec{q} = \phi(\vec{q}) $, and that such a fixed point can be computed in time quadratic in the number of sequences $\Seqs$ of Player~$i$.

As a key step in our algorithm, we will use the following well-known result about stationary distributions of stochastic matrices.

\begin{lemma}
    Let $\vec{M} \in \bbS^m$ be a stochastic matrix. Then, $\vec{M}$ admits a fixed point $\Delta^m\ni\vec{x} = \vec{M}\vec{x}$. Furthermore, such a fixed point can be computed in polynomial time in $m$.
\end{lemma}
Several algorithms are known for computing fixed points of stochastic matrices (see, \emph{e.g.}, \citep{Paige75:Computation} for a comparison of eight different methods). Since the particular choice of method is irrelevant, in this article we will make the following assumption.
\begin{assumption}
    Given any $m \in \bbN$, we assume access to an oracle for computing a fixed point of any $m\times m$ stochastic matrix $\vec{M}$. Furthermore, we assume that the oracle requires at most $O(\textsf{\upshape FP}(m))$ time in the worst case to compute any such fixed point. 
\end{assumption}

Our algorithm for computing the fixed point requires that the transformation $\phi\in\co\Phi$ be expressed as a convex combination of elements from the sets $\{\Phj\}_{\hat\sigma\in\Seqs_*}$, that is, an expression of the form
\[
\phi = \sum_{\hat\sigma\in\Seqs_*}\lambda_{\hat\sigma}\, \tdev[i][\hat\sigma][\vec{q}_{\hat\sigma}],\quad\text{ where } \sum_{\hat\sigma\in\Seqs_*}\lambda_{\hat\sigma}=1, \text{ and }\lambda_{\hat\sigma} \ge 0,\ \ \vec{q}_{\hat\sigma}\in\Q_j \quad\! \forall\ \hat\sigma=(j,a)\in\Seqs_*,
\numberthis{eq:phi combination}
\]
in accordance with the characterization of $\co\Ph$ established in \cref{eq:co expansion}, and \cref{lem:ph j structure}. Note that our regret minimizer $\tilde{\cR}^{(i)}$ for the set $\co\Ph$ (\cref{algo:R i}) already outputs transformations $\phi$ expressed in the form above.
Our algorithm operates \emph{incrementally}, constructing a fixed point strategy $\vec{q}$ for $\phi$ information set by information set, in a top down fashion. To formalize this notion of top-down construction, we will make use of the two following definitions.

\begin{definition}\label{def:trunk}
    Let $i \in [n]$ be a player, and $J \subseteq \cJ$ be a subset of that player's information sets. We say that $J$ is \emph{a trunk of $\cJ$} if, for every $j \in J$, all predecessors of $j$ (that is, all $j'\in\cJ, j'\prec j$) are also in $J$.
\end{definition}

\begin{example}
    In the small game of \cref{fig:preliminary_example} (Left), the sets $\emptyseq$, $\{\textsc{a}\}$, $\{\textsc{a},\textsc{b}\}$, $\{\textsc{a},\textsc{c}\}$, $\{\textsc{a},\textsc{d}\}$, $\{\textsc{a},\textsc{b},\textsc{d}\}$, $\{\textsc{a},\textsc{c}, \textsc{d}\}$, and $ \{\textsc{a},\textsc{b},\textsc{c}\,\textsc{d}\} = \cJ[1]$ are all possible valid trunks for Player~$1$. Contrarily, set $J = \{\textsc{b}\}$ is \emph{not} a trunk for Player~$1$, because $\textsc{a} \prec \textsc{b}$ and yet $\textsc{a} \not\in J$.
\end{example}

\begin{definition}\label{def:J partial}
    Let $i\in[n]$ be a player, $J \subseteq \cJ$ be a trunk of $\cJ$ (\cref{def:trunk}), and $\phi\in\co\Ph$. We say that a vector $\vec{x}\in\bbR_{\ge0}^{|\Seqs|}$ is a \emph{$J$-partial fixed point of $\phi$} if
    it satisfies the sequence-form constraints at all $j\in J$, that is,
    \[
        \vec{x}[\emptyseq] = 1, \qquad \vec{x}[\sigma^{(i)}(j)] = \sum_{a \in \A(j)} \vec{x}[(j,a)] \quad \forall\ j \in J,
        \numberthis{eq:partial sf}
    \]
    and furthermore
    \[
        \phi(\vec{x})[\emptyseq] = \vec{x}[\emptyseq] = 1, \qquad \phi(\vec{x})[(j,a)] = \vec{x}[(j,a)]\quad\forall\ j\in J,\ a\in \A(j).
        \numberthis{eq:partial fp}
    \]
\end{definition}

It follows from \cref{def:J partial} that a $\cJ$-partial fixed point of $\phi$ is a vector $\vec{q}\in\Q$ such that $\vec{q} = \phi(\vec{q})$. 
The core insight of our algorithm lies in the fact that a $J$-partial fixed point can be cheaply promoted to be a $(J\cup\{j^*\})$-partial fixed point where $j^* \in \cJ \setminus J$ is any information set whose predecessors are all in $J$.
Algorithm~\ref{algo:extend} below gives an implementation of such a promotion: $\textsc{Extend}(\phi, J, j^*, \vec{x})$ starts with a $J$-partial fixed point $\vec{x}$ of $\phi$, and modifies all entries $\vec{x}[(j^*,a)]$, $a\in\A(j^*)$, so that $\vec{x}$ becomes a $(J\cup\{j^*\})$-partial fixed point.
Therefore, at a conceptual level, one can repeatedly invoke $\textsc{Extend}$, growing the trunk $J$ one information set at a time until $J = \cJ$. The following simple lemma establishes the basis of induction, by exhibiting an $\emptyset$-partial fixed point for any transformation $\phi\in\co\Ph$.
\begin{lemma}\label{lem:x0}
    Let $i\in[n]$ be a player, and $\phi = \sum_{\hat\sigma\in\Seqs_*}\lambda_{\hat\sigma}\tdev[i][\hat\sigma][\vec{q}_{\hat\sigma}]$ be any transformation in $\co\Ph$. Then, the vector $\vec{x}_0 \in \bbR_{\ge0}^{|\Seqs|}$, whose entries are all zeros except for $\vec{x}_0[\emptyseq] = 1$, is a $\emptyset$-partial fixed point of $\phi$.
\end{lemma}
\begin{proof}
    Condition~\eqref{eq:partial sf} is straightforward. So, we focus on~\eqref{eq:partial fp}.
    Fix any $\hat\sigma=(j,a)\in\Seqs_*$. The definition of $\Mdev[i][\hat\sigma][\hat{\vec{q}}]$, given in~\eqref{eq:Mdev}, implies that
    \[
        \Mdev[i][\hat\sigma][\hat{\vec{q}}_{\hat\sigma}]{}[\sigma_r, \emptyseq] = \begin{cases}
            1 & \text{if }\sigma_r = \emptyseq\\
            0 & \text{otherwise}
        \end{cases}
        \qquad\quad\forall\,\sigma_r\in\Seqs.
    \]
    Consequently, $\tdev[i][\hat\sigma][\hat{\vec{q}}_{\hat\sigma}](\vec{x}_0) = \Mdev[i][\hat\sigma][\hat{\vec{q}}_{\hat\sigma}]\,\vec{x}_0 = \vec{x}_0$ (from expanding the matrix-vector multiplication). So, $\phi(\vec{x}_0) = \sum_{\hat\sigma\in\Seqs_*} \lambda_{\hat\sigma}\tdev[i][\hat\sigma][\hat{\vec{q}}_{\hat\sigma}](\vec{x}_0) = \vec{x}_0$ and in particular $\phi(\vec{x}_0)[\emptyseq]=\vec{x}_0[\emptyseq]=1$. So, \eqref{eq:partial fp} holds, as we wanted to show.
\end{proof}

Before giving a proof of correctness and an analysis of the complexity of $\textsc{Extend}$, we illustrate an application of the algorithm in the simple extensive-form game of \cref{fig:preliminary_example}.

\begin{algorithm}[t]\caption{$\textsc{Extend}(\phi, J, j^*, \vec{x})$}\label{algo:extend}
        \DontPrintSemicolon
        \SetKwInput{KwInput}{Input\hspace{2.5mm}}
        \SetKwInput{KwOutput}{Output}
        \KwInput{%
            \bull $\phi = \sum_{\hat\sigma\in\Seqs_*} \lambda_{\hat\sigma} \tdev[i][\hat\sigma][\vec{q}_{\hat\sigma}] \in \co\Ph$ transformation for a player $i\in[n]$, represented\newline\hspace*{4.9cm} in memory implicitly as the list $\{(\lambda_{\hat\sigma}, \vec{q}_{\hat\sigma})\}_{\hat\sigma\in\Seqs_*}$\newline
            \bull $J \subseteq \cJ$ trunk for Player $i$\newline
            \bull $j^* \in \cJ$ information set not in $J$ such that its immediate predecessor is in $J$\newline
            \bull $\vec{x}\in\bbR_{\ge0}^{|\Seqs|}$~~ $J$-partial fixed point of $\phi$}
        \KwOutput{\bull $\vec{x}'\in\bbR_{\ge0}^{|\Seqs|}$~~ $(J\cup\{j^*\})$-partial fixed point of $\phi$}
            \BlankLine
        $\sigma_p \gets \sigma^{(i)}(j^*)$\label{line:sigmap}\;
        Let $\vec{r} \in \bbR_{\ge0}^{|\A(j^*)|}$ be the vector whose entries are defined, for all $a \in \A(j^*)$, as\newline $\hspace*{8mm}
            \displaystyle\vec{r}[a] \defeq \sum_{j'\preceq \sigma_p}\sum_{a'\in\A(j')} \lambda_{(j',a')}\,\vec{q}_{(j',a')}[(j^*, a)]\,\vec{x}[(j',a')]
        $\;\vspace{1mm}\label{line:r}
        Let $\vec{W} \in \vec{x}[\sigma_p]\cdot\bbS^{|\A(j^*)|}$ be the matrix whose entries are defined, for all $a_r, a_c \in \A(j^*)$, as\newline $\hspace*{8mm}
            \displaystyle\vec{W}[a_r, a_c] \defeq \vec{r}[a_r] + \mleft(\lambda_{(j^*,a_c)}\vec{q}_{(j^*,a_c)}[(j^*, a_r)] + \mleft(1- \!\!\!\sum_{\substack{\hat\sigma\in\Seqs_*\!\!,\,\hat\sigma \preceq (j^*, a_c)}} \!\!\!\lambda_{\hat\sigma}\mright)\,\bbone[a_r = a_c]\mright)\,\vec{x}[\sigma_p]
        $\;\vspace{1mm}\label{line:W}
        \uIf{$\vec{x}[\sigma_p] = 0$\label{line:start}}{
            $\vec{w} \gets \vec{0} \in \bbR_{\ge0}^{|\A(j^*)|}$\label{line:five}\;
        }\Else{
            $\vec{b} \in \Delta^{|\A(j^*)|} \gets$ fixed point of stochastic matrix $\frac{1}{\vec{x}[\sigma_p]}\vec{W}$\;
            $\vec{w} \gets \vec{x}[\sigma_p]\,\vec{b}$\;\label{line:end}
        }
        $\vec{x}' \gets \vec{x}$\;
        \For{$a\in \A(j^*)$}{
            $\vec{x}'[(j^*,a)] \gets \vec{w}[(j^*,a)]$\;\label{line:last}
        }
        \Return{$\vec{x}'$}
    \end{algorithm}

\begin{example}
    Consider the simple extensive-form game of \cref{fig:preliminary_example} (Left), and recall the three deviation functions $\tdev[1][\seq{1}][\hatpure[]_a], \tdev[1][\seq{2}][\hatpure[]_b], \tdev[1][\seq{3}][\hatpure[]_c]$ considered in \cref{ex:deviations}. We will illustrate two applications of \textsc{Extend}, with respect to the transformation
    \[
        \phi \defeq \frac{1}{2}\, \tdev[1][\seq{1}][\hatpure[]_a] + \frac{1}{3}\, \tdev[1][\seq{2}][\hatpure[]_b] + \frac{1}{6}\,\tdev[1][\seq{3}][\hatpure[]_c] \in \co\Ph[1].
    \]
    \begin{itemize}
        \item In the first application, consider the trunk $J = \emptyset$, information set $j^* = \textsc{a}$, and the $\emptyset$-partial fixed point described in \cref{lem:x0}, that is, the vector $\vec{x}$ whose components are all $0$ except for the entry corresponding to the empty sequence $\emptyseq$, which is set to $1$. In this case, $\sigma_p = \sigma^{(i)}(j^*)$ (\cref{line:sigmap}) is the empty sequence. Since no information set $j'$ can possibly satisfy $j' \preceq \sigma_p$, the vector $\vec{r}$ defined on \cref{line:r} is the zero vector. Consequently, the matrix $\vec{W}$ defined on \cref{line:W} is
        \[
            \vec{W} = \begin{tikzpicture}[baseline=-\the\dimexpr\fontdimen22\textfont2\relax ]
            \tikzset{every left delimiter/.style={xshift=1.5ex},
                     every right delimiter/.style={xshift=-1ex}};
            \matrix [matrix of math nodes,left delimiter=(,right delimiter=),row sep=.007cm,column sep=.007cm](m)
            {
            \nicefrac{1}{2} & \nicefrac{1}{3}\\
            \nicefrac{1}{2} & \nicefrac{2}{3}\\
            };
            
            \node[left=3pt of m-1-1] (left-0) {\small\seq{1}};
            \node[left=3pt of m-2-1] (left-1) {\small\seq{2}};

            \node[above=1pt of m-1-1] (top-0) {\small\seq{1}};
            \node[above=1pt of m-1-2] (top-1) {\small\seq{2}};
            \end{tikzpicture}
        \]
        which is a stochastic matrix. A fixed point for $\vec{W}$ is given by the vector $\vec{b} \defeq (\nicefrac{2}{5}, \nicefrac{3}{5}) \in \Delta^{|\{\seq{2},\seq{3}\}|}$. So, the vector $\vec{x}'$ returned by \textsc{extend} is given by
        \[
            \vec{x}'[\emptyseq] = 1,\quad \vec{x}'[(\textsc{a},\seq{1})] = \frac{2}{5},\quad \vec{x}'[(\textsc{a},\seq{2})] = \frac{3}{5}
        \]
        and zero entries everywhere else. Direct inspection reveals that $\vec{x}'$ is indeed a $\{\textsc{a}\}$-partial fixed point of $\phi$.
        \item In the second application of \textsc{Extend}, we start from the $\{\textsc{a}\}$-partial fixed point that we computed in the previous bullet point, and extend it to a $\{\textsc{a},\textsc{d}\}$-partial fixed point. Here, $j^* = \textsc{d}$, and so $\sigma_p = (\textsc{a}, \seq{2})$. The only $j' \preceq \sigma_p$ is $\textsc{a}$, and so the vector $\vec{r}$ defined on \cref{line:r} is
        \[
            \vec{r}[\seq{7}] = \frac{1}{5}, \quad \vec{r}[\seq{8}] = 0.
        \]
        Consequently, the matrix $\vec{W}$ defined on \cref{line:W} is
        \[
            \vec{W} = \begin{tikzpicture}[baseline=-\the\dimexpr\fontdimen22\textfont2\relax ]
            \tikzset{every left delimiter/.style={xshift=1.5ex},
                     every right delimiter/.style={xshift=-1ex}};
            \matrix [matrix of math nodes,left delimiter=(,right delimiter=),row sep=.007cm,column sep=.007cm](m)
            {
            \nicefrac{3}{5} & \nicefrac{1}{5}\\
            0 & \nicefrac{2}{5}\\
            };
            
            \node[left=3pt of m-1-1] (left-0) {\small\seq{7}};
            \node[left=4pt of m-2-1] (left-1) {\small\seq{8}};

            \node[above=1pt of m-1-1] (top-0) {\small\seq{7}};
            \node[above=1pt of m-1-2] (top-1) {\small\seq{8}};
            \end{tikzpicture}
        \]
        As expected, $\vec{W} \in \nicefrac{3}{5}\, \bbS^{|\{\seq{7},\seq{8}\}|} = \vec{x}[(\textsc{a},\seq{2})]\,\bbS^{|\{\seq{7},\seq{8}\}|}$. A fixed point for $\frac{1}{\vec{x}[(\textsc{a},\seq{2})]}\vec{W} = \nicefrac{5}{3}\,\vec{W}$ is given by the vector $\vec{b} \defeq (1, 0)$. So, the vector $\vec{x}'$ returned by \textsc{Extend} is given by
            \[
                \vec{x}'[\emptyseq] = 1,\quad \vec{x}'[(\textsc{a},\seq{1})] = \frac{2}{5},\quad \vec{x}'[(\textsc{a},\seq{2})] = \frac{3}{5}, \quad  \vec{x}'[(\textsc{d},\seq{7})] = \frac{3}{5}, \quad \vec{x}'[(\textsc{d},\seq{8})] = 0
            \]
        and zero entries everywhere else. Once again, direct inspection reveals that $\vec{x}'$ is indeed a $\{\textsc{a},\textsc{d}\}$-partial fixed point of $\phi$.
    \end{itemize}
\end{example}

In order to prove correctness of $\textsc{Extend}$ in~\cref{prop:extend}, we will find useful the following technical lemma.


\begin{lemma}\label{lem:expansion}
    Let $i\in[n]$ be any player, and $\phi = \sum_{\hat\sigma\in\Seqs_*}\lambda_{\hat\sigma}\, \tdev[i][\hat\sigma][\vec{q}_{\hat\sigma}]$ be any linear transformation in $\co\Ph$ expressed as in~\eqref{eq:phi combination}. Then, for all $\sigma\in\Seqs$,
    \[
        \phi(\vec{x})[\sigma] = \mleft(1-\sum_{\substack{\hat\sigma\in\Seqs_*\!\!,\, \hat\sigma\preceq\sigma}}\lambda_{\hat\sigma}\mright)\vec{x}[\sigma]+\sum_{j' \preceq \sigma}\sum_{a' \in \A(j')}\lambda_{(j',a')}\,\vec{q}_{(j',a')}[\sigma]\, \vec{x}[(j',a')].
    \]
\end{lemma}
\begin{proof}
    Fix any trigger sequence $\hat\sigma = (j',a')\in\Seqs_*$. By expanding the matrix-vector multiplication between $\Mdev[i][\hat\sigma][\vec{q}_{\hat\sigma}]$ (\cref{def:M}) and $\vec{x}$, we have that for all $\sigma\in\Seqs$,
    \begin{align*}
        \tdev[i][\hat\sigma][\vec{q}_{\hat\sigma}](\vec{x})[\sigma] &= \vec{x}[\sigma]\bbone[\sigma \not\succeq\hat\sigma] + \vec{q}_{\hat\sigma}[\sigma] \vec{x}[\hat\sigma] \bbone[\sigma\succeq j'].\numberthis{eq:matrix vector general}
    \end{align*}
    Therefore, for all $\sigma\in\Seqs$,
    \begin{align*}
        \phi(\vec{x})[\sigma]
        &= \sum_{\hat\sigma \in \Seqs_*} \lambda_{\hat\sigma}\,\tdev[i][\hat\sigma][\vec{q}_{\hat\sigma}](\vec{x})[\sigma] = \sum_{\hat\sigma=(j',a')\in\Seqs_*} \lambda_{\hat\sigma}\mleft(\vec{x}[\sigma]\bbone[\sigma \not\succeq\hat\sigma] + \vec{q}_{\hat\sigma}[\sigma] \vec{x}[\hat\sigma] \bbone[\sigma\succeq j']\mright)\\
        &= \mleft(\sum_{\substack{\hat\sigma\in\Seqs_*\!\!,\,\sigma\not\succeq\hat\sigma}} \lambda_{\hat\sigma}\mright)\vec{x}[\sigma] + \sum_{j' \preceq \sigma}\sum_{a' \in \A(j')}\lambda_{(j',a')}\,\vec{q}_{(j',a')}[\sigma]\, \vec{x}[(j',a')]\\
        &= \mleft(1-\sum_{\substack{\hat\sigma\in\Seqs_*\!\!,\, \hat\sigma\preceq\sigma}}\lambda_{\hat\sigma}\mright)\vec{x}[\sigma]+\sum_{j' \preceq \sigma}\sum_{a' \in \A(j')}\lambda_{(j',a')}\,\vec{q}_{(j',a')}[\sigma]\, \vec{x}[(j',a')],
    \end{align*}
    as we wanted to show.
\end{proof}

\begin{proposition}\label{prop:extend}
    Let $i\in[n]$ be a player, $\phi = \sum_{\hat\sigma\in\Seqs_*}\lambda_{\hat\sigma}\, \tdev[i][\hat\sigma][\vec{q}_{\hat\sigma}]$ be a linear transformation in $\co\Ph$ expressed as in~\eqref{eq:phi combination}, $\vec{x} \in\bbR_{\ge0}^{|\Seqs|}$ be a $J$-partial fixed point of $\phi$, and $j^*\in\cJ$ be a information set not in $J$ such that its immediate predecessor is in $J$. Then, $\textsc{Extend}(\phi,J,j^*,\vec{x})$, given in Algorithm~\ref{algo:extend}, computes a $(J\cup\{j^*\})$-partial fixed point of $\phi$ in time upper bound by $O(|\Seqs|\,|\A(j^*)| + \textsf{\upshape FP}(|\A(j^*)|))$.
\end{proposition}
\begin{proof}
    We break the proof into four parts. In the first part, we analyze the sum of the entries of vector $\vec{r}$ defined in \cref{line:r} of \cref{algo:extend}. In the second part, we prove that $\frac{1}{\vec{x}[\sigma_p]}\vec{W} \in \bbS^{|\A(j^*)|}$, as stated in \cref{line:W}. In the third part, we show that the output $\vec{x}'$ of the algorithm is indeed a $(J\cup\{j^*\})$-partial fixed point of $\phi$. Finally, in the fourth part we analyze the computational complexity of the algorithm.
    \paragraph{Part 1: sum of the entries of $\vec{r}$} In this first part of the proof, we study the sum of the entries of the vector $\vec{r}$ defined on \cref{line:r} in \cref{algo:extend}.
    By hypothesis the immediate predecessor of $j^*$ is in $J$, and since $\vec{x}$ is a $J$-partial fixed point, the sequence $\sigma_p \defeq \sigma^{(i)}(j^*)$ satisfies $\phi(\vec{x})[\sigma_p] = \vec{x}[\sigma_p]$. Hence, expanding the $\phi(\vec{x})[\sigma_p]$ using \cref{lem:expansion}, we conclude that
    \[
        \mleft(1-\sum_{\substack{\hat\sigma\in\Seqs_*\!\!,\, \hat\sigma\preceq\sigma_p}}\lambda_{\hat\sigma}\mright)\vec{x}[\sigma_p]+\sum_{j' \preceq \sigma_p}\sum_{a' \in \A(j')}\lambda_{(j',a')}\,\vec{q}_{(j',a')}[\sigma_p]\, \vec{x}[(j',a')] = \vec{x}[\sigma_p].
    \]
    So, by rearranging terms, we have
    \[
        \sum_{j' \preceq \sigma_p}\sum_{a' \in \A(j')}\lambda_{(j',a')}\,\vec{q}_{(j',a')}[\sigma_p]\, \vec{x}[(j',a')] =  \mleft(\sum_{\substack{\hat\sigma\in\Seqs_*\!\!,\,\hat\sigma\preceq\sigma_p}}\lambda_{\hat\sigma}\mright)\vec{x}[\sigma_p].\numberthis{eq:fp step2}
    \]
    
    On the other hand, since $\vec{q}_{(j',a')} \in \Q_{j'}$ for all $j'\preceq \sigma_p, a'\in\A(j')$, the vector $\vec{q}_{(j',a')}$ satisfy the the sequence-form (probability-mass-conservation) constraint
    \[
        \vec{q}_{(j',a')}[\sigma_p] = \sum_{a \in \A(j^*)} \vec{q}_{(j',a')}[(j^*,a)].
    \]
    Hence, and we can rewrite~\eqref{eq:fp step2} as
    \begin{align*}
        \mleft(\sum_{\substack{\hat\sigma\in\Seqs_*\!\!,\,\hat\sigma\preceq\sigma_p}}\lambda_{\hat\sigma}\mright)\vec{x}[\sigma_p] &= \sum_{j' \preceq \sigma_p}\sum_{a' \in \A(j')}\sum_{a\in\A(j^*)}\lambda_{(j',a')}\,\vec{q}_{(j',a')}[(j^*,a)]\, \vec{x}[(j',a')] \\
        &= \sum_{a\in\A(j^*)} \mleft(\sum_{j' \preceq \sigma_p}\sum_{a' \in \A(j')}\lambda_{(j',a')}\,\vec{q}_{(j',a')}[(j^*,a)]\, \vec{x}[(j',a')] \mright)\\
        &= \sum_{a\in\A(j^*)} \vec{r}[a],
    \end{align*}
    where the last equality follows from the definition of $\vec{r}$ in \cref{line:r} of \cref{algo:extend}. So, in conclusion,
    \begin{equation}\numberthis{eq:fp step3}
        \sum_{a\in\A(j^*)} \vec{r}[a] = \mleft(\sum_{\substack{\hat\sigma\in\Seqs_*\!\!,\,\hat\sigma\preceq\sigma_p}}\lambda_{\hat\sigma}\mright)\vec{x}[\sigma_p].
    \end{equation}
    
     \paragraph{Part 2: $\vec{W}$ belongs to $\vec{x}[\sigma_p]\cdot\bbS^{|\A(j^*)|}$} In this second part of the proof, we will prove that $\vec{W}$ as constructed on \cref{line:W} of Algorithm~\ref{algo:extend} is a non-negative matrix whose columns sum to the value $\vec{x}[\sigma_p]$. 
     Fix any $a_c \in \A(j^*)$. Then, the sum of the column of $\vec{W}$ corresponding to action $a_c$ is
     \begin{align*}
        &\sum_{a_r\in\A(j^*)} \vec{W}[a_r, a_c]=\!\!\!\sum_{a_r\in\A(j^*)}\!\!\!\!\vec{r}[a_r] + \mleft(\lambda_{(j^*,a_c)}\vec{q}_{(j^*,a_c)}[(j^*, a_r)] + \mleft(1-\!\! \sum_{\substack{\hat\sigma\in\Seqs_*\!\!,\,\hat\sigma \preceq (j^*, a_c)}} \!\!\lambda_{\hat\sigma}\mright)\,\bbone[a_r = a_c]\mright)\,\vec{x}[\sigma_p]\\
        &\quad= \mleft(\sum_{\substack{\hat\sigma\in\Seqs_*\!\!,\,\hat\sigma \preceq \sigma_p}}\lambda_{\hat\sigma}\mright)\vec{x}[\sigma_p] + \vec{x}[\sigma_p]\,\lambda_{(j^*,a_c)}\mleft(\sum_{a_r\in\A(j^*)}\vec{q}_{(j^*,a_c)}[(j^*, a_r)]\mright) + \mleft(1-\!\! \sum_{\substack{\hat\sigma\in\Seqs_*\!\!,\,\hat\sigma \preceq \sigma_p}} \!\!\lambda_{\hat\sigma}\mright)\,\vec{x}[\sigma_p]\\
        &\quad=\mleft(\sum_{\substack{\hat\sigma\in\Seqs_*\!\!,\,\hat\sigma \preceq (j^*, a_c)}}\lambda_{\hat\sigma}\mright)\vec{x}[\sigma_p] + \vec{x}[\sigma_p]\,\lambda_{(j^*,a_c)} + \mleft(1- \!\!\sum_{\substack{\hat\sigma\in\Seqs_*\!\!,\,\hat\sigma \preceq (j^*, a_c)}} \!\!\lambda_{\hat\sigma}\mright)\,\vec{x}[\sigma_p],
     \end{align*}
    where we used~\eqref{eq:fp step3} in the second equality, and the fact that $\vec{q}_{(j^*,a_c)}\in\Q_{j^*}$ (\cref{def:Qj}) in the third. Using the fact that the set of all predecessors of sequence $(j^*,a_c)$ is the union between all predecessors of $\sigma_p$ and $\{(j^*,a_c)\}$ itself, we can write
    \begin{align*}
        \sum_{a_r\in\A(j^*)} \vec{W}[a_r, a_c] &=\mleft(\sum_{\hat\sigma\in\Seqs_*\!\!,\,\hat\sigma\preceq\sigma_p}\lambda_{\hat\sigma}\mright)\vec{x}[\sigma_p] + \vec{x}[\sigma_p]\,\lambda_{(j^*,a_c)} + \mleft(1-\!\! \sum_{\hat\sigma\in\Seqs_*\!\!,\,\hat\sigma \preceq (j^*, a_c)}\!\! \lambda_{\hat\sigma}\mright)\,\vec{x}[\sigma_p] \\
        &= \vec{x}[\sigma_p]\,\mleft(1+\lambda_{(j^*, a_c)}  + \sum_{\hat\sigma\in\Seqs_*\!\!,\,\hat\sigma\preceq\sigma_p}\lambda_{\hat\sigma} - \sum_{\hat\sigma\in\Seqs_*\!\!,\,\hat\sigma \preceq (j^*, a_c)} \!\!\lambda_{\hat\sigma}\mright) \\&= \vec{x}[\sigma_p].
    \end{align*}
    So, all columns of the nonnegative matrix $\vec{W}$ sum to the same nonnegative quantity $\vec{x}[\sigma_p]$ and therefore $\vec{W} \in \vec{x}[\sigma_p]\cdot\bbS^{|\A(j^*)|}$. 
    
    \paragraph{Part 3: $\vec{x}'$ is a $(J\cup\{j^*\})$-partial fixed point of $\phi$} We start by arguing that $\vec{x}'$ satisfies the sequence-form constraints~\eqref{eq:partial sf} for all $j\in J\cup\{j^*\}$. The crucial observation is that \cref{algo:extend} only modifies the indices corresponding to sequences $(j^*, a)$ for $a \in \A(j^*)$ and keeps all other entries unmodified. In particular, 
    \[
        \vec{x}'[(j,a)] = \vec{x}[(j,a)]\qquad\quad\forall\ j\in J, a\in \A(j).
        \numberthis{eq:xprime x}
    \]
    Furthermore, because $J$ is a trunk for Player~$i$, the above equation implies that
    \[
        \vec{x}'[\sigma^{(i)}(j)] = \vec{x}[\sigma^{(i)}(j)] \quad\quad\forall\ j\in J.
    \]
    Hence, using the hypothesis that $\vec{x}$ is a $J$-partial fixed point of $\phi$ at the beginning of the call, we immediately conclude that~\eqref{eq:partial sf} holds for $\vec{x}'$ for all $j\in J$, and the only condition that remains to be verified is that
    \[
        \vec{x}'[\sigma_p] = \sum_{a\in\A(j^*)} \vec{x}'[(j^*, a)].\numberthis{eq:to check1}
    \]
    To verify that, observe that if $\vec{x}[\sigma_p] = 0$, then all entries $\vec{x}'[(j^*, a)]$ are set to $0$ and so~\eqref{eq:to check1} is trivially satisfied. On the other hand, if $\vec{x}[\sigma_p] \neq 0$, then $\vec{x}'[(j^*,a)] = \vec{x}[\sigma_p]\, \vec{b}[a]$, and since $\vec{b}$ belongs to the simplex $\Delta^{|\A(j^*)|}$, \eqref{eq:to check1} holds in this case too. So, $\vec{x}'$ satisfies~\eqref{eq:partial sf} for all $j \in J\cup\{j^*\}$ as we wanted to show.
    
    We now turn our attention to conditions~\eqref{eq:partial fp}. From~\cref{lem:expansion} it follows that $\phi(\vec{x})[\sigma]$ only depends on the values of $\vec{x}[(j',a')] $ for $j' \preceq \sigma, a'\in\A(j')$. So, from~\eqref{eq:xprime x} it follows that
    \[
        \phi(\vec{x}')[(j,a)] = \vec{x}[(j,a)] = \vec{x}'[(j,a)]\quad\qquad\forall\ j\in J, a\in\A(j),
    \]
    and the only condition that remains to be verified is that
    \[
        \phi(\vec{x}')[(j^*,a^*)] = \vec{x}'[(j^*, a^*)] \qquad\quad \forall\ a^*\in\A(j^*).
        \numberthis{eq:fp claim}
    \]
    Pick any $a^*\in\A(j^*)$. We break the analysis into two cases.
    \begin{itemize}
        \item If $\vec{x}[\sigma_p] = 0$, then $\vec{w} = \vec{0}$ (\cref{line:five}) and therefore $\vec{x}'[(j^*,a^*)] = 0$. Hence, to show that~\eqref{eq:fp claim} holds, we need to show that $\phi(\vec{x}')[(j^*,a^*)]=0$. To show that, we start from applying \cref{lem:expansion}:
        \begin{align*}
            \phi(\vec{x}')[(j^*,a^*)] &= \sum_{j' \preceq (j^*,a^*)}\sum_{a' \in \A(j')}\lambda_{(j',a')}\,\vec{q}_{(j',a')}[(j^*,a^*)]\, \vec{x}'[(j',a')].
        \end{align*}
        Now, using the fact that $\{j'\in\cJ: j'\preceq (j^*, a^*)\}$ is equal to the dijoint union $\{j'\in\cJ: j' \preceq \sigma_p\} \cup \{j^*\}$, and that $\vec{x}'[(j^*, a')] = 0$ for all $a'\in\A(j^*)$, we have
        \begin{align*}
            \phi(\vec{x}')[(j^*,a^*)]
                &= \sum_{j' \preceq \sigma_p}\sum_{a' \in \A(j')}\lambda_{(j',a')}\,\vec{q}_{(j',a')}[(j^*,a^*)]\, \vec{x}'[(j',a')].\numberthis{eq:stepy}
        \end{align*}
        Since $\vec{q}_{(j',a')} \in \Q_{j'}$, from \cref{def:Qj} it follows that
        \[
            \vec{q}_{(j',a')}[\sigma_p] = \sum_{a \in \A(j^*)}\vec{q}_{(j',a')}[(j^*,a)] \ge \vec{q}_{(j',a')}[(j^*,a^*)].\numberthis{eq:stepx}
        \]
        Hence, substituting \eqref{eq:stepx} into \eqref{eq:stepy},
        \begin{align*}
            \phi(\vec{x}')[(j^*,a^*)]
                &\le \sum_{j' \preceq \sigma_p}\sum_{a' \in \A(j')}\lambda_{(j',a')}\,\vec{q}_{(j',a')}[\sigma_p]\, \vec{x}'[(j',a')]\\
                &= \phi(\vec{x}')[\sigma_p] = \vec{x}'[\sigma_p] = 0,
        \end{align*}
        where the first equality follows again from \cref{lem:expansion} for all $a\in\A(j^*)$ in the first equality, and the second equality follows from the inductive hypothesis that $\vec{x}'$ is a $J$-partial fixed point of $\phi$. Since $\vec{x}'$ is a nonnegative vector and $\phi$ clearly maps nonnegative vectors to nonnegative vectors, we conclude that $\phi(\vec{x}')[(j^*, a^*)] = 0$ as we wanted to show. 
        \item If $\vec{x}[\sigma_p] \neq 0$, then $\vec{b}$ is a fixed point of the stochastic matrix $\frac{1}{\vec{x}[\sigma_p]}\vec{W}$, and therefore it satisfies
        \[
            \sum_{a_c \in \A(j^*)} \vec{W}[a^*, a_c]\, \vec{b}[a_c] = \vec{x}[\sigma_p]\,\vec{b}[a^*].
        \]
        Hence, by using the fact that $\vec{x}'[(j^*,a^*)]= \vec{x}[\sigma_p]\,\vec{b}[a^*] $ (\cref{line:last}), we can write
        \[
        \vec{x}'[(j^*,a^*)] = \sum_{a_c \in \A(j^*)} \vec{W}[a^*, a_c]\, \vec{b}[a_c].
        \]
        By expanding the definition of $\vec{W}[a^*, a_c]$ (\cref{line:W}) on the right-hand side
        \begin{align*}
            \vec{x}'[(j^*,a^*)] &=
            \sum_{a_c\in\A(j^*)}\mleft[\vec{r}[a^*] + \mleft(\lambda_{(j^*,a_c)}\vec{q}_{(j^*,a_c)}[(j^*, a^*)] + \mleft(1- \!\!\sum_{\substack{\hat\sigma\in\Seqs_*\\\hat\sigma \preceq (j^*, a_c)}} \!\!\lambda_{\hat\sigma}\mright)\,\bbone[a^* = a_c]\mright)\,\vec{x}[\sigma_p]\mright]\vec{b}[a_c]
            \\
            &=\vec{r}[a^*] +  \mleft(1- \!\!\sum_{\hat\sigma\in\Seqs_*\!\!,\,\hat\sigma \preceq (j^*, a^*)} \!\!\lambda_{\hat\sigma}\mright)\vec{x}'[(j^*,a^*)] +\!\! \sum_{a_c\in\A(j^*)}\!\!\lambda_{(j^*,a_c)}\vec{q}_{(j^*,a_c)}[(j^*, a^*)]\vec{x}'[(j^*,a_c)]
        \end{align*}
        where in the second equality we used the fact that $\vec{b}\in\Delta^{|\A(j^*)|}$, and the fact that $\vec{x}'[(j^*,a)] = \vec{x}[\sigma_p]\,\vec{b}[a]$ for all $a\in\A(j^*)$ (\cref{line:last}). Using the definition of $\vec{r}$ (\cref{line:r}),
        \begin{align*}
            \vec{x}'[(j^*,a^*)] &= \sum_{j'\preceq \sigma_p}\sum_{a'\in\A(j')} \lambda_{(j',a')}\,\vec{q}_{(j',a')}[(j^*, a^*)]\,\vec{x}[(j',a')] + \mleft(1- \!\!\sum_{\hat\sigma\in\Seqs_*\!\!,\,\hat\sigma \preceq (j^*, a^*)} \!\!\lambda_{\hat\sigma}\mright)\vec{x}'[(j^*,a^*)] \\
            &\hspace{5cm}+\!\! \sum_{a_c\in\A(j^*)}\!\!\lambda_{(j^*,a_c)}\vec{q}_{(j^*,a_c)}[(j^*, a^*)]\vec{x}'[(j^*,a_c)]\\
            &=\mleft(1-\!\!\!\!\sum_{\hat\sigma\in\Seqs_*\!\!,\,\hat\sigma\preceq(j^*,a^*)}\!\!\!\!\lambda_{\hat\sigma}\mright)\vec{x}'[(j^*,a^*)]+\sum_{j' \preceq (j^*,a^*)}\sum_{a' \in \A(j')}\lambda_{(j',a')}\,\vec{q}_{(j',a')}[\sigma]\, \vec{x}'[(j',a')]\\
            &= \phi(\vec{x}')[(j^*, a^*)], 
        \end{align*}
        where we used \cref{lem:expansion} in the last equality.
    \end{itemize}
    
    \paragraph{Part 4: Complexity analysis} In this part, we bound the number of operations required by \cref{algo:extend}.
    \begin{itemize}
    \item \cref{line:r}: each entry $\vec{r}[a]$ can be trivially computed in $O(|\Seqs|)$ time by traversing all predecessors of $j^*$. So, the vector $\vec{r}$ requires $O(|\Seqs|\,|\A(j^*)|$) operations to be computed.
    \item \cref{line:W}: if $a_r = a_c$, then the number of operations required to compute $\vec{W}[a_r,a_c]$ is dominated by the computation of $\sum_{\hat\sigma\preceq (j^*,a_c)} \lambda_{\hat\sigma}$, which requires $O(|\Seqs|)$ operations. Otherwise, if $a_r \neq a_c$, the computation of $\vec{W}[a_r,a_c]$ can be carried out in a constant number of operations. Hence, the computation of $\vec{W}[a_r,a_c]$ for all $a_r,a_c\in\A(j^*)$ requires $O(|\Seqs|\,|\A(j^*)| + |\A(j^*)|^2)$ time. Since $|\A(j^*)|\le |\Seqs|$, the total number of operations required to compute all entries of $\vec{W}$ is $O(|\Seqs|\,|\A(j^*)|)$.
    \item Lines~\ref{line:start} to~\ref{line:end}: if $\vec{x}[\sigma_p]=0$, then the computation of $\vec{w}$ requires $O(|\A(j^*)|)$ operations. If, on the other hand, $\vec{x}[\sigma_p] \neq 0$, then the computation of $\vec{w}$ requires $O(\textsf{FP}(|\A(j^*)|) + |\A(j^*)|)$ operation. Since clearly any fixed point oracle for a square matrix of order $|\A(j^*)|$ needs to spend time at least $|\A(j^*)|$ time writing the output, $O(\textsf{FP}(|\A(j^*)|) + |\A(j^*)|)=O(\textsf{FP}(|\A(j^*)|))$. So, no matter the value of $\vec{x}[\sigma_p]$, the number of iterations is bounded by $O(\textsf{FP}(|\A(j^*)|))$.
    \item \cref{line:last}: finally, the algorithm spends $O(|\A(j^*)|)$  operations to set entries of $\vec{x}$.
    \end{itemize}
    Summing the number of operations of each of the different steps of the algorithm, we conclude that each call to $\textsc{Extend}(\phi,J,j^*,\vec{x})$ requires at most $O(|\Seqs|\,|\A(j^*)| + \textsf{FP}(|\A(j^*)|))$ operations.\qedhere
\end{proof}

\begin{algorithm}[t]\caption{$\textsc{FixedPoint}(\phi)$}\label{algo:fixedpoint}
        \DontPrintSemicolon
        \SetKwInput{KwInput}{Input\hspace{2.5mm}}
        \SetKwInput{KwOutput}{Output}
        \KwInput{%
            \bull $\phi = \sum_{\hat\sigma\in\Seqs_*} \lambda_{\hat\sigma} \tdev[i][\hat\sigma][\vec{q}_{\hat\sigma}] \in \co\Ph$ transformation for a player $i\in[n]$}
        \KwOutput{\bull $\vec{q}\in\Q$~~such that $\vec{q} = \phi(\vec{q})$ }
            \BlankLine
        $\vec{q} \gets \vec{0} \in \bbR^{|\Seqs|},~~\vec{q}[\emptyseq] \gets 1$\;
        $J \gets \emptyset$\;
        \For{$j \in \cJ$ in top-down order\footnotemark}{
            $\vec{q} \gets \textsc{Extend}(\phi, J, j, \vec{q})$\;
            $J \gets J \cup \{j\}$\;
        }
        \Return{$\vec{q}$}
    \end{algorithm}
\footnotetext{That is, according to a pre-order tree traversal: if $j \prec j'$, then $j$ appears before $j'$ in the iteration order.}

A fixed point for $\phi \in \co\Ph$ can therefore be computed by repeatedly invoking $\textsc{Extend}$ to grow the trunk $J$ one information set at a time, until $J = \cJ$, starting from the $\emptyset$-partial fixed point $\vec{x}_0 \in \bbR_{\ge0}^{|\Seqs|}$ introduced in \cref{lem:x0}. This leads to \cref{algo:fixedpoint}, whose correctness and computational complexity is a straightforward corollary of \cref{prop:extend}.

\begin{corollary}\label{cor:fixed}
    Let $i\in[n]$ be a player, and let $\phi = \sum_{\hat\sigma\in\Seqs_*} \lambda_{\hat\sigma}\tdev[i][\hat\sigma][\vec{q}_{\hat\sigma}]$ be a transformation in $\co\Ph$ expressed as in~\eqref{eq:phi combination}. Then, \cref{algo:fixedpoint} computes a fixed point $\Q\ni\vec{q}=\phi(\vec{q})$ in time upper bounded as $O(|\Seqs|^2 + \sum_{j\in\cJ}\textsf{\upshape FP}(|\A(j)|))$.
\end{corollary}

\begin{algorithm}[th]\caption{$\Tph$-regret minimizer $\bar{\cR}^{(i)}$ for the set $\Q$}\label{algo:final q}
        \DontPrintSemicolon
            \KwData{\bull $i \in [n]$ player\newline
                    \bull $\tilde{\cR}^{(i)}$ regret minimizer for $\Ph$, defined in Algorithm~\ref{algo:R i}}
            \BlankLine
        \Fn{\normalfont\textsc{NextElement}()}{
            $\phi^t = \sum_{\hat\sigma\in\Seqs_*}\lambda^t_{\hat\sigma}\tdev[i][\hat\sigma][\vec{q}_{\hat\sigma}^t]\in\co\Ph \gets \tilde{\cR}^{(i)}.\textsc{NextElement}()$\;
            $\vec{q}^t\in\Q \gets \textsc{FixedPoint}(\phi^t)$\;
            \Return{$\vec{q}^t$}
        }
        \Hline{}
        \Fn{\normalfont\textsc{ObserveUtility}($\ell^t$)}{
            $L^t \gets$ linear utility function $\phi \mapsto \ell^t(\phi(\vec{q}^t))$\;
            $\tilde{\cR}^{(i)}.\textsc{ObserveUtility}(L^t)$\;
        }
    \end{algorithm}
    
\begin{algorithm}[th]\caption{$\Ph$-regret minimizer $\cR^{(i)}$ for the set $\Pure$}\label{algo:final pi}
        \DontPrintSemicolon
            \KwData{\bull $i \in [n]$ player\newline
                    \bull $\bar{\cR}^{(i)}$, $\Tph$-regret minimizer for $\Q$, defined in Algorithm~\ref{algo:final q}}
            \BlankLine
        \Fn{\normalfont\textsc{NextElement}()}{
            $\vec{q}^t \in\Q \gets \bar{\cR}^{(i)}.\textsc{NextElement}()$\;
            Sample\footnotemark\ a deterministic sequence-form strategy $\vec{\pi}^t\in\Pure$ so that it is an unbiased estimator of $\vec{q}^t$, using the natural sampling scheme described in \cref{sec:extensive form}\;
            \Return{$\vec{\pi}^t$}
        }
        \Hline{}
        \Fn{\normalfont\textsc{ObserveUtility}($\ell^t$)}{
            $\bar{\cR}^{(i)}.\textsc{ObserveUtility}(\ell^t)$\;
        }
\end{algorithm}
\footnotetext{As discussed in \cref{lem:deterministic to mixed}, in principle any unbiased sampling scheme will work. For the purposes of analyzing the complexity of \cref{algo:final pi}, however, we will assume that the natural sampling scheme described in \cref{sec:extensive form} is used. That sampling scheme runs in linear time in $|\Seqs|$.}

\subsection{The Complete Algorithm}

In this subsection we put together all the pieces we constructed in the previous subsections, in order to build the desired $\Ph$-regret minimizer that solves \cref{prob:pure}.

First, we provide in \cref{algo:final q} our $\Tph$-regret minimizer.
Its correctness follows from the correctness of the construction by~\citet{Gordon08:No} described in~\cref{sec:phirm}, and by using \cref{thm:analysis R i}~and~\cref{cor:fixed}.
Formally:

\begin{theorem}\label{thm:final mixed}
    Let $i \in [n]$ be any player. $\bar{\cR}^{(i)}$, defined in \cref{algo:final q}, is a $\Tph$-regret minimizer for the set $\Q$, whose cumulative regret upon observing linear utility functions $\ell^1,\dots,\ell^T$ satisfies 
    \[
        R^T \le 2D |\Seqs| \sqrt{T},
    \]
    where $D$ is any constant such that $\max_{\vec{q},\vec{q}'} \ell^t(\vec{q})-\ell^t(\vec{q}') \le D$ for all $t=1,\dots, T$.
    Furthermore, the \textsc{ObserveUtility} operation requires time $O(|\Seqs|^2)$, and the \textsc{NextElement} operation requires time $O(|\Seqs|^2 + \sum_{j\in\cJ}\textsf{\upshape FP}(|\A(j)|))$ at all $t$.
\end{theorem}
\begin{proof}
    From the properties of Gordon et al.'s construction~\citep{Gordon08:No} (\cref{sec:phirm}), the cumulative $\Ph$-regret incurred by $\bar{\cR}^{(i)}$ is equal, at all times, to the cumulative regret incurred by the underlying regret minimizer $\tilde{\cR}^{(i)}$ for the set of deviations $\Ph$. So, the regret bound follows from the regret analysis of \cref{thm:analysis R i}. 
    
    Similarly, the complexity analysis follows from combining the analysis of $\tilde{\cR}^{(i)}$ and of \textsc{FixedPoint} (\cref{algo:fixedpoint}), together with the observation that the canonical representation $\langle L^t \rangle$ of the linear utility function $\co\Ph\ni\phi\mapsto\ell^t(\phi(\vec{q}^t))$ is the matrix $\langle\ell^t\rangle (\vec{q}^t)^\top$, which can be trivially computed in $O(|\Seqs|^2)$ time. 
\end{proof}

Since $\co\Ph \supseteq \Ph$, \cref{algo:final q} is in particular also a $\Ph$-regret minimizer for the set $\Q$. So, \cref{thm:final mixed} establishes that \cref{algo:final q} provides a solution to \cref{prob:mixed}. Consequently, by applying \cref{lem:deterministic to mixed}, we immediately get the following characterization of \cref{algo:final pi}, our $\Ph$-regret minimizer $\cR^{(i)}$ for the set of deterministic sequence-form strategies $\Pure$ of Player~$i$.

\begin{corollary}\label{thm:final}
    Let $i \in [n]$ be any player. $\cR^{(i)}$, defined in \cref{algo:final pi}, is a $\Ph$-regret minimizer for the set $\Pure$, whose cumulative regret $R^T$ upon observing linear utility functions $\ell^1,\dots,\ell^T$ satisfies 
    \[
        R^T \le 2D |\Seqs| \sqrt{T} + D \sqrt{8T\log(1/\delta)}\text{ with probability at least } 1-\delta,
    \]
    where $D$ is any constant such that $\max_{\vec{q},\vec{q}'}\{ \ell^t(\vec{q})-\ell^t(\vec{q}')\} \le D$ for all $t=1,\dots, T$. Furthermore, the \textsc{ObserveUtility} operation requires time $O(|\Seqs|^2)$, and the \textsc{NextElement} operation requires time $O(|\Seqs|^2 + \sum_{j\in\cJ}\textsf{\upshape FP}(|\A(j)|))$ at all $t$. 
\end{corollary}

Therefore, \cref{algo:final pi} is a solution to \cref{prob:pure}.

\section{Convergence to EFCE}

\cref{thm:empirical efce} implies that if all players $i\in[n]$ play the game repeatedly according  to the outputs of a $\Ph$-regret minimizer for $\Pure$ that observes, at each time $t$, the linear utility function given in~\eqref{eq:def loss}, then the empirical frequency of play is a $(\frac{1}{T}\max_i R^{(i),T})$-EFCE, where $R^{(i),T}$ is the regret cumulated by the $\Ph$-regret minimizer for Player~$i$. In particular, when all players play according to the strategies recommended by \cref{algo:final pi}, the following can be shown. 
\begin{theorem}\label{thm:conv finito}
    When all players $i = 1, \dots, n$ play according to the outputs of the regret minimizer $\cR^{(i)}$ defined in \cref{algo:final pi}, receiving as feedback at all times $t$ the linear utility functions $\ell^{(i),t}$ defined in~\eqref{eq:def loss}, the empirical frequency of play after $T$ repetitions of the game is a
    \[
        \mleft(D\frac{2|\Hist| + \sqrt{8\log(n/\delta)}}{\sqrt T}\mright)\text{-EFCE with probability at least } 1-\delta,
    \]
    where $D$ is the difference between the maximum and minimum payoff of the game, and $|\Hist|$ is the number of nodes in the game tree.
\end{theorem}
\begin{proof}
    Let $R^{(i),T}$ be the regret cumulated by Player~$i$ up to time $T$ when playing according to \cref{algo:final pi}. From \cref{thm:final}, we have that for all $\delta'\in (0,1)$,
    \[
        \bbP\mleft[R^{(i),T} \le 2D |\Hist| \sqrt{T} + D \sqrt{8T\log(1/\delta')}\mright] \ge \bbP\mleft[R^{(i),T} \le 2D |\Seqs| \sqrt{T} + D \sqrt{8T\log(1/\delta')}\mright] \ge 1-\delta',
    \]
    where the first inequality follows from the fact that $|\Seqs| = \sum_{j\in\cJ}|\A(j)| \le \sum_{h\in\Hist} |\A(h)| \le |\Hist|$ (the number of edges in a tree is always less than the number of nodes).
    So,
    \begin{align*}
        \bbP\mleft[\max_i R^{(i),T} \le 2D |\Hist| \sqrt{T} + D \sqrt{8T\log(1/\delta')}\mright]
        &= \bbP\mleft[\bigcap_i \mleft\{R^{(i),T} \le 2D |\Hist| \sqrt{T} + D \sqrt{8T\log(1/\delta')} \mright\} \mright] \\&\ge 1-n\delta',
    \end{align*}
    where the inequality follows from the union bound. Substituting $\delta\defeq n\delta'$ and using \cref{thm:empirical efce} yields the result.
\end{proof}

A standard application of the Borel-Cantelli lemma enables one to move from the high-probability guarantees at finite time of \cref{thm:conv finito} to almost-sure guarantees in the limit.

\begin{corollary}
    When all players $i = 1, \dots, n$ play infinitely many repetitions of the game according to the outputs of the regret minimizer $\cR^{(i)}$ defined in \cref{algo:final pi}, receiving as feedback at all times $t$ the linear utility functions $\ell^{(i),t}$ defined in~\eqref{eq:def loss}, the empirical frequency of play converges, almost surely, to an EFCE.
\end{corollary}